 \documentclass[11pt]{article}

\usepackage{fullpage}

\RequirePackage{natbib}

\usepackage{amsthm,amsmath,amssymb}
\usepackage{ifthen}
\usepackage{graphicx}
\usepackage{hyperref}
\usepackage{cleveref}
\usepackage{subcaption}
\usepackage{diagbox}

\newtheorem{theorem}{Theorem}[section]
\newtheorem{definition}{Definition}

\newtheorem{lemma}[theorem]{Lemma}
\newtheorem{fact}[theorem]{Fact}
\newtheorem{proposition}[theorem]{Proposition}
\newtheorem{corollary}[theorem]{Corollary}
\newtheorem{example}{Example}

\newenvironment{proofof}[1]{\vspace{0.1in}\noindent{\sc Proof of #1.}}{\hfill\qed}

\newenvironment{numberedtheorem}[1]{%
\renewcommand{\thetheorem}{#1}%
\begin{theorem}}{\end{theorem}\addtocounter{theorem}{-1}}

\newenvironment{numberedlemma}[1]{%
\renewcommand{\thetheorem}{#1}%
\begin{lemma}}{\end{lemma}\addtocounter{theorem}{-1}}

\newenvironment{numberedcorollary}[1]{%
\renewcommand{\thetheorem}{#1}%
\begin{corollary}}{\end{corollary}\addtocounter{theorem}{-1}}

\newcommand{\prob}[2][]{\text{\bf Pr}\ifthenelse{\not\equal{}{#1}}{_{#1}}{}\!\left[#2\right]}
\newcommand{\expect}[2][]{\text{\bf E}\ifthenelse{\not\equal{}{#1}}{_{#1}}{}\!\left[#2\right]}
\newcommand{\given}{\,\middle|\,}

\newcommand{\yestag}{\addtocounter{equation}{1}\tag{\theequation}}

\newcommand{\agind}[1][i]{_{#1}}

\newcommand{\inversed}[1]{#1^{-1}}
\newcommand{\ironed}{\bar}

\newcommand{\differentiated}[1]{#1'}
\newcommand{\fortype}{\tilde}
\newcommand{\estimated}{\hat}
\newcommand{\sampled}{\hat}

\newcommand{\noaccents}[1]{#1}
\newcommand{\composed}[3]{#1{#2{#3}}}

\newcommand{\newindexedvar}[4][\noaccents]{%
\expandafter\newcommand\expandafter{\csname #2\endcsname}{#1{#4}}%
\expandafter\newcommand\expandafter{\csname #2s\endcsname}{#1{\boldsymbol{#4}}}%
\expandafter\newcommand\expandafter{\csname #2sm#3\endcsname}[1][#3]{#1{\boldsymbol{#4}}_{-##1}}%
\expandafter\newcommand\expandafter{\csname #2#3\endcsname}[1][#3]{#1{#4}\agind[##1]}%
\expandafter\newcommand\expandafter{\csname #2#3th\endcsname}[1][#3]{#1{#4}_{(##1)}}%
}

\newcommand{\rev}{R}
\newcommand{\irev}{\bar{R}}
\newcommand{\marg}{R'}
\newcommand{\knalloc}[1][k]{x_{#1:n}}
\newcommand{\kalloc}[1][k]{x_{#1}}
\newcommand{\quant}{q}
\newcommand{\equant}{\hat{\quant}}

\newindexedvar{wal}{k}{w}
\newindexedvar[\differentiated]{margwal}{k}{\wal}
\newindexedvar{cumwal}{k}{W}

\newindexedvar{eval}{k}{V} 
\newcommand{\expval}{\mathcal{V}} 
\newcommand{\Zbar}{\bar{Z}} 
\newcommand{\sw}[1][\wals]{S_{#1}} 

\newindexedvar[\ironed]{iwal}{k}{\wal}
\newindexedvar[\ironed]{cumiwal}{k}{\cumwal}
\newindexedvar[\composed{\differentiated}{\ironed}]{margiwal}{k}{\wal}

\newindexedvar{yal}{k}{y}
\newindexedvar{cumyal}{k}{Y}
\newindexedvar[\ironed]{iyal}{k}{\yal}
\newindexedvar[\composed{\differentiated}{\ironed}]{margiyal}{k}{\yal}

\newindexedvar{newyal}{k}{\tilde{y}}
\newindexedvar{cumnewyal}{k}{\tilde{Y}}

\newindexedvar{epswal}{k}{\epsilon}
\newindexedvar{cumepswal}{k}{E}
\newindexedvar[\differentiated]{margepswal}{k}{\epswal}

\newindexedvar{murev}{k}{P}
\newindexedvar[\estimated]{emurev}{k}{\murev}
\newindexedvar[\differentiated]{mumarg}{k}{\murev}

\newindexedvar[\ironed]{imurev}{k}{\murev}
\newindexedvar[\composed{\differentiated}{\ironed}]{imumarg}{k}{\murev}

\newcommand{\bhat}{\hat{b}}

\newindexedvar{val}{i}{v}

\newindexedvar{eff}{i}{e}

\newindexedvar{bid}{i}{b}
\newindexedvar[\estimated]{ebid}{i}{\bid}
\newindexedvar[\sampled]{sbid}{i}{\bid}

\newindexedvar{bidap}{i}{\bid}
\newindexedvar[\estimated]{ebidap}{i}{\bid}
\newindexedvar{bidfp}{i}{c}
\newindexedvar[\estimated]{ebidfp}{i}{\bidfp}

\newindexedvar{strat}{i}{s}

\newindexedvar{virt}{i}{\phi}
\newindexedvar[\inversed]{virtinv}{i}{\virt}

\newindexedvar[\ironed]{ivirt}{i}{\virt}

\newindexedvar{dist}{i}{F}

\newindexedvar{dens}{i}{f}

\newindexedvar{price}{i}{p}
\newindexedvar[\fortype]{tprice}{i}{p}

\newindexedvar{alloc}{i}{x}

\newindexedvar[\fortype]{talloc}{i}{\alloc}

\newindexedvar{util}{i}{u}
\newindexedvar[\fortype]{tutil}{i}{u}

\newcommand{\REV}[1]{P_{#1}}

\newcommand{\samples}{N}
\newcommand{\eps}{\epsilon}

\newcommand{\opt}[2]{\mathbf{OPT}(#1,#2)}

\newcommand{\T}{T}
\newcommand{\Tij}{T_{i,j}}
\newcommand{\bwhat}{\widehat{b}}
\newcommand{\btild}{\tilde{b}}
\newcommand{\Phat}{\hat{P}}

\newcommand{\rxy}[1][y]{\Phi_{\alloc,#1}}
\newcommand{\Err}[1]{\expect[\hat{\bid}]{|\hat{#1}-#1|}}

\newcommand{\lastcrossing}{m}

\newcommand{\ceil}[1]{\left\lceil#1\right\rceil}
\newcommand{\floor}[1]{\left\lfloor#1\right\rfloor}

\renewcommand{\qedhere}{}

\begin{document}

\title{Mechanism Redesign\footnote{This paper combines and improves
    upon results from manuscripts {\em Mechanism Design for Data
      Science} \citeyearpar{CHN-14} and {\em A/B Testing of Auctions} \citeyearpar{CHN-16}.}}

\author{Shuchi Chawla\thanks{University of Wisconsin - Madison, WI. This work was supported by NSF CCF 1101429.} \and
Jason D. Hartline\thanks{Northwestern University, Evanston, IL.  This work was supported by NSF CCF 1101717.} 
\and Denis Nekipelov\thanks{University of Virginia, Charlottesville, VA. This work was supported by NSF CCF 1101706.}
\and Anant Shah\thanks{Northwestern University, Evanston, IL.}}

\date{}

\maketitle

\begin{abstract}
This paper develops the theory of {\em mechanism redesign} by which an
auctioneer can reoptimize an auction based on bid data collected from
previous iterations of the auction on bidders from the same market.
We give a direct method for estimation of the revenue of a
counterfactual auction from the bids in the current auction.  The
estimator is a simple weighted order statistic of the bids and has the
optimal error rate.  Two applications of our estimator are {\em A/B
  testing} (a.k.a., randomized controlled trials) and {\em instrumented
  optimization} (i.e., revenue optimization subject to being able to
do accurate inference of any counterfactual auction revenue).
\end{abstract}


\section{Introduction}
\label{s:intro}

%
%
This paper develops data-driven methods that enable a principal to
adjust the parameters of an auction so as to optimize its performance,
i.e., for {\em mechanism redesign}.  These methods are inspired by
settings of online markets such as online advertising, hotel booking
platforms, online auctions, etc.\footnote{We have had extensive
  discussions over the last decade with R\&D teams at companies in
  this space, especially Microsoft, which brought us to the model and
  questions studied in this paper.}  For a paradigmatic family of
auctions, we derive estimators for the revenue and welfare of
a counterfactual auction in the family from equilibrium bids in an
incumbent auction.  Our analysis exposes the relationship between the
the error of the estimator and the tuned parameters of the
counterfactual and incumbent auctions.  From this analysis, we
identify a family of mechanisms that are statistically efficient to
redesign.

Our work is motivated by the simple observation that revenue
maximization in auctions is at odds with statistical inference.
Specifically, in the classic setting of \citet{mye-81}, an auctioneer
who adopts a mechanism from the family of revenue optimal auctions,
which employ reserve prices, will not generally be able to
infer distribution of values of the bidders in the auction and, thus,
will not generally be able to determine the revenue optimal auction
from the family.  Reserve pricing and ironing pool bidders with
distinct values and, thereafter, no procedure for structural inference
can distinguish them.  Consequently, counterfactual revenue estimation
from bids in revenue optimal auctions is not generally possible.

Our work focuses on a canonical family of mechanisms that generalize
multi-unit auctions.  Position auctions where introduced and studied
by \citet{var-06} and \citet{EOS-07} as a model of auctions for
selling advertisements on Internet search engines.  A position auction
is defined by a decreasing sequence of weights which correspond to
allocation probabilities, bidders are assigned to weights
assortatively by bid.  The configurable parameters in this family of
auctions are the weights of the positions.  Position auctions
generalize classical single-item and multi-unit auction models and are
an important form of auction for theoretical study.  We study the
winner-pays-bid and all-pay variants of position
auctions. The main results for the all-pay setting are stated in \Cref{s:inference} and \Cref{s:iron-rank} while the winner-pays-bid variant is considered in \Cref{s:fp-inf}  \footnote{Unfortunately, our methods cannot be directly
applied to position auctions with the so-called ``generalized second
price'' payment rule of Google's AdWords platform.  Online ad auctions
are, however, moving towards winner-pays-bid formats like the ones we
consider \citep[e.g.,][]{PST-20}.}

Position auctions (without reserve prices or ironing) are not revenue
optimal; however, we show that they can be tuned so that their revenue
is approximately optimal.  For example, in a multi-unit auction for
bidders with identically distributed values it is revenue optimal to
impose a reserve price; however this optimal revenue is closely
approximated by reducing supply instead (e.g., to sell the same number
of items in expectation as under the reserve price).  Moreover, in
followup work to ours \citet{HT-19} show that the linking-decisions
approach of \citet{JS-07} can be applied to repeated Bayes-Nash
auctions with non-identically distributed bidders to give position
auctions that approximate the revenue optimal auction arbitrarily
closely.

Our main technical contribution is a method for counterfactual revenue
estimation: given two position auctions we define an estimator for the
equilibrium revenue of one (the counterfactual auction) from
equilibrium bids of the other (the incumbent auction).  Our estimator
has three important properties that contrast with the classical
approach to counterfactual revenue estimation in auctions.
\begin{enumerate}
\item The estimator is a linear functional of the bid distribution and
  can be applied to any estimator for the bid distribution.

  In contrast, the classical approach requires the estimator for the
  bid distribution be continuously differentiable so that the bidders'
  first order condition can be inverted to infer values.

  In contrast, the classical analysis requires uniform estimation of
  the distribution of values to obtain error bounds on counterfactual
  revenue estimates derived from value estimates.
  
\item Applied to the empirical bid distribution, the estimator is a
  weighted order statistic.  Specifically the order statistics of the
  bids in the incumbent mechanisms are mapped to a revenue estimate as
  a simple weighted sum with weights that are determined by the
  parameters of the counterfactual and incumbent
  auctions.\footnote{For numerical stability this estimator should be
  computed in its algebraically equivalent form as a weighted sum of
  the difference of consecutive bids.}
  
\item The estimator converges at the rate equal to the square root of the number of samples with leading coefficients
  that depend directly on the relationship between the tuned
  parameters of the counterfactual and incumbent auction; thus, we are
  able to identify mechanisms that are good for redesign, i.e., for
  which the estimator has small leading coeficients.
\end{enumerate}

Our theoretical analysis suggests that the typical smoothing in the
classical analysis is not necessary.  We confirm suggestion with Monte
Carlo simulations that show that the empirically optimal smoothing
for counterfactual revenue estimation is, in fact, no smoothing.  Our
theoretical analysis also suggests that error can be reduced by
truncating, i.e., zeroing out, the contribution to the estimator from
extreme (low or high) bids.  We find empirically that the level of
truncation suggested by the analysis, which does not depend on
distributions or mechanisms, is indeed close to the correct level.  



Our detailed analysis of the error of the counterfactual revenue
estimation, in terms of the parameters of the incumbent and
counterfactual mechanisms allows for direct comparison to the industry
standard practice of A/B testing, a.k.a., randomized controled trials
\citep[e.g.,][]{KLSH-09}.  An ideal A/B test for auctions would work
as follows: The auction format is randomized between formats A and B
with equilibrium bids for A collected for format A and equilibrium
bids for B collected for format B.  Bids in A are used to estimate
revenue of A and bids in B are used to estimate revenue in B.  This
ideal is typically not achieved in practice in online markets, often
instead the bids are collected in advance of the realization of the
randomized format, and thus can be assumed to be in equilibrium of
format C given by the convex combination of the two
formats.\footnote{Consider the example of online advertizing market.
  The advertisers bid in advance; the users to whom the ads are shown
  arrive online; and the randomization occurs at the user level.
  Typically these A/B tests are to understand user behavior; however,
  development teams that employ A/B testing also track ``key
  performance indicator'' like revenue.  Implications for revenue for
  A and B, of course, cannot be inferred from averages across the A
  and B treatments, since bids are in equilibrium for C not A nor B.}
Our analysis allows the revenue for B (also A) to be estimated from
the equilibrum bids in C and our error bounds depend on the
relationship between B and C.  When the A/B test mixes in B with
probability $\epsilon$ (and A with probability $1-\epsilon$), our
estimators dependence on epsilon is $\log 1/\epsilon$ while the ideal
A/B test has dependence $\sqrt{1/\epsilon}$.

The main result of the paper is an analysis of the instrumented
optimization problem where the principal aims to run a mechanism that
is both good for revenue and good for subsequent inference.  As a
first result, we show that there is universal treatment B such that in
the A/B test mechanism C, the revenue of any other position auction
can be estimated with low error.  A heuristic solution to the
instrumented optimization problem, then, is to run the A/B test
mechanism C that corresponds to the revenue optimal position auction A
and this universal treatment B.  Our second result incorporates a
bound on the desired rate of estimation into the revenue optimization
problem and derives the revenue optimal auction subject to good
revenue inference.  Our analysis gives a tradeoff between revenue
bounds (relative to the optimal position auction) and the desired rate
of convergence.

Our bounds on the error of our estimator are expressed in terms of the
the number of samples (from the bid distribution, subsequently denoted
by $N$), the number of bidders (in each auction, subsequently denoted
by $n$, and the similarity between the incumbant and counterfactual
auction by $\epsilon$).  A straightforward econometric analysis might
treat the number of bidders and parameters of the auctions as
constants without quantifying the detailed dependence of the error on
these constants.  Such analysis does not preclude the possibility that
there is a very large error until the number of samples is
exponentially larger than this number of bidders.\footnote{For
  example, in a single-unit $n$-agent all-pay auction the bids of the
  median bidder are an exponential factor smaller than their value.
  Of course, the revenue of the $n/2$-unit all-pay auction depends
  very much on good estimates on the contribution to the revenue from
  the median bidder.}  In contrast, our error bounds show at most
polynomial dependence on the number of bidders and logarithmic
dependence on the similarity between auctions $1/\epsilon$.

While much of this paper focuses on estimating and optimizing revenue,
in \Cref{s:welfare} we extend the method to the estimation of
social welfare.

\subsection{Motivating Example: Auctions for Internet Search Advertising.}
\label{s:example-search-advertising}

Our work is motivated by the auctions that sell advertisements on
Internet search engines (see historical discussion by
\citealp{FP-06}).  The first-price position auction we study in this
paper was introduced in 1998 by the Internet listing company GoTo.
This auction was adapted by Google in 2002 for their AdWords platform,
modified to have a second-price-like payment rule, and is known as the
{\em generalized second-price auction}.  Early theoretical studies of
equilibrium in the generalized second-price auction were conducted by
\citet{var-06} and \citet{EOS-07}; unlike the second-price auction for
which it is named, the generalized second-price auction does not admit
a truthtelling equilibrium.

Internet search advertising works as follows.  A user looking for web
pages on a given topic specifies {\em keywords} on which to search.  The
search engine then returns a listing of web pages that relate to these
keywords.  Alongside these search results, the search engine displays
sponsored results.  These results are conventionally explicitly
labeled as sponsored and appear in the {\em mainline}, i.e., above the
search results, or in the {\em sidebar}, i.e., to the right of the
search results.  The mainline typically contains up to four ads and
the sidebar contains up to seven ads.  The order of the ads is of
importance as the Internet user is more likely to read and click on
ads in higher positions on the page.  In the classic model of
\citet{var-06} and \citet{EOS-07} the user's click behavior is
exogenously given by weights associated with the
positions,\footnote{Endogenous click models have also been considered,
  e.g., \citet{AE-11}, but are less prevalent in the literature.}  and
the weights are decreasing in position.  An advertiser only pays for
the ad if the user clicks on it.  Thus in the classic first-price
position auction, advertisers are assigned to positions in order of
their bids, and the advertisers on whose ads the user clicks each pay
their bids.

As described above, the ads in the mainline have higher click rates
than those in the sidebar.  The mainline, however, is not required
to be filled to capacity (a maximum of four ads).  In the first-price
position auction described above, the choice of the number of ads to
show in the mainline affects the revenue of the auction and, in the
standard auction-theoretic model of Bayes-Nash equilibrium, this
choice is ambiguous with respect to revenue ranking.  For some
distributions of advertiser values, showing more ads in the mainline
gives more revenue, while for other distributions fewer ads gives more
revenue.

The keywords of the user enable the advertisers to target users with
distinct interests.  For example, hotels in Hawaii may wish to show
ads to a user searching for ``best beaches,'' while a computer
hardware company would prefer users searching for ``laptop reviews.''
Thus, the search advertising is in fact a collection of many partially
overlapping markets, with some high-volume high-demand keywords and a
long tail of niche keywords.  The conditions of each of these markets
are distinct and thus, as per the discussion of the preceding
paragraph, the number of ads to show in the mainline depends on the
keywords of the search.

One empirical method for evaluating two alternatives, e.g., showing
one or two mainline ads, is A/B testing.
In the ideal setting of A/B
testing, the auctions for a given keyword would be randomly divided
into the A and B groups. In part A the advertisers would bid in
Bayes-Nash equilibrium for A and in part B they would bid in
equilibrium for B.  Unfortunately, because we need to test both A and
B in each market, ideal A/B testing would require soliciting distinct
bids for each variant of the auction.  This approach is impractical,
both from an engineering perspective and from a public relations
perspective.  In practice, A/B tests are run on these ad platforms all
the time and without informing the advertisers.  Of course,
advertisers can observe any overall change in the mechanism and adapt
their bids accordingly, i.e., they can be assumed to be in
equilibrium.  Our approach of A/B testing, where bids are in
equilibrium for auction C, the convex combination of A and B, is
consistent with the industry standard practice for Internet search
advertising.

Our A/B testing framework is motivated specifically by the goal of
optimizing an auction to local characteristics of the market in which
the auction is run.  It is important to distinguish this goal from
that of another framework for A/B testing which is commonly used to
evaluate global properties of auctions across a collection of disjoint
markets.  This framework randomly partitions the individual markets
into a control group (where auction A is run) and a treatment group
(with auction B).  From such an A/B test we can evaluate whether it is
better for all markets to run A or for all to run B.  It cannot be
used, however, for our motivating application of determining the
number of mainline ads to show, where the optimal number naturally
varies across distinct markets.  The work of \citet{OS-11} on reserve pricing
in Yahoo!'s ad auction demonstrates how such a global A/B test can be
valuable.  They first used a parametric estimator for the value
distribution in each market to determine a good reserve price for that
market.  Then they did a global A/B test to determine whether the
auction with their calculated reserve prices (the B mechanism) has
higher revenue on average than the auction with the original reserve
prices (the A mechanism).  Our methods relate to and can replace the
first step of their analysis.

\subsection{Related Work}

Our work is motivated, in part, by field work in the past decade that
considers the empirical optimization of reserve prices in auctions
(e.g., \citealp{Ril-06}; \citealp{BM-09}; and \citealp{OS-11}).  The
field study of \citeauthor{OS-11} is most similar to our theoretical
study and the motivating example of
\Cref{s:example-search-advertising}.  They consider the
generalized second-price position auction of Internet search
advertising (on Yahoo!).  They assume that the distribution of
advertiser (bidder) values is lognormal and use structural inference
to estimate the parameters of the distribution for the keywords of
each search.  This allows for inference of the optimal reserve price.
They then suggest using a reserve price that is slightly smaller.
Finally, they evaluate the method of setting reserves via a global A/B
test that compairs the original reserves with their reserves across
all keywords.  While the authors motivate the usage of reserves
slightly smaller than the optimal reserves for reasons of robustness,
in the context of our motivation these smaller reserves also allow
future inference around the optimal reserve price where the optimal
reserves do not.


The classical approach to counterfactual inference is based on
recovering the values of bidders by inverting their best responses
using the empirical distribution of bids.  This approach was developed
by \citet{guerre} for single-unit first-price auctions and it has seen
application broadly in auction theory (e.g.\@ see
\citealp{athey:2007}, \citealp{paarsch}, and \citealp{marmer}). There
are several ways in which our revenue estimator improves upon this
standard approach. First, \citet{guerre} and subsequent works assume
that values lie in a bounded range, the value density is bounded away
from zero, and the density is differentiable. We also assume that
values are bounded, but do not require any other assumptions on the
value density. Second, this literature focuses on single item first
price auctions. In contrast, our work applies to first-price and
all-pay $k$-unit auctions, as well as mixtures over them. On the other
hand, while \citet{marmer} allow for the seller to impose a reserve
price, our techniques do not extend to auctions with reserve
prices. However, we show that auctions with reserve prices can be
approximated in revenue by position auctions. Third, the classical
approach requires selecting an appropriate bandwidth for bid density
estimation that is tuned to properties of the endogenous bid
distribution, whereas our estimator is not parameterized.  Our
inference method is based on a technique similar to that in
\citet{mye-81} that ``integrates out'' best responses of agents so
that the auction revenue can be expressed directly in terms of the
observable bids.  Our proposed estimator is equivalent to the plug-in
estimator with no smoothing and we empirically show that estimation
error as a function of the degree of smoothing is minimized with no
smoothing.

%
%
Our work focuses on position auctions with agents with identically
distributed values, i.e., symmeteric position auctions.  Recently,
\citet{HT-19} generalize our results to single-dimensional mechanism
design problems that are repeated with asymmetric agents within rounds
but identically distributed agents across rounds, i.e., a repeated
agent-normal-form game.  Their approach employs our analysis as a
black box and demonstrates that symmetric position auctions are
fundamental to the theory of mechanism design for single-dimensional
agents more generally.

Our estimator of the counterfactual auction revenue is simply a
weighted order statistic of samples from the bid distribution.  Other
works have proposed using similarly simple estimators to obtain bounds
on the performance a counterfactual auctions.  Unlike our work, these
estimators do not use the first-order condition of Bayes-Nash
equilbrium, but in exchange for a weakening of the assumptions of the
model, they obtain bounds instead of point estimates.  For example,
\citet{coey:2014} consider ascending single-item auctions and use the
main theorem of \citet{BK-96}, the revenue submodularity of \citet{DRS-12},
and the expected second- and third-highest bids to bound the revenue
of the (counterfactual) optimal auction.  See their related work
section for a discussion of similar studies.


The mechanism design literature has previously considered the problem
of an uninformed designer who wishes optimize a mechanism under three
conditions: (a) repeatedly on agents from the same population (each
agent participates only once; see \citealp{KL-03}; \citealp{BH-05};
and \citealp{CGM-15}), (b) with samples from the value distribution (see
\citealp{CR-14}, and \citealp{FHHK-14}), and (c) on the fly in one
mechanism (see \citealp{GHKSW-06}; \citealp{seg-03}; and
\citealp{BV-03}).  These works exclusively consider mechanisms that
have truthtelling equilibria and for which, consequently, inference is
trivial.  The papers listed in category (a) also consider a model
where the designer only learns the revenue of the mechanism in each
round and not the individual bids.  These papers adapt methods from
the multi-armed bandit literature, e.g., \citet{ACFS-02}, which tradeoff
exploring the performance of mechanisms that the designer is less
informed about with exploiting the mechanisms which have been learned
to perform well.  Our approach of instrumented optimization is similar
to the exploration steps of these multi-armed bandit algorithms,
except that we assume that bids are in equilibrium for the
distribution over mechanisms rather than for each individual
mechanism.  This distinction is important for mechanisms that do not
have truthtelling equilibria.

Finally, the theory that we develop for optimizing revenue over the
class of rank-by-bid position auctions is isomorphic to the theory of
envy-free optimal pricing developed by \citet{DHY-14}.

\section{Preliminaries}\label{s:prelim}
\label{S:PRELIM}

%
%
%
%
%

\subsection{Auction Theory}

A standard auction design problem is defined by a set $[n] =
\{1,\ldots,n\}$ of $n\ge 2$ agents, each with a private value $\vali$
for receiving a service.  The values are bounded as $\vali\in[0,1]$
and are independently and identically distributed according to a
continuous distribution $\dist$.  If $\alloci$ indicates the
probability of service and $\pricei$ the expected payment required,
agent $i$ has linear utility $\utili = \vali \alloci - \pricei$.  An
auction elicits bids $\bids = (\bidi[1],\ldots,\bidi[n])$ from the
agents and maps the vector $\bids$ of bids to an allocation
$\tallocs(\bids) = (\talloci[1](\bids),\ldots,\talloci[n](\bids))$,
specifying the probability with which each agent is served, and prices
$\tprices(\bids) = (\tpricei[1](\bids),\ldots,\tpricei[n](\bids))$,
specifying the expected amount that each agent is required to pay.  An
auction is denoted by $(\tallocs,\tprices)$.  

\paragraph{Standard payment formats}

In this paper we study two standard payment formats. In a {\em
  first-price} format, each agent pays his bid upon winning, that is,
$\tpricei(\bids) = \bidi \, \talloci(\bids)$. In an {\em all-pay} format,
each agent pays his bid regardless of whether or not he wins, that is,
$\tpricei(\bids) = \bidi$.


\paragraph{Bayes-Nash equilibrium}
The values are independently and identically distributed according to
a continuous distribution $\dist$.  This distribution is common
knowledge to the agents.  A strategy $\strati$ for agent $i$ is a
function that maps the
value of the agent to a bid.  The distribution of values $\dist$ and a
profile of strategies $\strats = (\strati[1],\cdots,\strati[n])$
induces interim allocation and payment rules (as a function of bids)
as follows for agent $i$ with bid $\bidi$.
\begin{align*}
\talloci(\bidi) &= \expect[\valsmi \sim \dist]{\talloci(\bidi,\stratsmi(\valsmi))} \text{ and}\\
\tpricei(\bidi) &= \expect[\valsmi \sim \dist]{\tpricei(\bidi,\stratsmi(\valsmi))}.
\intertext{Agents have linear utility which can be expressed in the interm as:}
\tutili(\vali,\bidi) &= \vali\talloci(\bidi) - \tpricei(\bidi).
\end{align*}
The strategy profile forms a {\em Bayes-Nash equilibrium} (BNE) if for
all agents $i$, values $\vali$, and alternative bids $\bidi$, bidding
$\strati(\vali)$ according to the strategy profile is at least as good
as bidding $\bidi$.  I.e.,
\begin{align}
\label{eq:br}
\tutili(\vali,\strati(\vali)) &\geq \tutili(\vali,\bidi).
\end{align}

A symmetric equilibrium is one where all agents bid by the same
strategy, i.e., $\strats$ statisfies $\strati = \strat$ for all $i$
and some $\strat$.  For a symmetric equilibrium of a symmetric
auction, the interim allocation and payment rules are also symmetric,
i.e., $\talloci = \talloc$ and $\strati = \strat$ for all $i$.  For
implicit distribution $\dist$ and symmetric equilibrium given by
stratey $\strat$, a mechanism can be described by the pair
$(\talloc,\tprice)$.  \citet{CH-13} show that the equilibrium of every
auction in the class we consider is unique and symmetric.

The strategy profile allows the mechanism's outcome rules to be
expressed in terms of the agents' values instead of their bids; the
distribution of values allows them to be expressed in terms of the
agents' values relative to the distribution.  This latter
representation exposes the geometry of the mechanism.  Define the {\em
  quantile} $\quant$ of an agent with value $\val$ to be the
probability that $\val$ is larger than a random draw from the
distribution $\dist$, i.e., $\quant=\dist(\val)$.  Denote the agent's
value as a function of quantile as $\val(\quant) =
\dist^{-1}(\quant)$, and his bid as a function of quantile as
$\bid(\quant) = \strat(\val(\quant))$.  The outcome rule of the
mechanism in quantile space is the pair
$(\alloc(\quant),\price(\quant)) =
(\talloc(\bid(\quant)),\tprice(\bid(\quant)))$.

\paragraph{Revenue curves and auction revenue}

\citet{mye-81} characterized Bayes-Nash equilibria and this
characterization enables writing the revenue of a mechanism as a
weighted sum of revenues of single-agent posted pricings.  Formally,
the {\em revenue curve} $\rev(\quant)$ for a given value distribution
specifies the revenue of the single-agent mechanism that serves an
agent with value drawn from that distribution if and only if the
agent's quantile exceeds $\quant$: $\rev(\quant) =
\val(\quant)\,(1-\quant)$.  Myerson's characterization of BNE then
implies that the expected revenue of a mechanism at BNE from an agent
facing an allocation rule $\alloc(\quant)$, notated $\REV{\alloc}$,
can be written as follows:
\begin{eqnarray}
  \label{eq:bne-rev}
  \REV{\alloc} = \rev(0)\,\alloc(0) + \expect[\quant]{\rev(\quant)\,\alloc'(\quant)}  = \rev(1)\,\alloc(1) - \expect[\quant]{\rev'(\quant)\,\alloc(\quant)} 
\end{eqnarray}
where $\alloc'$ and $\rev'$ denote the derivative of $\alloc$ and
$\rev$ with respect to $\quant$, respectively.  For value
distributions supported on $[0,1]$, $\rev(0) = \rev(1) = 0$ and the
constant terms in equation~\eqref{eq:bne-rev} are identically zero.

The expected revenue of an auction is the sum over the agents of its
per-agent expected revenue; for auctions with symmetric equilibrium
allocation rule $\alloc$ this revenue is $n \, \REV{\alloc}$.

\paragraph{Position environments and rank-based auctions}
\label{sec:rank-based-basics}

A {\em position environment} expresses the feasibility constraint of
the auction designer in terms of {\em position weights} $\wals$
satisfying $1 \ge \walk[1]\ge \walk[2] \ge \cdots \ge \walk[n] \geq
0$.  A {\em position auction} assigns agents (potentially randomly) to
positions $1$ through $n$, and an agent assigned to position $i$ gets
allocated with probability $\walk[i]$.  The {\em rank-by-bid position
  auction} orders the agents by their bids, with
ties broken randomly, and assigns agent $i$, with the $i$th largest
bid, to position $i$, with allocation probability $\walk[i]$.  {\em
  Multi-unit environments} are a special case and are defined for $k$
units as $\walk[j] = 1$ for $j \in \{1,\ldots,k\}$ and $\walk[j] = 0$
for $j \in \{k+1,\ldots,n\}$.  The {\em highest-$k$-bids-win}
multi-unit auction is the special case of the rank-by-bid position
auction for the $k$-unit environment.

In our model with agent values drawn i.i.d.\@ from a continuous
distribution, rank-by-bid position auctions with either all-pay or
first-price payment semantics have a unique Bayes-Nash equilibrium and
this equilibrium is symmetric and efficient, i.e., in equilibrium, the
agents' bids and values are in the same order \citep{CH-13}.

Rank-by-bid position auctions can be viewed as convex combinations of
highest-bids-win multi-unit auctions.  The {\em marginal weights} of a
position environment are $\margwals =
(\margwalk[1],\ldots,\margwalk[n])$ with $\margwalk = \walk -
\walk[k+1]$.  Define $\margwalk[0] = 1-\walk[1]$ and note that the
marginal weights $\margwals$ can be interpreted as a probability
distribution over $\{0,\ldots,n\}$.  As rank-by-bid position auctions
are efficient, the rank-by-bid position auction with weights $\wals$
has the exact same allocation rule as the mechanism that draws a
number of units $k$ from the distribution given by $\margwals$ and
runs the highest-$k$-bids-win auction.

Denote the highest-$k$-bids-win allocation rule as $\knalloc$ and its
per-agent revenue as $\murevk = \REV{\knalloc} =
\expect[\quant]{-\marg(\quant)\,\knalloc(\quant)}$.  This allocation
rule is precisely the probability an agent with quantile $\quant$ has
one of the highest $k$ quantiles of $n$ agents, or at most $k-1$ of
the $n-1$ remaining agents have quantiles greater than $\quant$.
Formulaically,
\begin{align*}
\knalloc(\quant) &= \sum_{i=0}^{k-1} \tbinom{n-1}{i}
\quant^{n-1-i}(1-\quant)^{i}.  
\\\intertext{Importantly, the allocation rule (in quantile space) of a
  rank-by-bid position auction does not depend on the distribution at
  all.  The allocation rule $\alloc$ of the rank-by-bid position
  auction with weights $\wals$ is:}
\alloc(\quant) &= \sum\nolimits_k \margwalk\, \knalloc(\quant).
\\\intertext{By revenue equivalence \citep{mye-81}, the per-agent
  revenue of the rank-by-bid position auction with weights $\wals$
  is:}
\REV{\alloc} &= \sum\nolimits_k \margwalk\,\murevk.
\end{align*}
Of course, $\murevk[0] = \murevk[n] = 0$ as always serving or never
serving the agents gives zero revenue.

A {\em rank-based} auction is one where the probability that an agent
is served is a function only of the rank of the agent's bid among the
other bids and not the magnitudes of the bids.  Any rank-based auction
induces a position environment where $\iwalk$ denotes the probability
that the agent with the $k$th ranked bid is served.  This auction is
equivalent to the rank-by-bid position auction with these weights
$\iwals$.  In a position environment with weights $\wals$, the following
lemma characterizes the weights $\iwals$ that are induced by
rank-based auctions.  

\begin{lemma}[e.g., \citealp{DHH-13}]
\label{l:rank-based-feasibility}
There is a rank-based auction with induced position weights $\iwals$
for a position environment with weights $\wals$ if and only if their
cumulative weights satisfy $\sum_{j=1}^{k} \iwalk[j] \leq
\sum_{j=1}^{k} \walk[j]$ for all $k$.
\end{lemma}

\subsection{Inference}

As we discussed in the introduction, the traditional structural
inference in the auction settings is based on inferring distribution
of values, which is unobserved but can be inferred from the
distribution of bids, which is observed.  Once the value distribution
is inferred, other properties of the value distribution such as its
corresponding revenue curve, which is fundamental for optimizing
revenue, can be obtained.  In this section we briefly overview this
approach.

The key idea behind the inference of the value distribution from the
bid distribution is that the strategy which maps values to bids is a
best response, by equation~\eqref{eq:br}, to the distribution of bids.
As the distribution of bids is observed, and given suitable continuity
assumptions, this best response function can be inverted.


We assume that the value distribution function $\dist(\cdot)$, the
allocation rule $\alloc(\cdot)$, and consequently also the quantile function of bid
distribution $\bid(\cdot)$, are monotone, continuously
differentiable, and invertible.

\paragraph{Inference for first-price auctions}
Consider a first-price rank-based auction with a symmetric bid
function $\bid(\quant)$ and allocation rule $\alloc(\quant)$ in BNE.
To invert the bid function we solve for the bid that the agent with
any value would make.  Continuity of this bid function implies that
its inverse is well defined.  Applying this inverse to the bid
distribution gives the value distribution.

The utility of an agent with quantile $\quant$ as a function of his bid $z$
is
\begin{align*}
\yestag\label{eq:fp-util}
\util(\quant,z) &= (\val(\quant) - z) \, \alloc(\bid^{-1}(z)).\\ 
\intertext{Differentiating with respect to $z$ we get:} 
\tfrac{d}{dz}\util(\quant,z) &= -\alloc(\bid^{-1}(z)) +
\big(\val(\quant) - z\big) \, \alloc'(\bid^{-1}(z))\,
\tfrac{d}{dz}\bid^{-1}(z).\\ 
\intertext{Here $\alloc'$ is the
  derivative of $\alloc$ with respect to the quantile $q$. Because
  $\bid(\cdot)$ is in BNE, the derivative $\tfrac{d}{dz}\util(\quant,z)$ is $0$ at
  $z=\bid(\quant)$. Rarranging, we obtain:}
  \yestag
  \label{eq:fp-inf}
  \val(\quant) &= \bid(\quant) + \tfrac{\alloc(\quant)\,\bid'(\quant)}{\alloc'(\quant)}
\end{align*}

\paragraph{Inference for all-pay auctions}
We repeat the calculation above for rank-based all-pay auctions; the starting
equation \eqref{eq:fp-util} is replaced with the analogous equation for all-pay auctions:
\begin{align*}
\yestag\label{eq:ap-util}
\util(\quant,z) &= \val(\quant)\,\alloc(\bid^{-1}(z)) - z.
\intertext{Differentiating with respect to $z$ we obtain:}
\tfrac{d}{dz}\util(\quant,z) &= \val(\quant)\,\alloc'(\bid^{-1}(z)) \,\frac{d}{dz}\bid^{-1}(z) - 1,\\
\intertext{Again the first-order condition of BNE implies that this expression is zero at $z = \bid(\quant)$; therefore,}
\yestag  \label{eq:ap-inf}
  \val(\quant) & = \tfrac{\bid'(\quant)}{\alloc'(\quant)}.
\end{align*}

\paragraph{Known and observed quantities}
Recall that the functions $\alloc(\quant)$ and $\alloc'(\quant)$ are
known precisely: these are determined by the rank-based auction
definition.  The functions $\bid(\quant)$ and $\bid'(\quant)$ are
observed.  The calculations above hold in the limit as the number of
samples from the bid distribution goes to infinity, at which point these
obserations are precise.

Equations~\eqref{eq:fp-inf} and~\eqref{eq:ap-inf} enable the value
function, or equivalently, the value distribution, to be estimated
from the estimated bid function and an estimator for the derivative of
the bid function, or equivalently, the density of the bid
distribution. Estimation of densities is standard; however, it
requires assumptions on the distribution, e.g., continuity, and the
convergence rates in most cases will be slower.  Our main results do not take
this standard approach. Below we discuss errors in estimation of the
bid function. 


\section{Revenue estimator and error bounds for all-pay auctions}
\label{s:inference}
\label{s:param-inf}
\label{S:PARAM-INF}

We will now describe our estimator for the revenue of one rank-based
auction using bids from another rank-based auction. There are two
advantages of the restriction of our analysis to rank-based auctions.
First, the allocation rule (in quantile space) of a rank-based auction
is independent of the bid and value distribution; therefore, it is
known and does not need to be estimated.  Second, the allocation rules
that result from rank-based auctions are well behaved, in particular
their slopes are bounded, and our error analysis makes use of this
property.

Recall from Section~\ref{sec:rank-based-basics} that the
revenue of any rank-based auction can be expressed as a linear
combination of the multi-unit revenues $\murevk[1],\ldots,\murevk[n]$
with $\murevk$ equal to the per-agent revenue of the
highest-$k$-bids-win auction. Therefore, in order to estimate the
revenue of a rank-based auction, it suffices to estimate each
$\murevk$ accurately.

In Section~\ref{s:inference-equation} we derive the counterfactual
revenue estimator.  We state and discuss the error bounds of this
estimator in Section~\ref{s:inference-k-thm}.  Two bounds are given;
the first bound holds in worst case over counterfactual and incumbent
mechanisms and the second bound depends on the closeness of the
allocation rules of the counterfactual and incumbent mechanisms.  The
main ideas of the derivation of the error bounds are given in
Section~\ref{s:inference-k}.

\subsection{The revenue estimator}
\label{s:inference-equation}


Consider estimating the revenue of an auction with allocation rule $y$
from the bids of an all-pay rank-based auction.  In terms of the
revenue curve $R(\cdot)$ or inverse demand function $\val(\cdot)$,
the per-agent revenue of the allocation rule $y$ is given by:
\begin{align}
\notag
P_y &= \expect[\quant]{y'(\quant)\,\rev(\quant)} =
\expect[\quant]{y'(\quant)\,\val(\quant)\,(1-\quant)}.\\
\intertext{Let $\alloc$ denote the allocation rule of the auction that we run,
and $\bid$ denote the bid distribution in BNE of this auction.  Recall
that for an all-pay auction format, we can convert the bid
distribution into the value distribution as follows: $\val(\quant) =
\bid'(\quant)/\alloc'(\quant)$.  Substituting this equation into the expression
for $P_y$ above we get}
\label{eq:inference-pre-integration-by-parts}
P_y &=
\expect[\quant]{y'(\quant)(1-\quant)\frac{\bid'(\quant)}{\alloc'(\quant)}}
      = \expect[\quant]{Z_y(\quant)\,\bid'(\quant)}
  \end{align}
where $Z_y(\quant)=(1-\quant)\frac{y'(\quant)}{\alloc'(\quant)}$.

To estimate $P_y$ the analyst
obtains $\samples$ samples from the bid distribution. Each sample is
the corresponding agent's best response to the true bid
distribution. 
%
%
We can estimate the quantile function of the equilibrium bid
distribution $\bid(\quant)$ as follows.  Let
$\sbidi[1],\cdots,\sbidi[\samples]$ denote the $\samples$ samples
drawn from the bid distribution.  Sort the bids so that $\sbidi[1]
\leq \sbidi[2] \leq \cdots \leq \sbidi[\samples]$ and define the {\em
quantile function of the empirical bid distribution} $\ebid(\cdot)$ as
\begin{align}\label{bid function}
\ebid(\quant) &= \sbidi &\forall i \in \samples, \quant \in [i-1,i)/\samples
\end{align}

We further observe that truncating the bid
distribution at its extremes results in a tradeoff of the variance of
the resulting estimator (which can diverge at the extreme quantiles of the bid distribution) with a bias (which is
bounded). Accordingly, we obtain the following estimator.

\begin{definition}
\label{d:estimator}
The estimator $\hat{P}_y$ (with truncation parameter $\delta_N$) for
the revenue of an auction with allocation rule $y$ from $N$ samples
$\hat{b}_1 \leq \cdots \leq \hat{b}_N$ from the equilibrium bid
distribution of an all-pay auction with allocation rule $x$ is:
\begin{align*}
\hat{P}_y& = \sum_{i=\delta_N N}^{N-\delta_N N} \left(1-\frac
    {i} N\right)\frac{y'(\frac {i}N)}{x'( \frac {i}N)} 
           \, \left(\hat{\bid}_{i+1}-\hat{\bid}_{i}\right).
\end{align*}
\end{definition}

This estimator is obtained
from (\ref{eq:inference-pre-integration-by-parts}) by integration by parts
over range of the quantile function $\ebid(\cdot)$
of the empirical
bid distribution with its support
truncated to $[\delta_N,\,
1-\delta_N]$
and re-grouping of the terms
in the resulting sum.

Our main theorems set
the truncation parameter to a specific value $\delta_N = \max(25 \log
\log N,n)/N$ and show that, with no assumptions on the distribution of
values or bids, the truncated estimator's mean absolute error is
bounded.  Importantly, this truncated estimator does not have any
parameters that need to be tuned to the distribution of bids or
values.  

We refer to the estimator with truncation parameter set to zero as the
untruncated estimator.  We study the untrucated estimator in
simulations in Section~\ref{s:simulations} and demonstrate that, when
$Z_y(0)$ and $Z_y(1)$ are small, it can be quite accurate with very
few bid samples.

\subsection{Error bounds} 
\label{s:inference-k-thm}

Our first main result of this section is the following error bound for
the estimator of Definition~\ref{d:estimator}.

\begin{theorem}\label{thm:allpay-simple}
  The mean absolute error in estimating the revenue of a rank-based
  auction with allocation rule $y$ using $N$ samples from the bid
  distribution for an all-pay rank-based auction with allocation rule
  $x$ is bounded as below. Here $n$ is the number of positions in the
  two auctions, and $\hat{P}_y$ is the estimator in Definition~\ref{d:estimator}
  with $\delta_N$ set to $\max(25\log\log N,n)/N$.
\begin{align*}
\Err{P_y} & \le \frac{16n^2\log N}{\sqrt{N}}.
\end{align*}
\end{theorem}


Observe that the above error bound is independent of the allocation
rules $x$ and $y$. When $x$ and $y$ are similar to each other, our
estimator should in fact achieve a much better error rate than the one
above. For example, when $x$ and $y$ are identical, the error in
estimation should have the same dependence on the number of samples as
the statistical error in bids, namely $1/\sqrt{\samples}$. Our next
theorem quantifies this relationship.

In order to capture the dependence of our error bounds in estimating
$P_y$ on the relationship between the incumbent allocation rule $x$
and the counterfactual allocation rule $y$, we define a new
quantity, $\rxy$, as follows:
\begin{align}
\label{eq:rxy}
\rxy &:= \sup_{\quant}\{y'(\quant)\}\,\max\left\{1, \,\log 
\sup_{\quant: y'(\quant)\ge 1}\frac{\alloc'(\quant)}{y'(\quant)}, \,\log
\sup_{\quant}\frac{y'(\quant)}{\alloc'(\quant)} \right\}.
\end{align}

\noindent
We then obtain the following theorem for the special case of
estimating the multi-unit revenues.

\begin{theorem}
\label{thm:allpay-general}
  Let $\alloc$ and $\kalloc$ denote the allocation rules for any
  all-pay rank-based auction and the $k$-highest-bids-win auction over
  $n$ positions, respectively. Let $\hat{\murevk}$ denote the
  estimator from Definition~\ref{d:estimator} for estimating the
  revenue $P_k$ of the latter auction from $\samples$ samples of the
  bid distribution of the former, with $\delta_N$ set to
  $\max(25\log\log N,n)/N$. If $\delta_N\le 1/n$, the mean absolute
  error of the estimator $\hat{\murevk}$ is bounded as follows.
\begin{align*}
  \Err{\murevk}
\le & \,\frac{80}{\sqrt{\samples}}\,\rxy[\kalloc].
\end{align*}
\end{theorem}

We obtain a slightly worse error bound when $y$ is a general rank-based auction:

\begin{corollary}
\label{cor:allpay-y}
  Let $\alloc$ and $y$ denote the allocation rules for any two all-pay
  rank-based auctions over $n$ positions. Let $\hat{P}_y$ denote the
  estimator from Definition~\ref{d:estimator} for estimating the
  revenue of the latter from $\samples$ samples of the bid
  distribution of the former, with $\delta_N$ set to $\max(25\log\log
  N,n)/N$. If $\delta_N\le 1/n$, the mean absolute error of the
  estimator $\hat{P}_y$ is bounded as follows.
\begin{align*}
  \Err{P_y}
\le & \,\frac{80}{\sqrt{\samples}}\, n \, \log \sup\nolimits_q n\, \frac{y'(q)}{x'(q)}.
\end{align*}
\end{corollary}

We sketch the main ideas of \Cref{thm:allpay-simple} in
\Cref{s:inference-k}.  The full proof and the refinement
necessary to obtain \Cref{thm:allpay-general} and
\Cref{cor:allpay-y} are given in \Cref{s:proofs-3}.

From Theorem~\ref{thm:allpay-general}, the error in the estimator
depends on the slopes of the allocation rules $x$ and $y$. The maximum
slope of the multi-unit allocation rules, and therefore also that of
any rank-based auction, is always bounded by $n$, the number of agents
in the auction (summarized as Fact~\ref{fact:max-slope}, below).

\begin{fact}
\label{fact:max-slope}
The maximum slope of the allocation rule $\alloc$ of any $n$-agent
rank-based auction is bounded by $n$: $\sup_q \alloc'(q)\le n$. More
specifically, the maximum slope of the allocation rule $\kalloc$ for
the $n$-agent highest-$k$-bids-win auction is bounded by $$\sup_q
\kalloc'(q)\in \left[\frac{1}{\sqrt{2\pi}},\frac{1}{\sqrt{\pi}}\right]
\frac{n-1}{\sqrt{\min\{k-1, n-k\}}} = \Theta\left(\frac
n{\sqrt{\min\{k,n-k\}}}\right).$$
\end{fact}

We evaluate the error bound given by \Cref{thm:allpay-general}
for a few special cases of $x$ and $x_k$.  For simplicity in applying
\Cref{fact:max-slope}, we assume $k < n/2$.

\begin{itemize}
\item When $\alloc=\kalloc$, $\rxy[\kalloc]\le n/\sqrt{k}$ and we get
  an error bound of $80
  \frac{n}{\sqrt{Nk}}$, which is the same (within a constant factor)
  as the statistical error in bids.
\item The bound in the previous case degrades smoothly when $\alloc$
  and $\kalloc$ are close but not identical, as in $\eps\kalloc'\le \alloc'\le
  \kalloc'/\eps$ for $\eps>0$. We have $\rxy[\kalloc]\le
  \log(1/\eps)\, n/\sqrt{k}$ and the error bound is: $80 \log(1/\eps)
  \frac{n}{\sqrt{Nk}}$.
\item Finally, as long as $\alloc'\ge \eps\kalloc'$, that is, the
  highest-$k$-bids-win auction is mixed in with $\epsilon$ probability
  into $\alloc$, we observe via Fact~\ref{fact:max-slope} that $\sup_{q:
    \kalloc'(q)\ge 1} \alloc'(q)/\kalloc'(q) \le \sup_q \alloc'(q)\le
  n$, and obtain an error bound of $80 \log(n/\eps) \frac{n}{\sqrt{Nk}}$.
\end{itemize}

\subsection{Applications to A/B testing}
\label{s:ab-testing}
\label{S:AB-TESTING}
\label{s:ab-revenue}



Consider the setup described in the introduction where an
auction house running auction $A$ would like to determine the revenue
of a novel mechanism $B$. The typical approach for doing so is to run
the auction $B$ with some probability $\eps>0$ and $A$ with the
remaining probability. Ideally, if in doing so, the auction house
obtains $\eps N$ bids in response to the auction $B$ out of a total of
$N$ bids, the revenue of $B$ can be estimated within an error bound of
\begin{align}
\label{error-ideal}
\Theta\left(\frac 1{\sqrt{\eps}}\right)\frac{\sup_{\quant}\{x_B'(\quant)\}}{\sqrt{N}}
\end{align}
where $x_B$ denotes the allocation rule corresponding to $B$.  We
refer to this approach as {\em ideal A/B testing}.

In practice, however, instead of obtaining bids in equilibrium for
auction $B$, the analyst obtains bids in equilibrium for the aggregate
mechanism $C=(1-\eps) A +\eps B$. We can then use
Definition~\ref{d:estimator} to estimate the revenue of $B$.
As a consequence of Corollary~\ref{cor:allpay-y}, and noting that 
$x_B'(\quant)/x_C'(\quant)\le 1/\eps$ for all quantiles $\quant$, we
obtain the following error bound.
\begin{corollary}
\label{cor1} 
The revenue of a rank based mechanism $B$ can be estimated from $N$
bids of a rank-based mechanism $C=(1-\eps)A+\eps B$ with absolute error
bounded by
\begin{align}
\label{error-true}
 \frac{80\, n \log(n/\eps)}{\sqrt{\samples}}.
\end{align}
\end{corollary}

Relative to the ideal situation described above, our error bound has a
better dependence on $\eps$ and a worse dependence on $n$. Note that
when $\eps$ is very small, our error bound of equation
\eqref{error-true} may be smaller than the ideal bound of equation
\eqref{error-ideal}. In fact, we obtain a non-trivial bound on the
error even when $\eps=0$, per Theorem~\ref{thm:allpay-simple}:
\begin{corollary}
\label{cor1-new} 
The revenue of a rank based mechanism $B$ can be estimated from $N$
bids of any rank-based mechanism $C$ with absolute error
bounded by
\begin{align}
\label{error-true-new}
\frac{16\,n^2 \log \samples}{\sqrt{\samples}}.
\end{align}
\end{corollary}
This is not surprising: the ideal bound ignores information that we
can learn about the revenue of $B$ from the $(1-\eps)N$ bids obtained
when $B$ is not run.


When $B$ is a multi-unit auction, we obtain a slightly better error bound using
Theorem~\ref{thm:allpay-general} which is closer to the ideal bound of
equation \eqref{error-ideal}.
\begin{corollary}
\label{cor2}
  The revenue of the highest-$k$-bids-win mechanism $B$ can be
  estimated from $N$ bids of a rank-based mechanism $C=(1-\eps)A+\eps
  B$ with absolute error bounded by
\begin{align}
\label{error-true-k}
\frac{80 \sup_{\quant}\{x_B'(\quant)\} \log (n/\eps)}{\sqrt{N}}.
\end{align}
\end{corollary}

\section{Applications to instrumented optimization}
\label{s:iron-rank}

In this section we consider the problem of the principal who would
like a mechanism that optimizes revenue for the current distribution
of agent values while simultaneously enabling the inference necessary
to reoptimize the mechanism in the future, should the distribution of
values change.  Recall, the optimal auctions of the classical theory
pool agents with distinct values and are thus not well suited to
counterfactual inference.  In \Cref{s:iron-rank-opt} we develop
a theory for optimizing revenue over the class of all rank-based
auctions that resembles Myerson's theory for optimal auction design.
Importantly, the optimal rank-based auction does not require knowledge
of the full distribution, instead the multi-unit revenues
$(\murevk[1],\ldots,\murevk[n])$ are sufficient; moreover, the
previous developments of this paper enable the estimation of these
multi-unit revenues.  Where Myerson's theory employs ironing by value
and value reserves, our approach analogously employs ironing by rank
and rank reserves. In \Cref{s:universal} we extend the
A/B-testing approach of \Cref{s:param-inf} to develop a
``universal B test mechanism'' that can be used to estimate all of the
multi-unit revenues $(\murevk[1],\ldots,\murevk[n])$ simultaneously.
Combined with the optimal rank-based mechanism A, the A/B-test mechanism
with this universal B test is simultaneously good for revenue and
counterfactual inference.  In \Cref{sec:strict-opt} we take a
more principled approach to optimization subject to inference, and
solve for the revenue-optimal mechanism subject to the constraint that
all multi-unit revenues can be estimated.



We begin by reviewing position environments and rank-based
auctions. In a rank-based auction the allocation to an agent depends
solely on the ordinal rank of his bid among other agents' bids, and
not on the cardinal value of the bid. For a position environment, a
rank-based auction assigns agents (potentially randomly) to positions
based on their ranks. Consider a position environment given by
non-increasing weights $\wals = (\walk[1],\ldots,\walk[n]$).  For
notational convenience, define
$\walk[n+1] = 0$.  Define the cumulative position
weights $\cumwals = (\cumwalk[1],\ldots,\cumwalk[n])$ as $\cumwalk =
\sum_{j=1}^k \walk[j]$, and $\cumwalk[0]=0$. We can view the cumulative weights as defining
a piece-wise linear, monotone, concave function given by connecting
the point set $(0,\cumwalk[0]), \ldots, (n,\cumwalk[n])$.

Multi-unit highest-bids-win auctions form a basis for position
auctions. Consider the marginal position weights $\margwals =
(\margwalk[1],\ldots,\margwalk[n])$ defined by $\margwalk = \walk -
\walk[k+1]$.  The allocation rule induced by the position auction with
weights $\wals$ is identical to the allocation rule induced by the
convex combination of multi-unit auctions where the $k$-unit auction
is run with probability $\margwalk$.

A randomized assignment of agents to positions based on their ranks
induces an expected weight to which agents of each rank are assigned,
e.g., $\iwalk$ for the $k$th ranked agent.  These expected weights can
be interpreted as a position auction environment themselves with
weights $\iwals$. As for the original weights, we can define the
cumulative position weights $\cumiwals$ as $\cumiwalk = \sum_{j=1}^k
\iwalk[j]$. \Cref{l:rank-based-feasibility} characterizes the position
weights $\iwals$ that can be induced by any rank-based auction in a
position environment $\wals$ as those with cumumlative weights upper bounded
by thoses of the position environment, i.e., $\iwals$ is feasible for
$\wals$ if and only if $\cumiwalk \leq \cumwalk$ for all $k$.

Any feasible weights $\iwals$ can be constructed from $\wals$ by a
(random) sequence of the following two operations (cf. \citet{HLP-29},
and proof in Appendix~\ref{a:position-auction-construction}).
\begin{description}
\item[rank reserve] For a given rank $k$, all agents with ranks
  between $k+1$ and $n$ are rejected.  The resulting weights $\iwals$
  are equal to $\wals$ except $\iwalk[k'] = 0$ for $k' > k$.
 
\item[iron by rank] Given ranks $k' < k''$, the ironing-by-rank
  operation corresponds to, when agents are ranked, assigning the
  agents ranked in an interval $\{k',\ldots,k''\}$ uniformly at
  random to these same positions.  The ironed position weights
  $\iwals$ are equal to $\wals$ except the weights on the ironed
  interval of positions are averaged.  The cumulative ironed position
  weights $\cumiwals$ are equal to $\cumwals$ (viewed as a concave
  function) except that a straight line connects
  $(k'-1,\cumiwalk[k'-1])$ to $(k'',\cumiwalk[k''])$.  Notice that
  concavity of $\cumwals$ (as a function) and this perspective of the
  ironing procedure as replacing an interval with a line segment
  connecting the endpoints of the interval implies that $\cumiwals \leq
  \cumwals$ coordinate-wise, i.e., $\cumiwalk \leq \cumwalk$ for all
  $k$.
\end{description}



\subsection{Optimal rank-based auctions}
\label{s:iron-rank-opt}

In this section we describe how to optimize for expected revenue over
the class of rank-based auctions. Recall that rank-based auctions
are linear combinations over $k$-unit auctions. The characterization
of Bayes-Nash equilibrium, cf.\@ equation~\eqref{eq:bne-rev}, shows
that revenue is a linear function of the allocation rule. Therefore,
the revenue of a position auction can be calculated as the convex
combination of the revenue $\murevk$ from the $k$-highest-bids-win
auction for $k\in \{0,\ldots, n\}$.  Note that $\murevk[0] = \murevk[n] = 0$.

Given these multi-unit revenues, $\murevs =
(\murevk[0],\ldots,\murevk[n])$, the problem of designing the optimal
rank-based auction is well defined: given a position environment
with weights $\wals$, find the weights $\iwals$ for a rank-based
auction with cummulative weights $\cumiwals \leq \cumwals$ maximizing
the sum $\sum_{k} (\iwalk-\iwalk[k+1])P_k$.  This optimization problem
is isomorphic to the theory of envy-free optimal pricing
developed by \citet{DHY-14}.  We summarize this theory below; a
complete derivation can be found in Appendix~\ref{s:iron-opt-app}.

Define the {\em multi-unit revenue curve} as the piece-wise linear
function connecting the points $(0,\murevk[0]),\ldots,(n,\murevk[n])$.
This function may or may not be concave.  Define the {\em ironed
  multi-unit revenues} as $\imurevs =
(\imurevk[0],\ldots,\imurevk[n])$ according to the smallest concave
function that upper bounds the multi-unit revenue curve.  Define the
multi-unit marginal revenues, $\mumargs =
(\mumargk[1],\ldots,\mumargk[n])$ and $\imumargs =
(\imumargk[1],\ldots,\imumargk[n])$, as the left slope of the
multi-unit and ironed multi-unit revenue curves, respectively.  I.e.,
$\mumargk = \murevk - \murevk[k-1]$ and $\imumargk = \imurevk -
\imurevk[k-1]$.  The proof of the following theorem is given in the
appendix.

\begin{theorem}
\label{thm:rank-based-opt}
Given a position environment with weights $\wals$, the revenue-optimal
rank-based auction is defined by position weights $\iwals$ that are
equal to $\wals$, except ironed on the same intervals as $\murevs$ is
ironed to obtain $\imurevs$, and set to $0$ at positions $k$ for which
$\imumargk$ is negative. 
\end{theorem}

As is evident from this description of the optimal rank-based
auction, the only quantities that need to be ascertained to run this
auction is the multi-unit revenue curve defined by $\murevs$.
Therefore, an econometric analysis for optimizing rank-based
auctions need not estimate the entire value distribution; estimation
of the multi-unit revenues is sufficient.

\subsection{Universal B test}
\label{s:universal}

In Section~\ref{s:ab-revenue} we discussed how to estimate the revenue
of a single auction B from the bids of the A/B test mechanism
C. Corollary~\ref{cor1-new} shows that an A/B test is not necessary as
long as we have enough samples from the bid distribution: the revenue
of B can be estimated from any incumbent mechanism A directly. In
fact, we can estimate the revenue of {\em all} rank-based mechanisms
simultaneously from the bids of a single mechanism A. However, the
error in estimation depends suboptimally on the number of samples, as
$\log(\samples)/\sqrt{\samples}$ rather than $1/\sqrt{\samples}$. A
natural question is whether it is possible to estimate all rank-based
revenues simultaneously at an optimal error rate from bids of a single
incumbent auction. Precisely, we now consider the problem identifying
a B test mechanism for which the revenue of any position auction D can
be estimated from the equilibrium bids in the A/B test mechanism C.
Since the revenue of D is given by the convex combination of the
multi-unit revenues $\murevk$, it suffices to estimate all of these
multi-unit revenues.  What properties should the auction B have in
order to enable this estimation?  (Equivalently, what properties
should C have?)

\begin{definition} 
A {\em universal B test mechanism} satisfies; for
any rank-based auctions A and D, any $\eps>0$, and auction C defined
by $x_C = (1-\eps)x_A +\eps x_B$; the revenue $\murevk[D]$ 
can be estimated from $N$ equilibrium bids of C with the dependence of
the mean absolute error on $N$ and $\eps$ bounded by
$O(\log(1/\eps)/\sqrt{N})$.
\end{definition}

Since the revenue of D can be estimated from the revenue of all
multi-unit auctions, Corollary~\ref{cor2} implies that it suffices to
mix every multi-unit auction into C with some small probability.  The
uniform-stair mechanism (Definition~\ref{d:uniform-stair} in
Section~\ref{s:simulations}), with position weights $\walk =
\frac{n-k}{n-1}$ for each $k$, gives a mechanism B with such a
mixture.


\begin{corollary}
\label{cor3}
The uniform-stair position auction is a universal B test mechanism
with mean absolute error bounded by
$80\, n\log (n/\eps)/\sqrt{N}$.
\end{corollary}

Next we observe that in fact we can get similar results by mixing in
just a few of the multi-unit auctions. In particular, in order to
estimate $\murevk$ accurately, it suffices to mix in a multi-unit
auction with no more than $k$ units, and another one with no less than
$k$ units. This gives us a more efficient universal B test for
simultaneously inferring all of the multi-unit revenues (see
Corollary~\ref{cor:universal}).  

\begin{lemma}
\label{lem:univ}
  The revenue of the highest-$k$-bids-win mechanism B can be
  estimated from $N$ bids of a rank-based all-pay auction C $=$ $(1-2\eps)$A $+\eps
  $B$_1+\eps $B$_2$ where A is an arbitrary rank-based auction, and
  B$_1$ and B$_2$ are the highest-$k_1$-bids-win and
  highest-$k_2$-bids-win auctions respectively, with $k_1\le k\le
  k_2$. The absolute error of the estimator is bounded by
\begin{align*}
\frac{80}{\sqrt{N}} (n+\log (1/\eps))\,\,\sup_{\quant}\{x_k'(\quant)\}
\end{align*}
\end{lemma}
\begin{proof}
We begin by noting that for any $j$ and $k$ with $k\le j$,
$$
\frac{\alloc'_k(\quant)}{\alloc'_j(\quant)} = \frac{{n-2\choose
   k-1}}{{{n-2}\choose {j-1}}} \left(\frac{q}{1-q}\right)^{j-k}.
$$
When $k\le j$ and $q\le 1/2$, this ratio is less than $2^n$.
Likewise, when
$k\ge j$ and $q\ge 1/2$, the ratio is less than $2^n$.
Therefore, for any $q$, and C$\ =\ (1-2\eps)$A$\ +\ \eps$B$_1\ +\ \eps$B$_2$
where B$_1$ and B$_2$ are the highest-$k_1$-bids-win and
highest-$k_2$-bids-win auctions respectively, with $k_1\le k\le k_2$,
we have
$$
\sup_q \frac{\alloc'_k(\quant)}{\alloc'_C(\quant)} \le \frac{2^n}{\epsilon}.
$$
Next we note that $\sup_q \alloc'_C(q)\le n$ and, therefore, $\sup_{\quant:
    \kalloc'(\quant)\ge 1}\frac{\alloc_C'(\quant)}{\kalloc'(\quant)}\le n$.
Putting these quantities together with Theorem~\ref{thm:allpay-general}, we
get that the absolute error in estimating $\murevk$ from bids drawn from C is at
most 
\begin{align*}
\Err{\murevk} \le & \frac{80}{\sqrt{\samples}}\,\,
\sup_{\quant}\{\kalloc'(\quant)\}\,\, \left(n+\log
  1/\eps\right). \qedhere 
\end{align*}
\end{proof}

\begin{corollary}
\label{cor:universal}
The rank-by-bid position auction with weights $w_1=1$, $w_k=1/2$ for
$1< k<n-1$, and $w_n=0$
is a universal B test mechanism
with mean absolute error bounded by $O(n(n+\log
  (1/\eps))/\sqrt{N})$.
  \end{corollary}

\subsection{Optimal rank-based auctions with strict monotonicity}
\label{sec:strict-opt}

Position auctions, by definition, have non-increasing position weights
$\wals$.  The ironing in the iron-by-rank optimization of
Section~\ref{s:iron-rank-opt} converted the problem of optimizing
multi-unit marginal revenue subject to non-increasing position weight,
to a simpler problem of optimizing multi-unit marginal revenue without
any constraints.  In this section, we describe the optimization of
rank-based auctions (i.e., ones for which position weights can be
shifted only downwards or discarded) subject to {\em strictly
  decreasing} position weights.  In particular, we can reinterpret the
decreasing position weights of the universal B test mechanism from
\autoref{s:universal} as such a strictness requirement.  The optimal
mechanism with this strictness requirement will satisfy the same
inference guarantee proven for the A/B test while improving its
revenue.

As described by Lemma~\ref{l:rank-based-feasibility}, position weights
$\iwals$ are feasible as a rank-based auction in the position environment
$\wals$ if the cumulative position weights satisfy $\cumwalk \geq
\cumiwalk$ for all $k$.  Suppose we would like to optimize $\iwals$ 
for position weight $\wals$ subject to the monotonicity constraint
that the difference in successive position weights is at least that of
some other position weights $\epswals =
(\epswalk[1],\ldots,\epswalk[n])$.  Formally, $\margiwalk = \iwalk
- \iwalk[k+1] \geq \epswalk - \epswalk[k+1] = \margepswalk$ for all
$k$.  For example, $\epswals$ could be $\epsilon$ times the position
weights of the universal B test mechanism of the preceding section.
We call an allocation rule satisfying these monotonicity constraints
an $\epswals$-strictly-monotone allocation rule. As non-trivial
ironing by rank always results in consecutive positions with the same
weight, i.e., $\margiwalk = 0$ for some $k$, the optimal rank-based
mechanism with strict monotonicity will require overlapping ironed intervals.

To our knowledge, performance optimization subject to a strict
monotonicity constraint has not previously been considered in the
literature.  At a high level our approach is the following.  We start
with $\wals$ which induces the cumulative position weights $\cumwals$
which constrain the resulting position weights $\iwals$ of any
feasible rank-based auction via its cumulative $\cumiwals$.  We view
$\iwals$ as the combination of two position auctions. The first has
weakly monotone weights $\iyals = (\iyalk[1],\ldots,\iyalk[n])$; the
second has strictly monotone weights $\epswals =
(\epswalk[1],\ldots,\epswalk[n])$; and the combination has weights $\iwalk
= \iyalk + \epswalk$ for all $k$.  The revenue of the combined position
auction is the sum of the revenues of the two component position
auctions.  Since the second auction has fixed position weights, its
revenue is fixed.  Since the first position auction is weakly monotone
and the second is strictly, the combined position auction is strictly
monotone and satisfies the constraint that
$\margiwalk \geq \margepswalk$ for all $k$.

This construction focuses attention on optimization of $\iyals$
subject to the induced constraint imposed by $\wals$ and after the
removal of the $\epswals$-strictly-monotone allocation rule.  I.e.,
$\iwals$ must be feasible for $\wals$.  The suggested feasibility
constraint for optimization of $\iyals$ is given by position weights
$\yals$ defined as $\yalk = \walk - \epswalk$.  Notice that, in
this definition of $\yals$, a lesser amount is subtracted from
successive positions.  Consequently, monotonicity of $\wals$ does
not imply monotonicity of $\yals$.

To obtain $\iyals$ from $\yals$ we may need to iron for two reasons,
(a) to make $\iyals$ monotone and (b) to make the multi-unit revenue
curve monotone.  In fact, both of these ironings are good for revenue.
The ironing construction for monotonizing $\yals$ constructs the
concave hull of the cumulative position weights $\cumyals$.  This
concave hull is strictly higher than the curve given by $\cumyals$ (i.e.,
connecting $(0,\cumyalk[0]), \ldots, (n,\cumyalk[n])$).  Similarly the
ironed multi-unit revenue curve given by $\imurevs$ is the concave
hull of the multi-unit revenue curve given by $\murevs$.  The correct
order in which to apply these ironing procedures is to first (a) iron
the position weights $\yals$ to make it monotone, and second (b) iron
the multi-unit revenue curve $\murevs$ to make it concave.  This order
is important as the revenue of the position auction with weights
$\iyals$ is only given by the ironed revenue curve $\imurevs$ when the
$\margiyals = 0$ on the ironed intervals of $\imurevs$.

\begin{theorem}
\label{thm:rank-based-opt-strict}
The optimal $\epswals$-strictly-monotone rank-based auction for position
weights $\wals$ has position weights $\iwals$ constructed by
\begin{enumerate}
\item defining $\yals$ by $\yalk = \walk - \epswalk$ for all $k$,
\item averaging position weights of $\yals$ on intervals where $\yals$
  should be ironed to be monotone,
\item averaging the resulting position weights on intervals where
  $\murevs$ should be ironed to be concave to get $\iyals$, and
\item setting $\iwals$ as $\iwalk = \iyalk + \epswalk$;
\end{enumerate}
and is feasible for $\wals$ if $\epswals$ is feasible for $\wals$.
\end{theorem}

\begin{proof} 
The proof of this theorem follows directly by the construction and its
correctness.
\end{proof}

As described previously, the rank-based auction given by $\iwals$ in
position environment given by $\wals$ can be implemented by a sequence
of iron-by-rank and rank-reserve operations.  Such a sequence of
operations can be found, e.g., via an approach of \citet{AFHHM-12}
or \citet{HLP-29}.

The following proposition shows that this optimal
$\epswals$-strictly-monotone mechanism inherits the inference
properties of the mechanism with position weights $\epswals$, in
particular, the A/B testing results of
\Cref{cor1}, \Cref{cor2}, \Cref{cor3},
and~\Cref{cor:universal}.

\begin{proposition} 
\label{prop:optimization-with-inference-metatheorem}
For position weights $\epswals$ defined as $\epsilon$ times the
position weights of a B test mechanism, if position weights $\epswals$
are feasible for $\wals$ then the optimal $\epswals$-strictly-monotone
rank-based auction for position weights $\wals$ has the same inference
guarantee as the A/B test with $\epsilon$ probability of B.
\end{proposition}

\section{Simulation evidence}

\label{s:simulations}

In this section we present evidence from simulations that corroborate
our theoretical analyses and provide further understanding of various
methods for controlling estimation error.  In the first subsection, we
explore the dependence of the error of the estimator on the main
parameters of our theoretical analyses: the number $n$ of agents in
each auction and the number $\samples$ of samples the analyst obtains
from the bid distribution.  The second subsection considers ex ante
methods for controlling estimation error, a.k.a., instrumentation.
The two methods compared are A/B testing, where the counterfactual
mechanism B is mixed in with the incumbent mechanisms, and the
universal B-test where the universal B-test mechanism is mixed in with
the incumbent mechanism.  For both methods, the error is considered as
a function of the amount $\epsilon$ by which the B-test auction is
mixed with the incumbent auction.  The third subsection contrasts ex
post methods of controlling estimation error.  We compare a standard
approach of smoothing to estimate the empirical bid distribution that is plugged into the bidder's first order condition versus our approach
of truncating the contribution to the estimator from extreme
quantiles.  The above analyses are conducted under the assumption that
agent values are drawn from a beta distribution; the final subsection
demonstrates that the same qualitative results hold for a wide range
of value distributions.

\begin{figure}
    \begin{subfigure}[t]{0.48\textwidth}
    \centering
    \includegraphics[width=\linewidth]{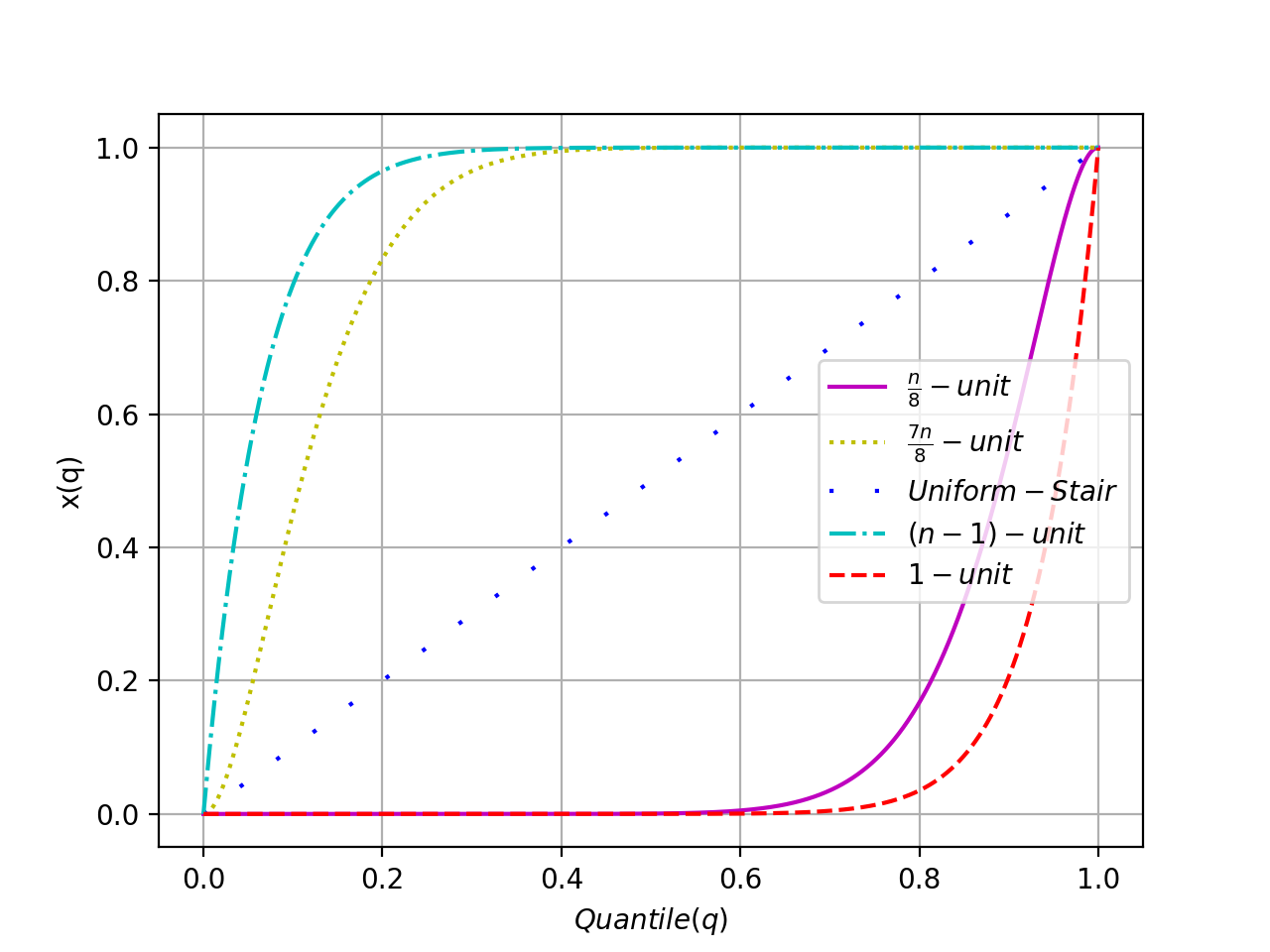} 
    \caption{$n=16$} 
  \end{subfigure}
  \hfill
  \begin{subfigure}[t]{0.48\textwidth}
    \centering
    \includegraphics[width=\linewidth]{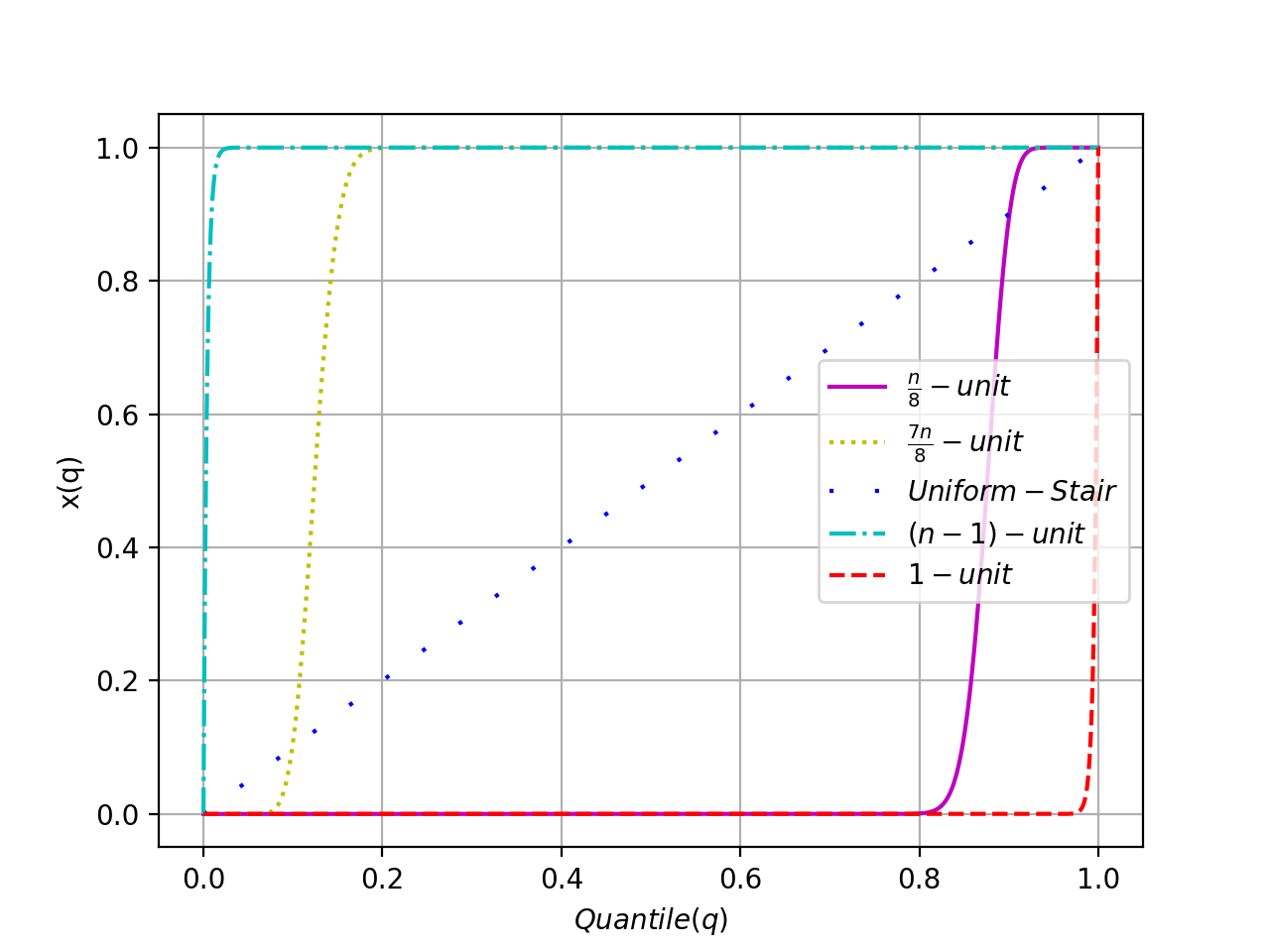} 
    \caption{$n=256$} 
  \end{subfigure}
  \caption{The five allocation rules of our empirical study are
    depicted.  Note that with $n=16$ the active regions of the high
    supply and 1-unit auction overlap, while with $n=256$ they do not
    (similarly for the low supply and $(n-1)$-unit auction).}
  \label{fig:five-alloc-rules}
\end{figure}

Our empirical analysis focuses on the following five $n$-agent position auctions (\Cref{fig:five-alloc-rules}). 
\begin{itemize}
\item the {\em low supply} ($\ceil{n/8}$-unit) auction.
\item the {\em high supply} ($\floor{7n/8}$-unit) auction.
\item the {\em uniform-stair} auction with allocation rule $x(q) = q$ (see \Cref{d:uniform-stair}, below).
\item the {\em 1-unit} auction with allocation rule $x(q) = q^{n-1}$.
\item the {\em $(n-1)$-unit} auction with allocation rule $x(q) = 1-(1-q)^{n-1}$.
\end{itemize}
We will be interested in estimating the revenue of one auction from
the sample of $N$ bids in another auction.  The 1-unit and
$\ceil{n/8}$-unit auctions are extremal low-unit auctions.  The
$\floor{7n/8}$-unit and $(n-1)$-unit auctions are extremal high-unit
auctions.  The uniform stair auction is a position auction with
uniformly decreasing weights, equivalently, constant marginal weights.

\begin{definition}
  \label{d:uniform-stair}
  The {\em uniform-stair allocation rule} is $x(q) = q$; it is induced by the {\em uniform-stair auction}, an $n$-agent position auction defined by weights $\wals = (\walk[1],\ldots,\walk[n]) = (1,\tfrac{n-2}{n-1},\ldots,\tfrac{1}{n-1},0)$.
\end{definition}

Our empirical study allows benchmarking the error of our estimator
against the {\em counterfactual error}, i.e., the estimation error had
the incumbent mechanism been the counterfactual.  With this benchmark,
we see the loss (or gain) in accuracy of our approach relative to the
straightforward statistical task of estimating the revenue of an
auction from samples from the auction's bid distribution.  Over all of
the studies we ran, we did not observe mean absolute error of our
estimator to exceed the counterfactual error by more than a factor of
10.

\subsection*{Methodology}

We perform simulations to calculate the mean absolute deviation of our
estimator $\hat{P}_B$ for the revenue of auction B with the auction's
expected revenue $P_B$.  The allocation rules $\alloc_B$ and
$\alloc_C$, their derivatives $\alloc'_B$ and $\alloc'_C$, and the
revenue curve $\rev$ are calculated analytically.  The expected
revenue $P_B$ is calculated from the revenue curve $\rev$ and
$\alloc'_B$ by equation~\eqref{eq:bne-rev} via numerical integration
(i.e., by averaging the values of $\rev(\quant)\,\alloc'_B(\quant)$ on
a grid).  
The equilibrium bid distribution in auction C for values on a uniform
grid are calculated from equation~\eqref{eq:ap-inf} via numerical
integration on a grid.  
Each simulation draws $N$ bids from this bid distribution, 
the estimated revenue $\hat{P}_B$ is
calculated from Definition~\ref{d:estimator}, and the mean absolute
deviation is calculated by averaging $|P_B - \hat{P}_B|$ over 8000
Monte Carlo simulations.

\subsection{Empirical Evidence versus Theoretical Bound}

\begin{figure}[tb]
  \begin{subfigure}[t]{0.32\textwidth}
    \centering
    \includegraphics[width=\linewidth]{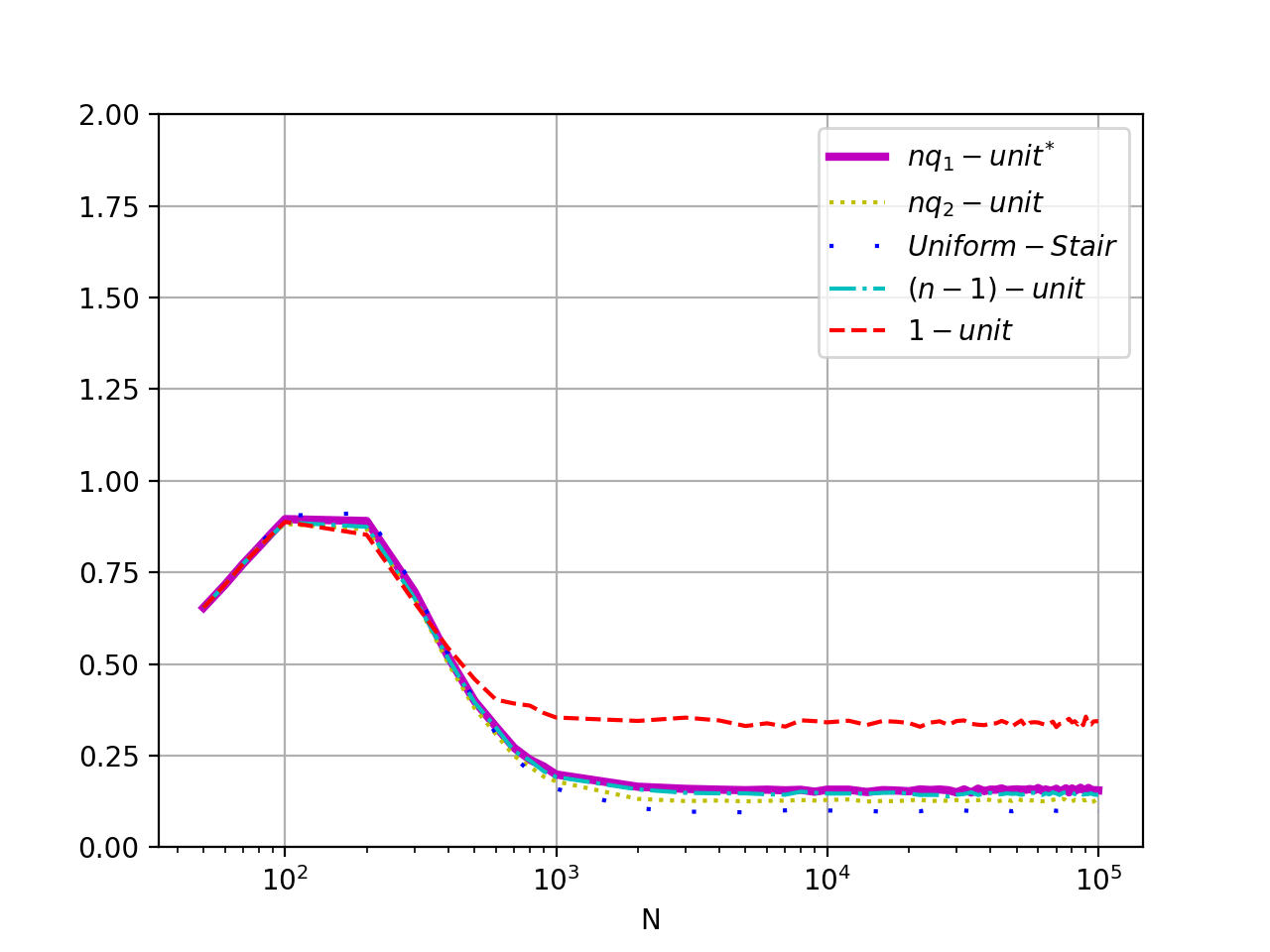} 
    \caption{low supply} 
  \end{subfigure}
  \hfill
  \begin{subfigure}[t]{0.32\textwidth}
    \centering
    \includegraphics[width=\linewidth]{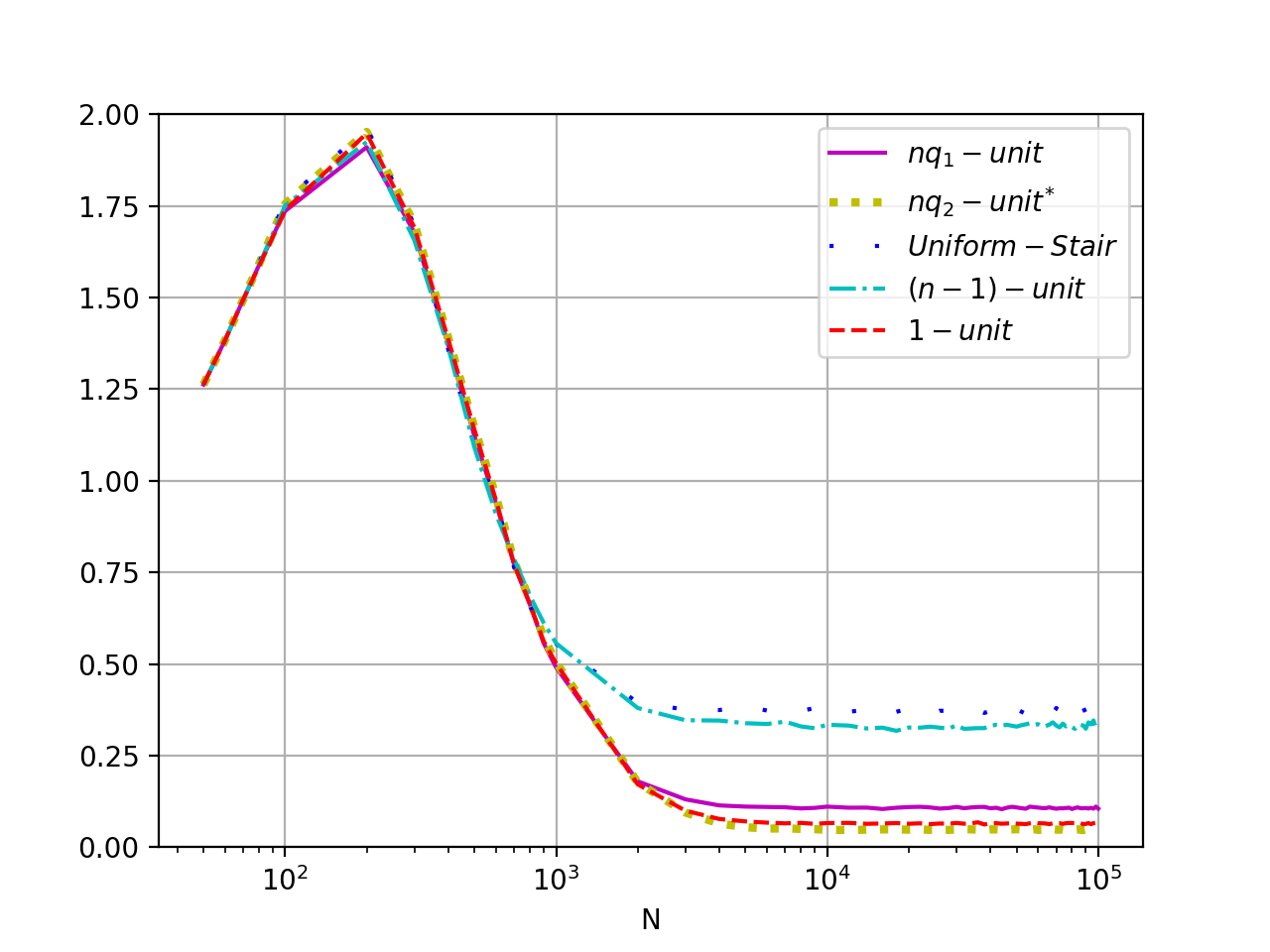} 
    \caption{high supply} 
    \end{subfigure}
  \hfill
  \begin{subfigure}[t]{0.32\textwidth}
    \centering
    \includegraphics[width=\linewidth]{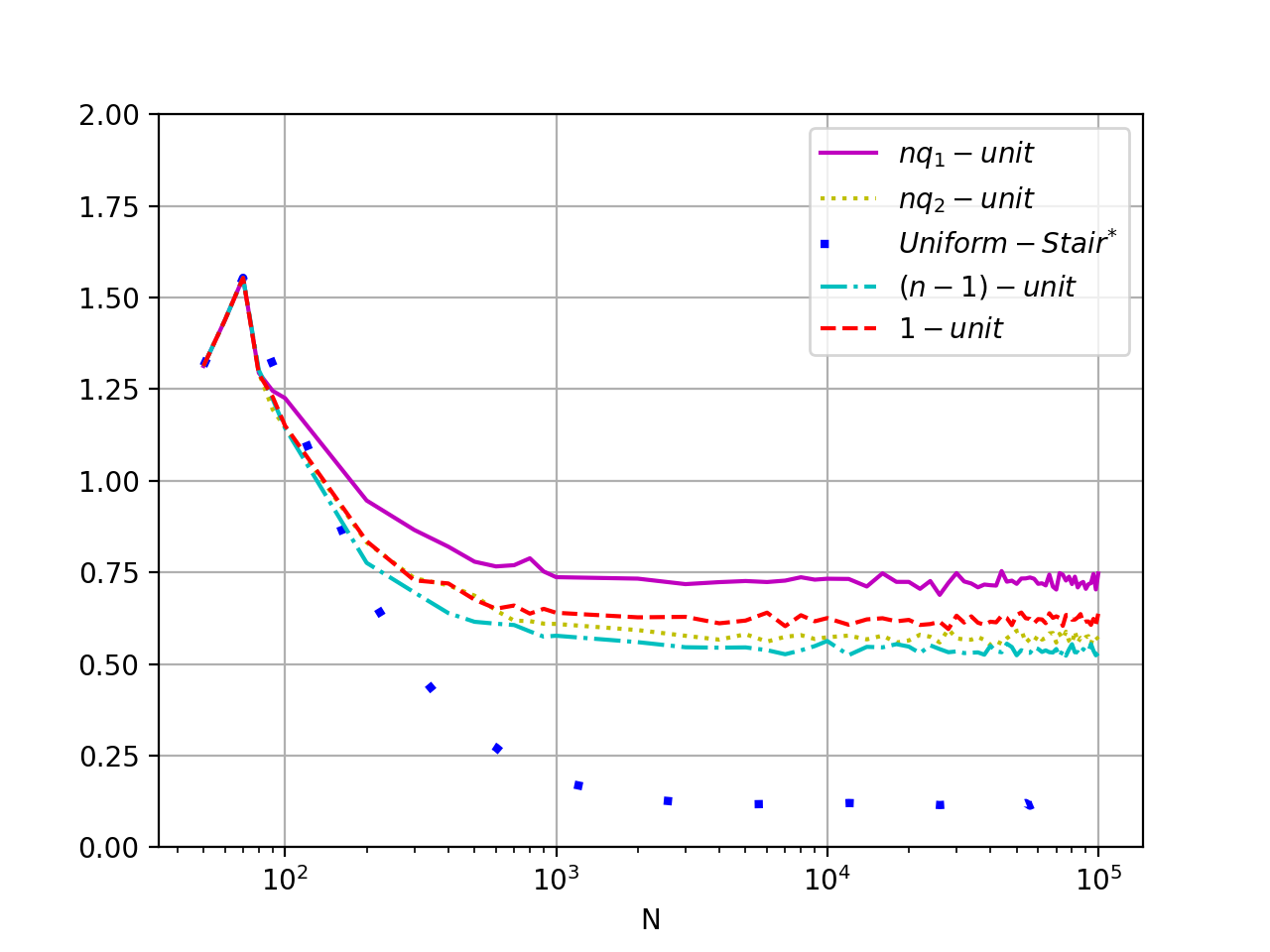} 
    \caption{uniform stair} 
  \end{subfigure}
  \caption{This figure depicts the error normalized by
    $\sqrt{\samples}$ as a function of the sample size $\samples$
    (fixed parameters $n=16$, $\epsilon = 0.001$, $\dist =
    \text{Beta}(2,2)$).  The counterfactual auction B is fixed, and
    the error is depicted as a function of $\samples$ when the data is
    drawn from the bid distribution of the A/B-test mechanism C that
    corresponds to several incumbent mechanisms A. The counterfactual
    mechanism is marked with a ``$*$'' in the key. The thick line in
    each figure shows the error when the counterfactual and the
    incumbent mechanisms are the same, a.k.a., the counterfactual
    error. }
  \label{fig:N}
\end{figure}

In this section we compare and contrast empirical evidence with the
theoretical bound by exploring the dependence of the empirical error
on the $\samples$, the number of samples of the analyst, and $n$, the
number of agents in each auction.  Recall that theoretical dependence
on $\samples$ is $\Theta(\sqrt{1/\samples})$ and on $n$ is $O(n \log
n)$ in the worst-case bound of \Cref{cor1}.

We consider three cases for the counterfactual auction B as the
uniform-stair, the low-supply auction, and the high-supply auction.
We allow auction A to be any of these auctions and also the 1-unit and
$(n-1)$-unit auctions.  We fix the mixing probability of $\eps =
0.001$ and the incumbent mechanism C is $(1-\epsilon)$ A + $\epsilon$
B.  (A subsequent study will explore the role of truncation.)  The
results of these studies are depicted the figures below.  Not shown
here, when mixing in the counterfactual auction with a large
probability $\epsilon$, there is little benefit from truncation.

In \Cref{fig:N} we observe that as a function of the number of samples
$\samples$ the error is indeed the optimal $\Theta(\sqrt{1/\samples})$
rate.  Specifically, we observe that the error in the estimation of
the revenue of a counterfactual auction indeed depends on the
incumbent mechanism as a constant times $\sqrt{1/\samples}$ and this constant
is different for different incumbent mechanisms.  Moreover, unlike the
result obtained by our theoretical bound, this constant is always much
less than 1.  Moreover, this limit behavior is already achieved with $\samples \approx 1000$ bids in the sample.

In \Cref{fig:n} we observe that, within the range where our bound
holds (which requires $\samples > n^2$), the dependence on the number
of agents $n$ is at most slowly increasing and far from the $O(n \log
n)$ worst-case bound.  Considering the whole range, we see that the
dependence varies, and thus precise theoretical analysis may be
difficult.  An important consideration in our choice of counterfactual
auctions is that the per-agent revenue is roughly constant in the
number of agents $n$; thus relative changes in revenue are not
confounding our empirical analysis of the error.  

\begin{figure}[tb]
  \begin{subfigure}[t]{0.32\textwidth}
    \centering
    \includegraphics[width=\linewidth]{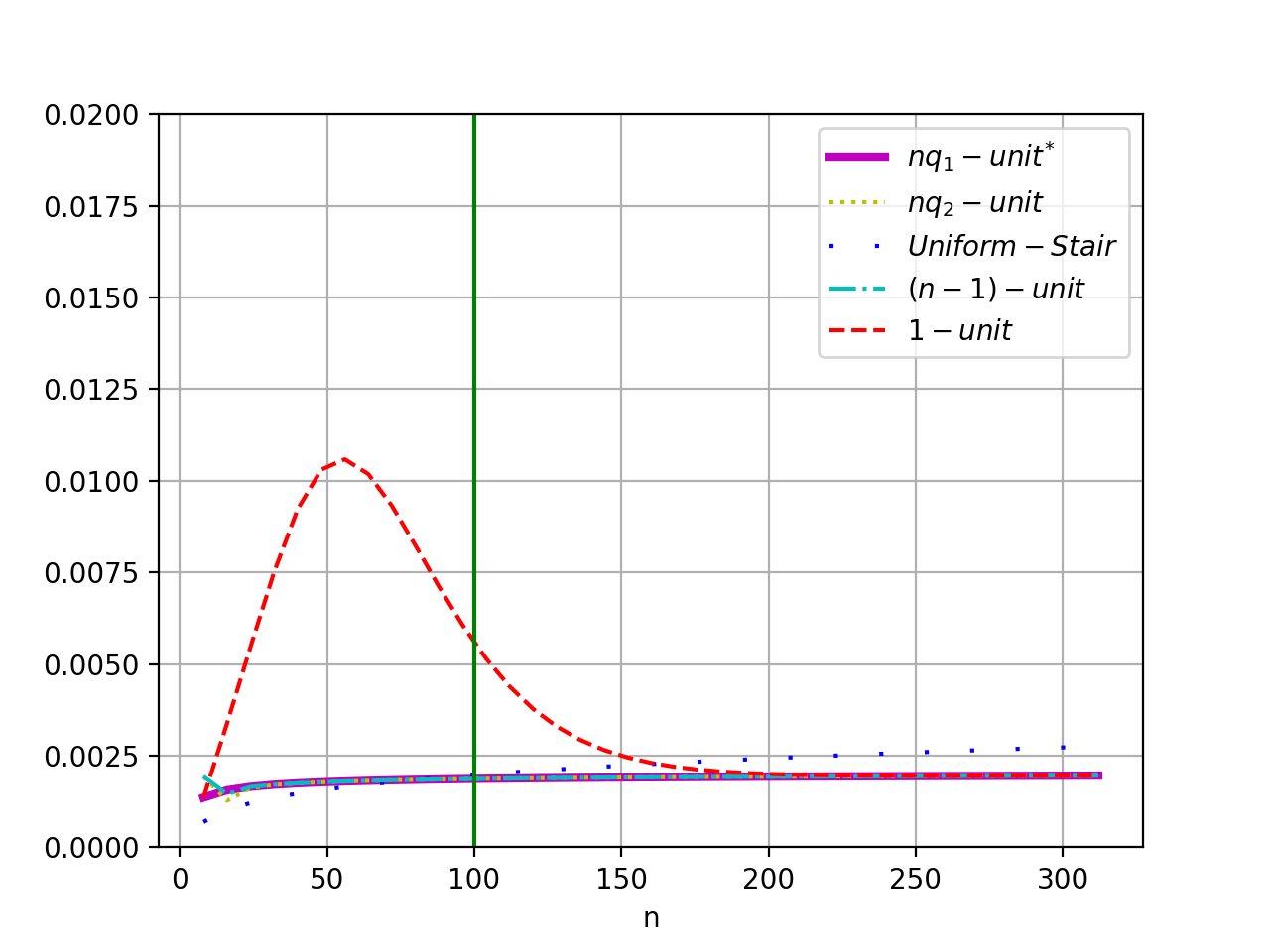} 
    \caption{low supply} 
  \end{subfigure}
  \hfill
  \begin{subfigure}[t]{0.32\textwidth}
    \centering
    \includegraphics[width=\linewidth]{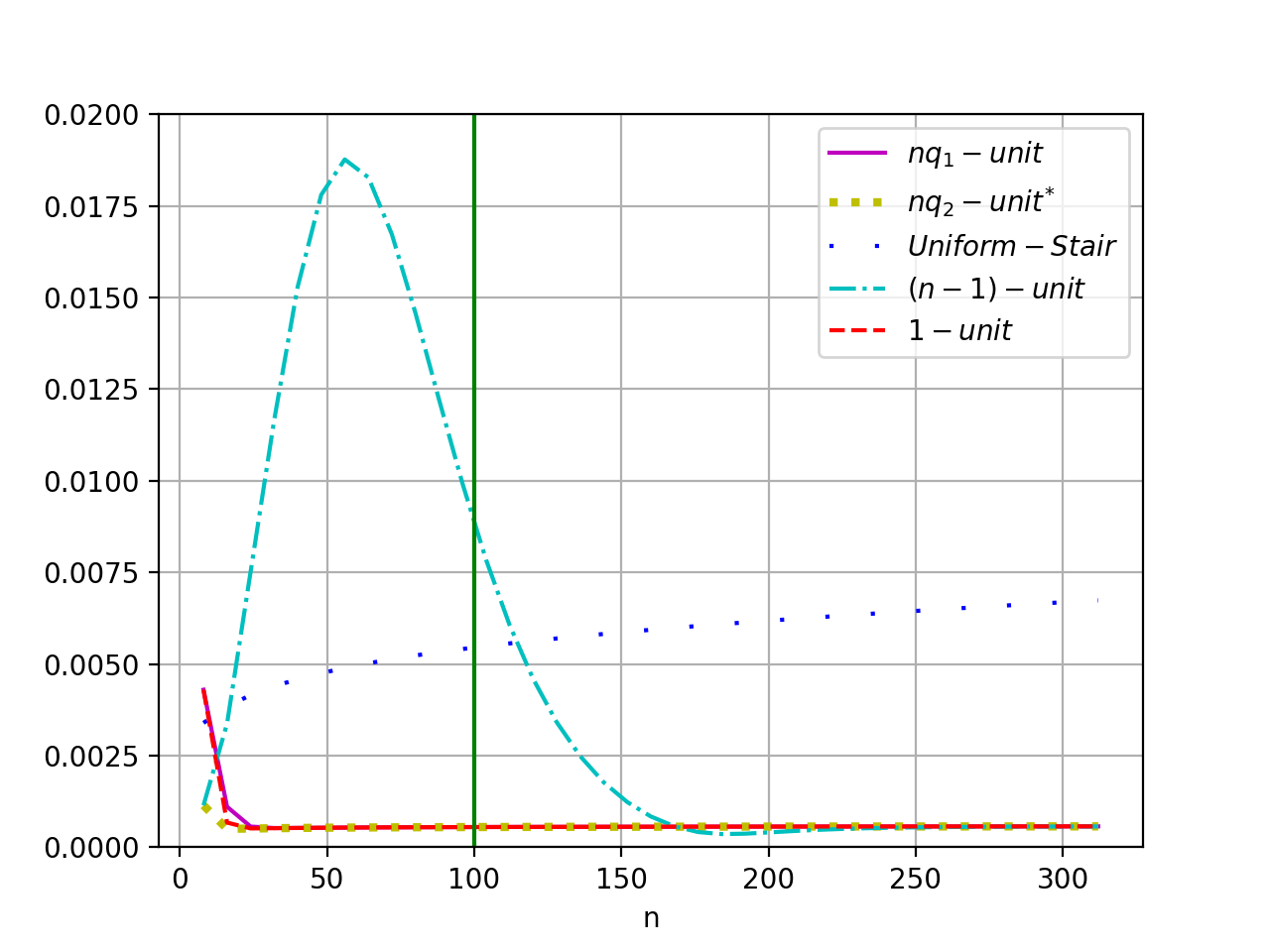} 
    \caption{high supply} 
    \end{subfigure}
  \hfill
  \begin{subfigure}[t]{0.32\textwidth}
    \centering
    \includegraphics[width=\linewidth]{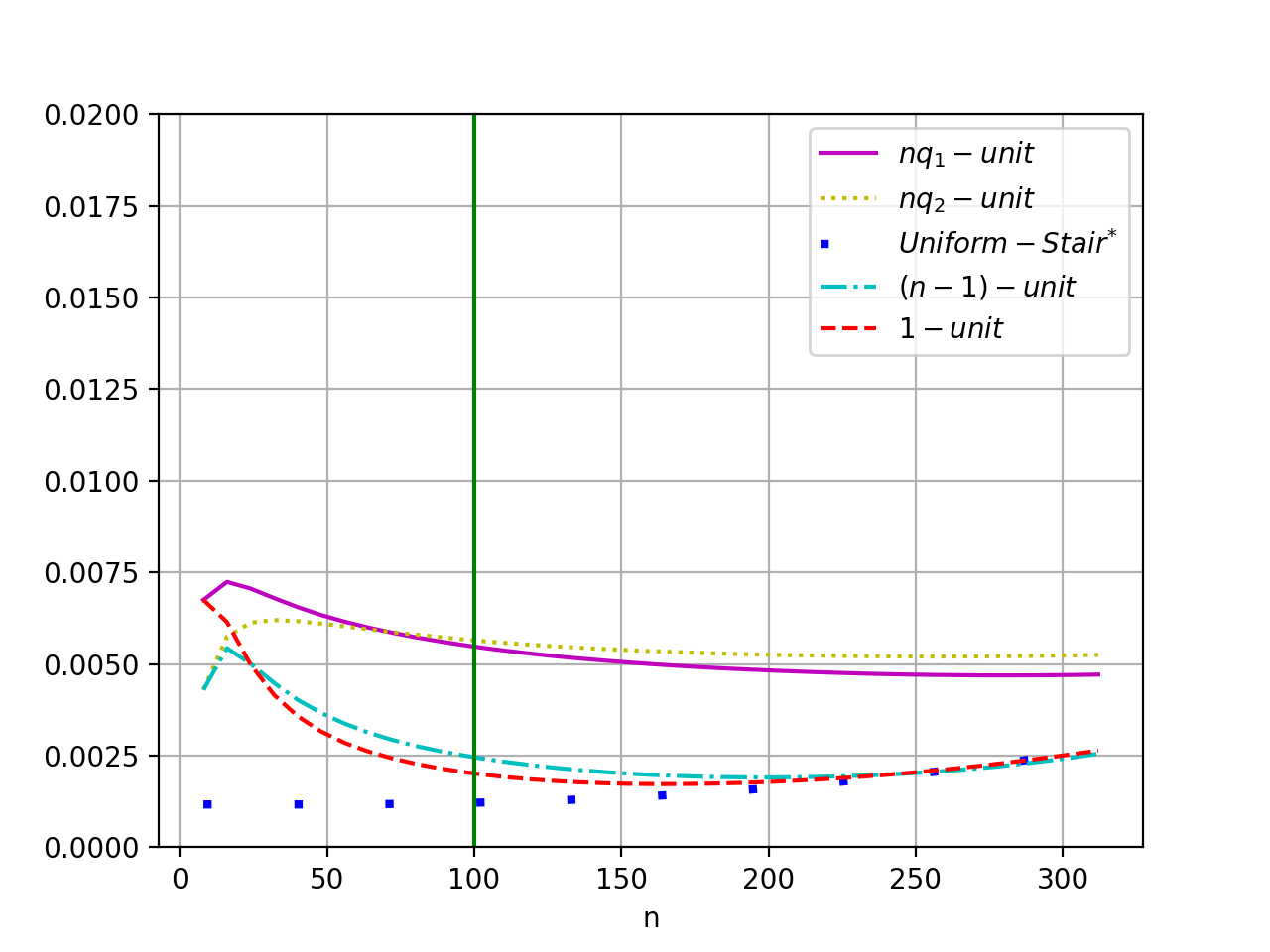} 
    \caption{uniform stair}
    \label{fig:n-unif}
  \end{subfigure}
  \caption{This figure depicts the error as a function of the number
    of agents $n$ (fixed parameters $\samples=10000$, $\epsilon =
    0.001$, $\dist = \text{Beta}(2,2)$).  The counterfactual auction B
    is fixed, and the error is depicted as a function of $n$
    when the data is drawn from the bid distribution of the A/B-test
    mechanism C that corresponds to several incumbent mechanisms A.
    The thick line in
    each figure shows the error when the counterfactual and the
    incumbent are the same, a.k.a., the counterfactual error. The vertical line
    corresponds to $n^2=
    \samples$.}
  \label{fig:n}
\end{figure}

We observe in these empirical results that the error is not increasing
at the same rate as our bounds suggest.  The theoretical bounds are
symmetric with respect to swapping the high supply and low supply
auction.  The empirical results are all better than these bounds:
\begin{itemize}
\item uniform-stair incumbent; low-supply counterfactual: $O(\sqrt{n} \log(1/\epsilon))$

\item high-supply incumbent; low-supply counterfactual: $O(\sqrt{n} \log (n/\epsilon))$
\item low-supply incumbent; uniform-stair counterfactual: $O(\log n)$.
\end{itemize}
In the above bounds, we could swap the low-supply auction for the
$1$-unit auction and the high-supply auction for the 1-unit auction
and the bounds by replacing the $\sqrt{n}$ term with an $n$ term.

One perhaps unexpected outcome that is present in these empirical
results is the large error for small $n$ and similar incumbent and
counterfactual mechanisms, specifically, the low-supply counterfactual
with the 1-unit incumbent or the high-supply counterfactual with the
$(n-1)$-unit incumbent.  These auctions have allocation rules that are
near zero for low values, near 1 for high values, and at some point in
between transitioning from 0 to 1 (see \Cref{fig:five-alloc-rules}).
For small $n$ these transitions overlap and this results in higher
error.  For large $n$ these transitions do not overlap and the error
is small.

As a final note, in \Cref{fig:n-unif} the error is trending upwards
with large $n$ even when the counterfactual and incumbent are
identical.  This trend is from truncation which is increasing with $n$ (relative to the 
fixed sample size $\samples$)
at these parameter settings.

\subsection{Ex Ante Error Control (A.k.a., Instrumentation)}

A main focus of this paper is identifying properties of auctions that
make them good for inference.  Specifically, our A/B-testing method
suggests that better estimates of the revenue of mechanism B are
possible by running mechanism C that mixes in B with mechanism A.  See
\Cref{fig:eps}.  These empirical result should be compared with
\Cref{thm:allpay-general} which gives the dependence on $\epsilon$ as
$O(\log 1/\epsilon)$ when $\samples > 1/\epsilon$, i.e., to the left
of the vertical dashed line. (When $\samples > 1/\epsilon$ the
theoretical bound of \Cref{thm:allpay-simple} with term $\log \samples
/ \sqrt \samples$ becomes the better bound.)  In the relevant region
we observe that the dependence on $1/\epsilon$ is sub-logarithmic
except when the counterfactual mechanism is the uniform-stair auction
where the dependence is $\Theta(\log 1/\epsilon)$.

\begin{figure}
  \begin{subfigure}[t]{0.48\textwidth}
    \centering
    \includegraphics[width=\linewidth]{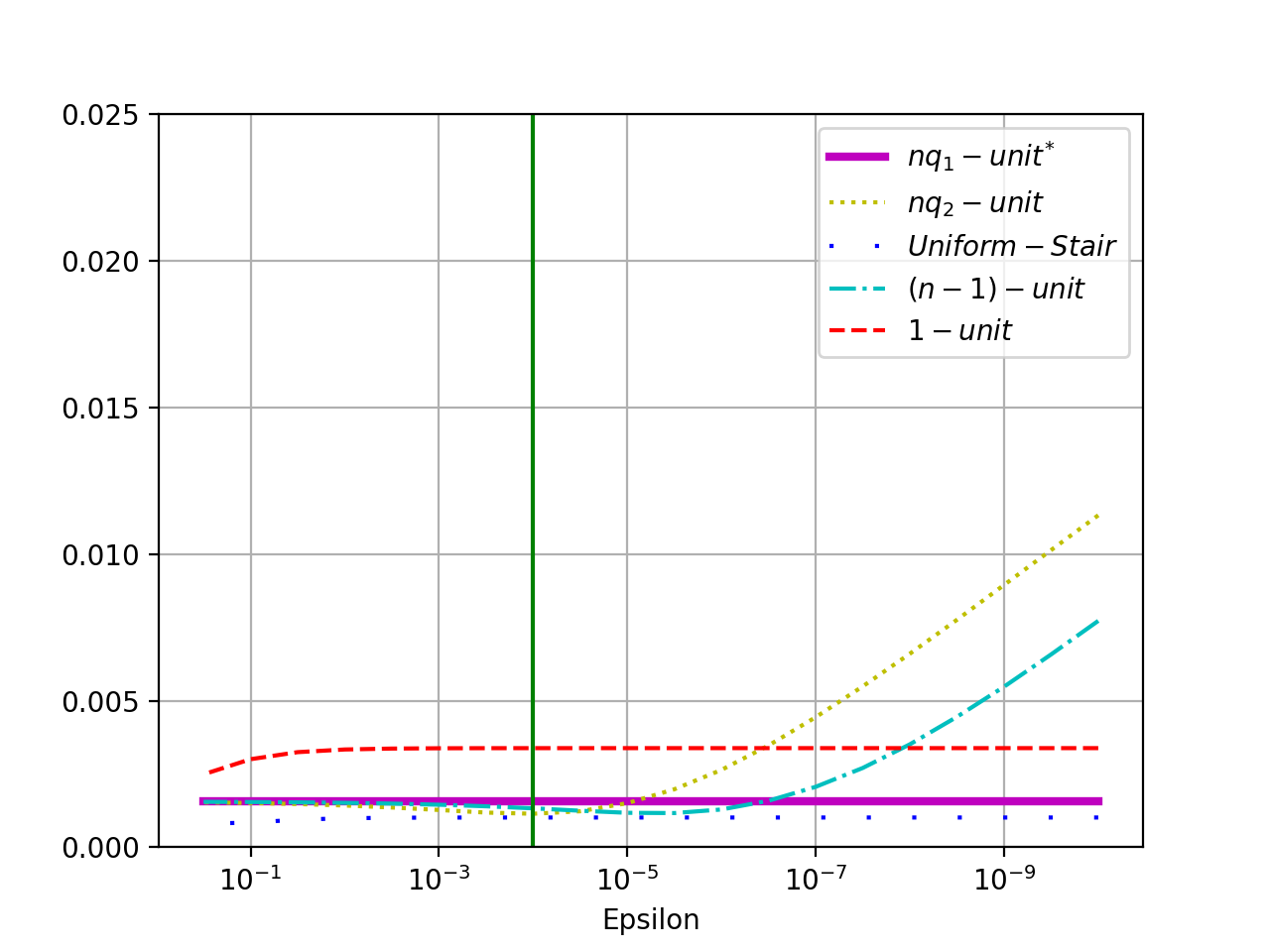} 
    \caption{A/B-test, low supply} 
  \end{subfigure}
  \hfill
    \begin{subfigure}[t]{0.48\textwidth}
    \centering
    \includegraphics[width=\linewidth]{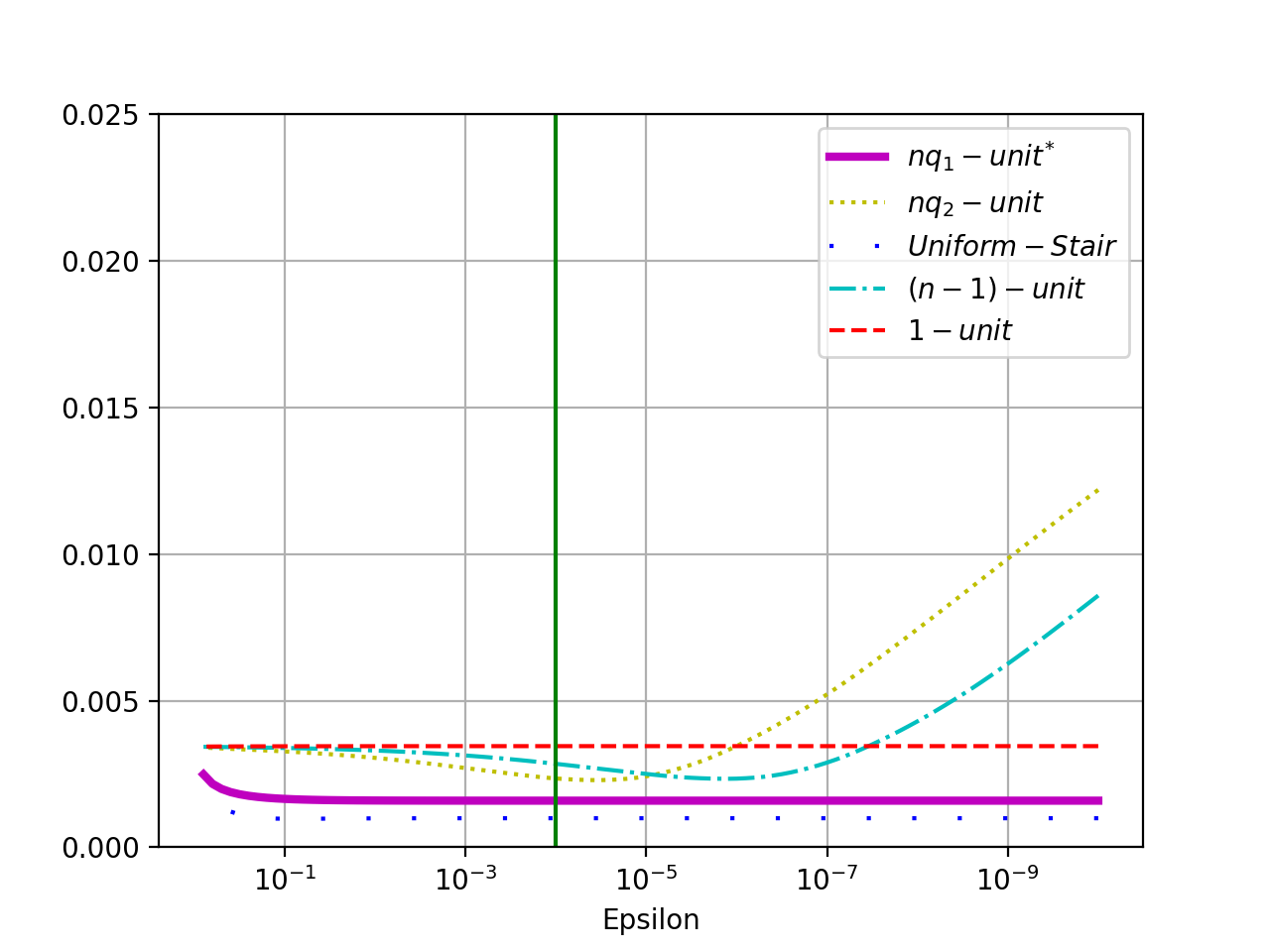} 
    \caption{universal B-test, low supply}
  \end{subfigure}
  \\
  \begin{subfigure}[t]{0.48\textwidth}
    \centering
    \includegraphics[width=\linewidth]{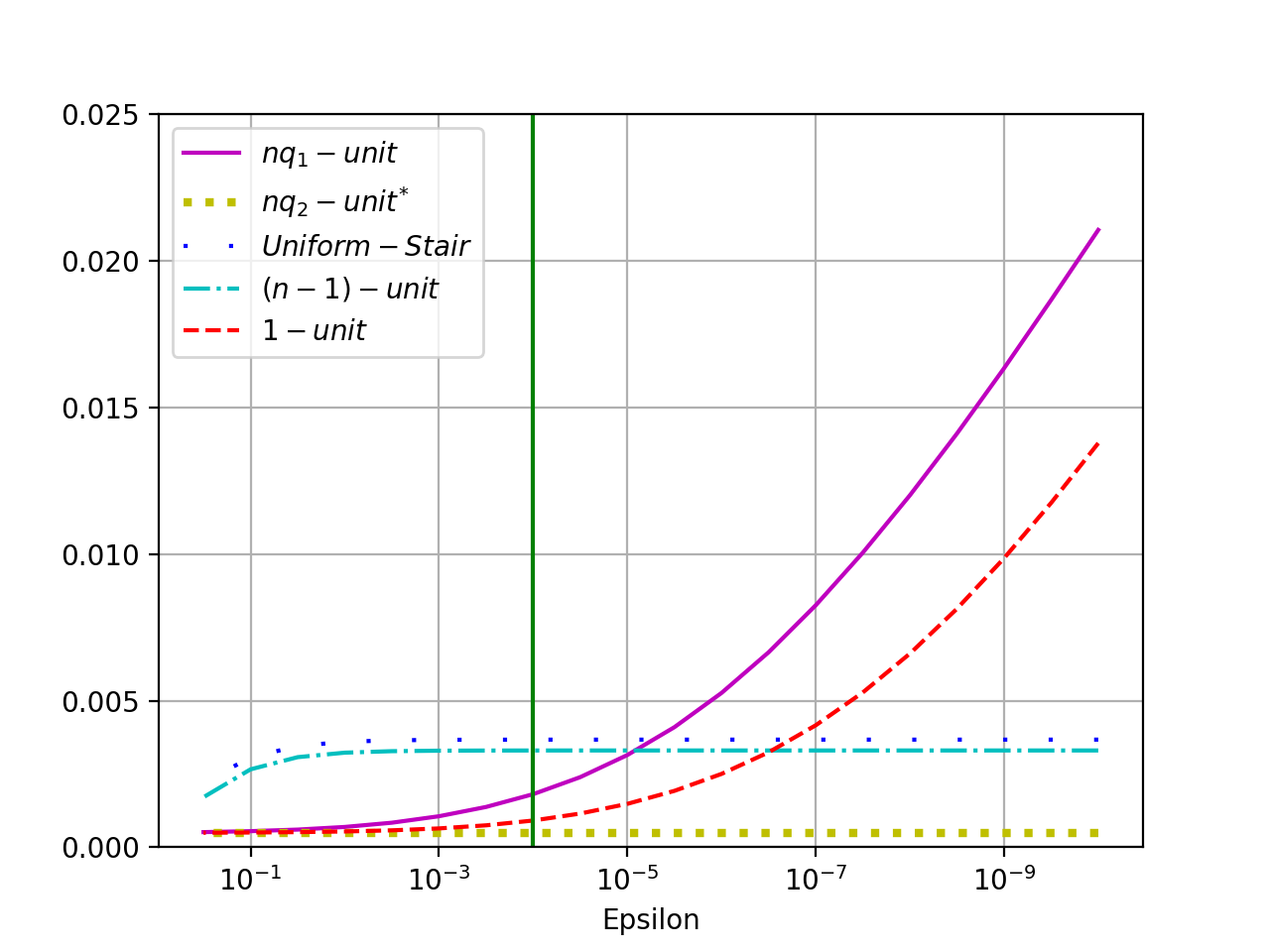} 
    \caption{A/B-test, high supply}
    \end{subfigure}
  \hfill
    \begin{subfigure}[t]{0.48\textwidth}
    \centering
    \includegraphics[width=\linewidth]{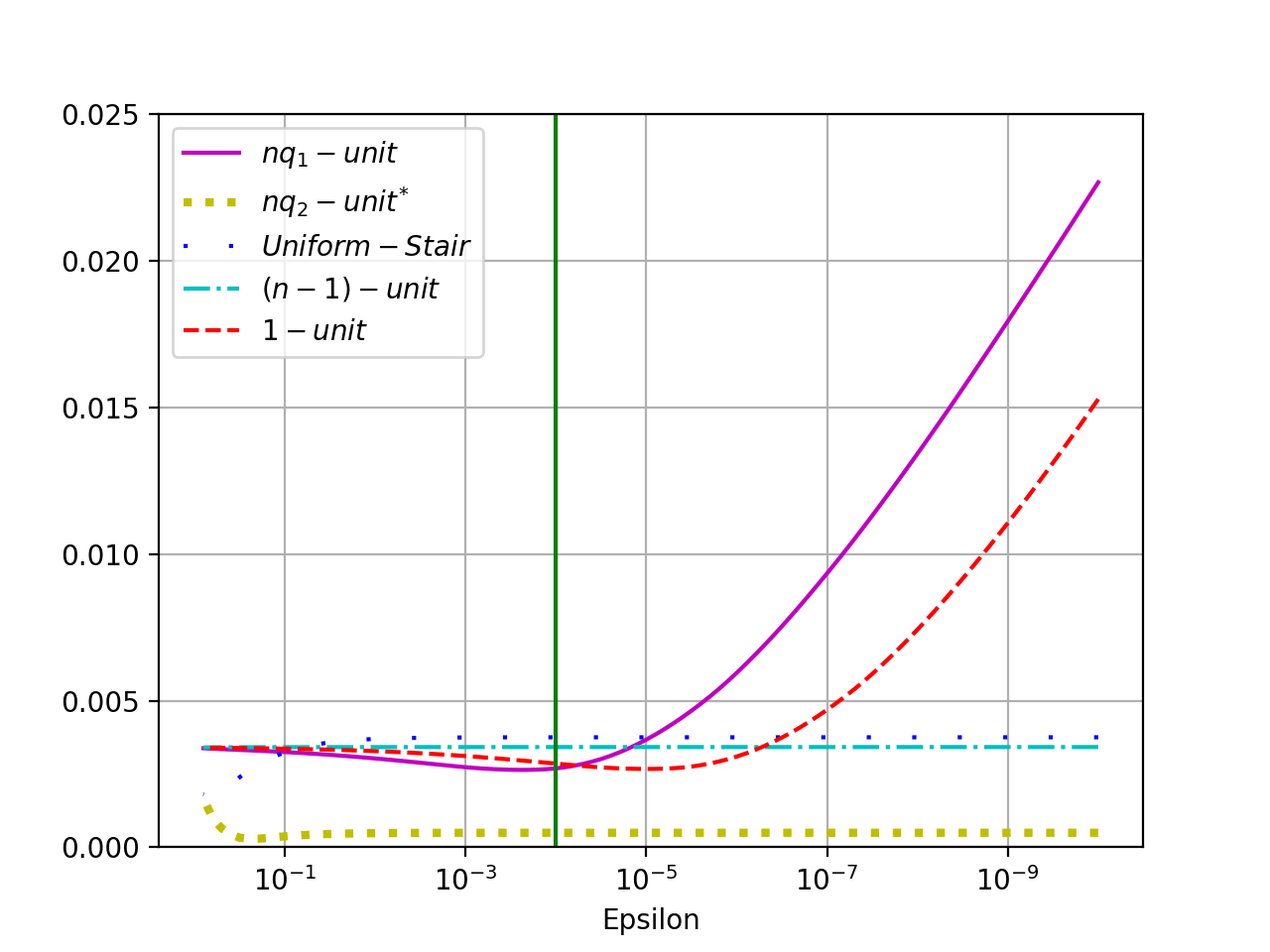} 
    \caption{universal B-test, high supply}
    \end{subfigure}
  \\
  \begin{subfigure}[t]{0.48\textwidth}
    \centering
    \includegraphics[width=\linewidth]{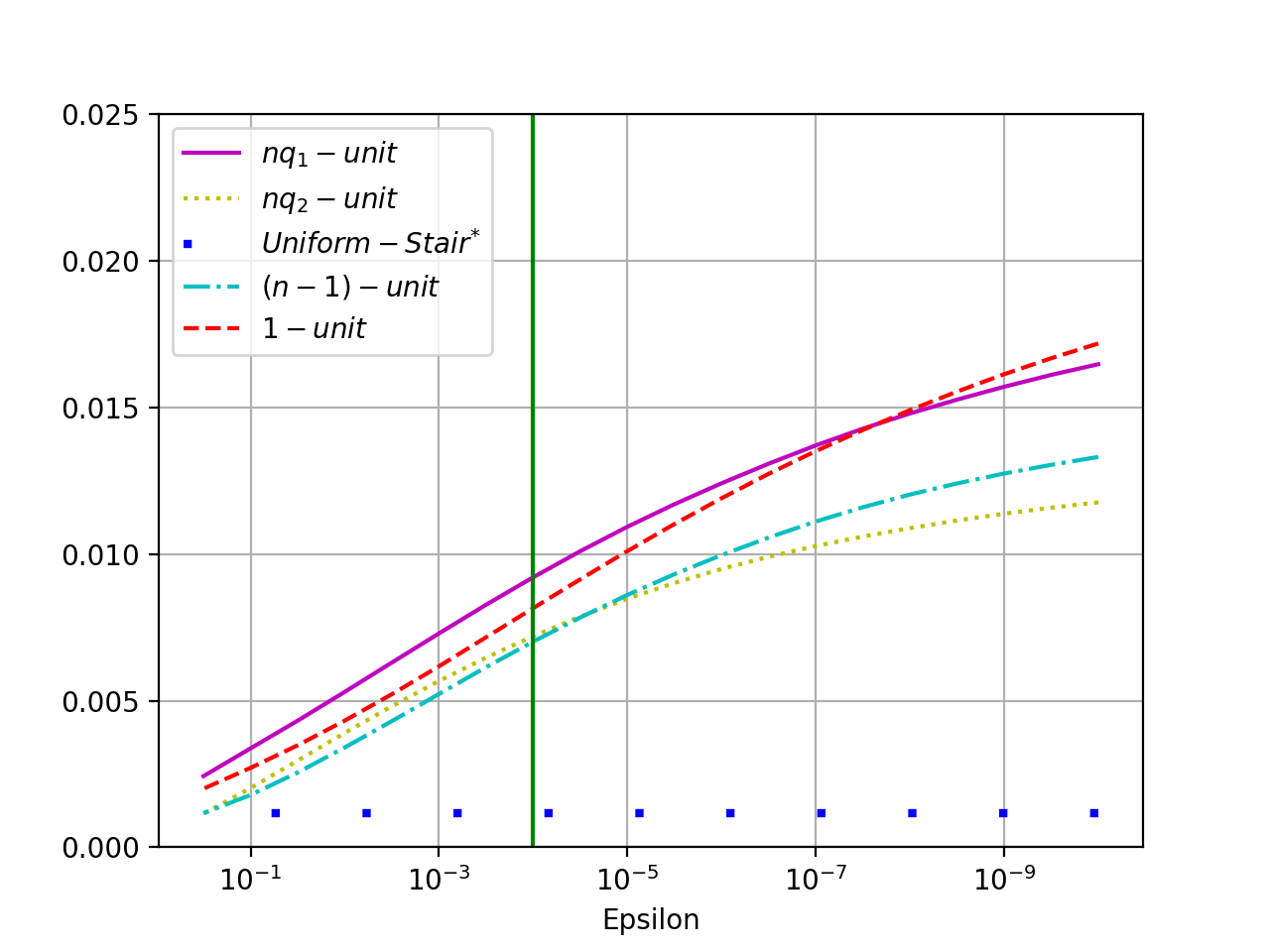} 
    \caption{A/B-test, uniform stair} 
    \label{fig:eps:unif}
  \end{subfigure}
  \hfill
  \begin{subfigure}[t]{0.48\textwidth}
    \centering
    \includegraphics[width=\linewidth]{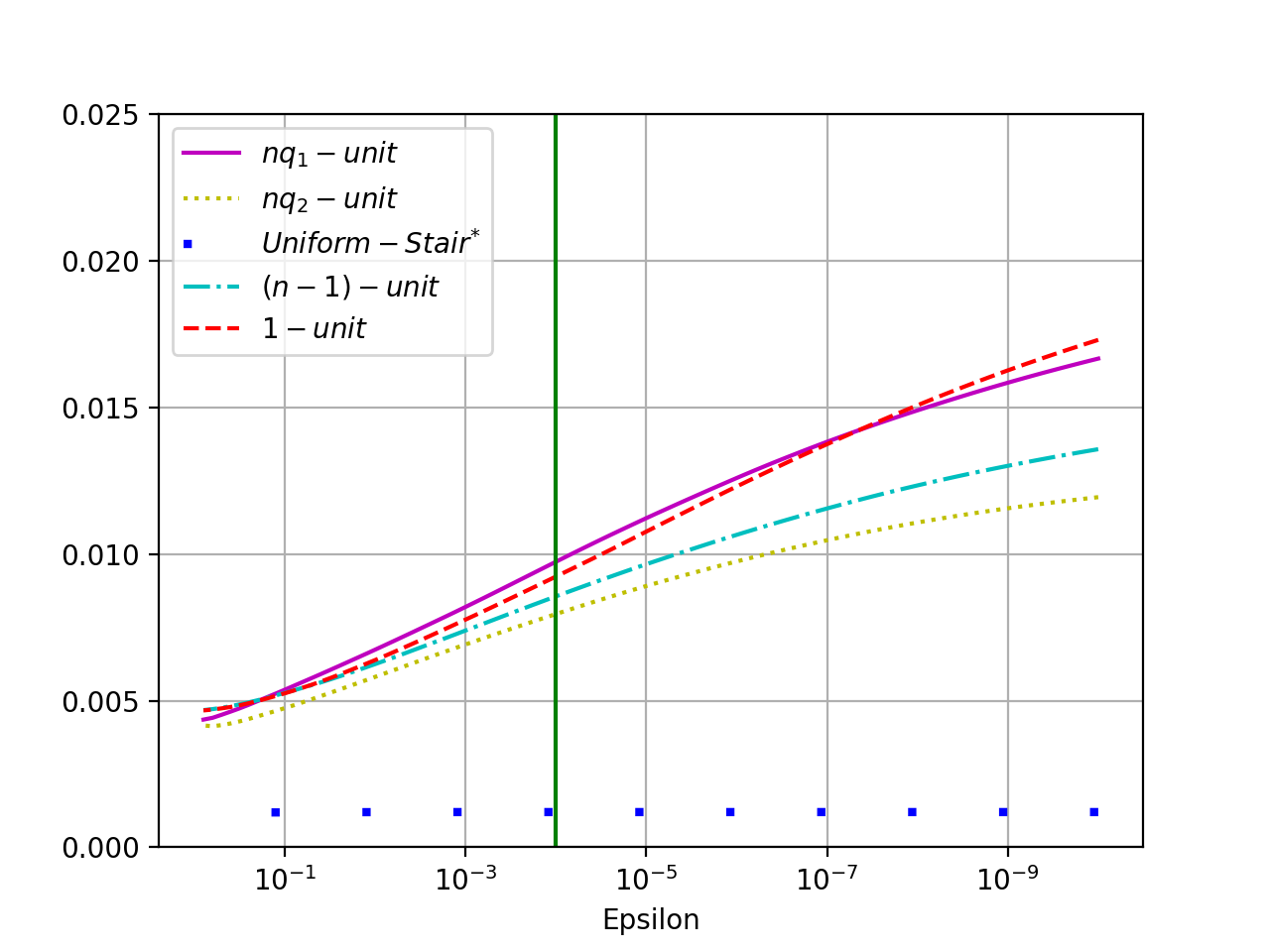} 
    \caption{universal B-test, uniform stair}
  \end{subfigure}
  \caption{This figure depicts the error as a function of the A/B-test
    mixing probability $\epsilon$ (fixed parameters $\samples=10000$,
    $n = 16$, $\dist = \text{Beta}(2,2)$).  The counterfactual auction
    B is fixed, and the error is depicted as a function of $\epsilon$.
    The left column is the A/B test where the counterfactual is mixed
    in with probability $\epsilon$.  The right column is the universal
    $B$ test where a mix of the 1-unit and $(n-1)$-unit auction is
    mixed in with probability $\epsilon$. The vertical
    line corresponds to 
    $\epsilon=1/\samples.$}
  \label{fig:eps}
\end{figure}

As described in \Cref{s:universal}, there is a universal B-test
mechanism which is a mixture of the $1$-unit auction and the
$(n-1)$-unit auction.  Mixing this auction with any other position auction makes it
possible to infer the revenue of that position auction.  See
\Cref{fig:eps}.  Comparing the A/B-test with a universal B-test
empirically we see that there is not much improvement from the
A/B-test.  As described in preceding sections, the benefit of the
universal B-test is that instrumentation with it makes it possible to
infer the performance of any other position auction.

\subsection{Ex Post Error Control}

The classical econometric approach to estimation in auctions is to
use a consistent estimator for the distribution of values of bidders and then to estimate
revenue in a counterfactual auction from this distribution.
To obtain a consistent estimator for the distribution of values, the derivative of the bid
function needs to be estimated and error of this estimator is typically
controlled by smoothing, i.e., averaging bids with adjacent bids in
the sorted order.  In contrast, the estimators of this paper do not
employ smoothing of the bid distribution, and instead errors are
controlled by truncation, i.e., zeroing out the contribution to the
estimated revenue from potentially-ill-behaved extremal quantiles of
the bid distribution.

This section compares these approaches to controlling error and makes
several empirical findings.  First, we observe that ex post methods
for error control are necessary in some scenarios.  This observation
comes from comparing the truncated estimator described above with the
same estimator with no truncation.  We see in \Cref{fig:truncation}
that when the mechanisms are extreme and opposite that our estimator
with truncation has low error while without truncation the error is
generally worse than the trivial bound.\footnote{We have assumed
  values to be bounded on $[0,1]$; thus, the per-agent revenue is at
  most 1 and error bounds that exceed 1 are trivial.}

\begin{figure}
  \begin{subfigure}[t]{0.48\textwidth}
    \centering
    \includegraphics[width=\linewidth]{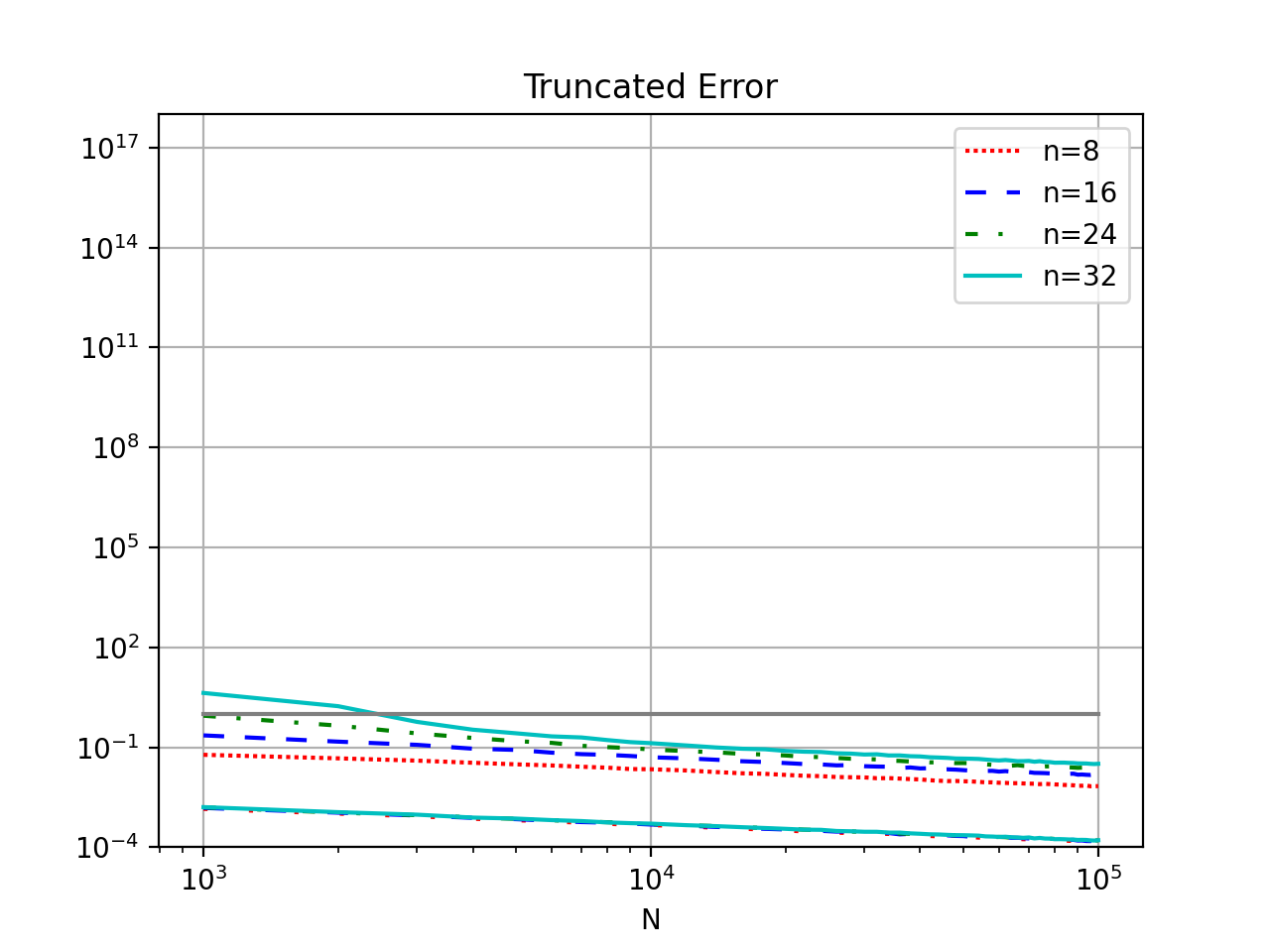} 
    \caption{truncation} 
  \end{subfigure}
  \hfill
  \begin{subfigure}[t]{0.48\textwidth}
    \centering
    \includegraphics[width=\linewidth]{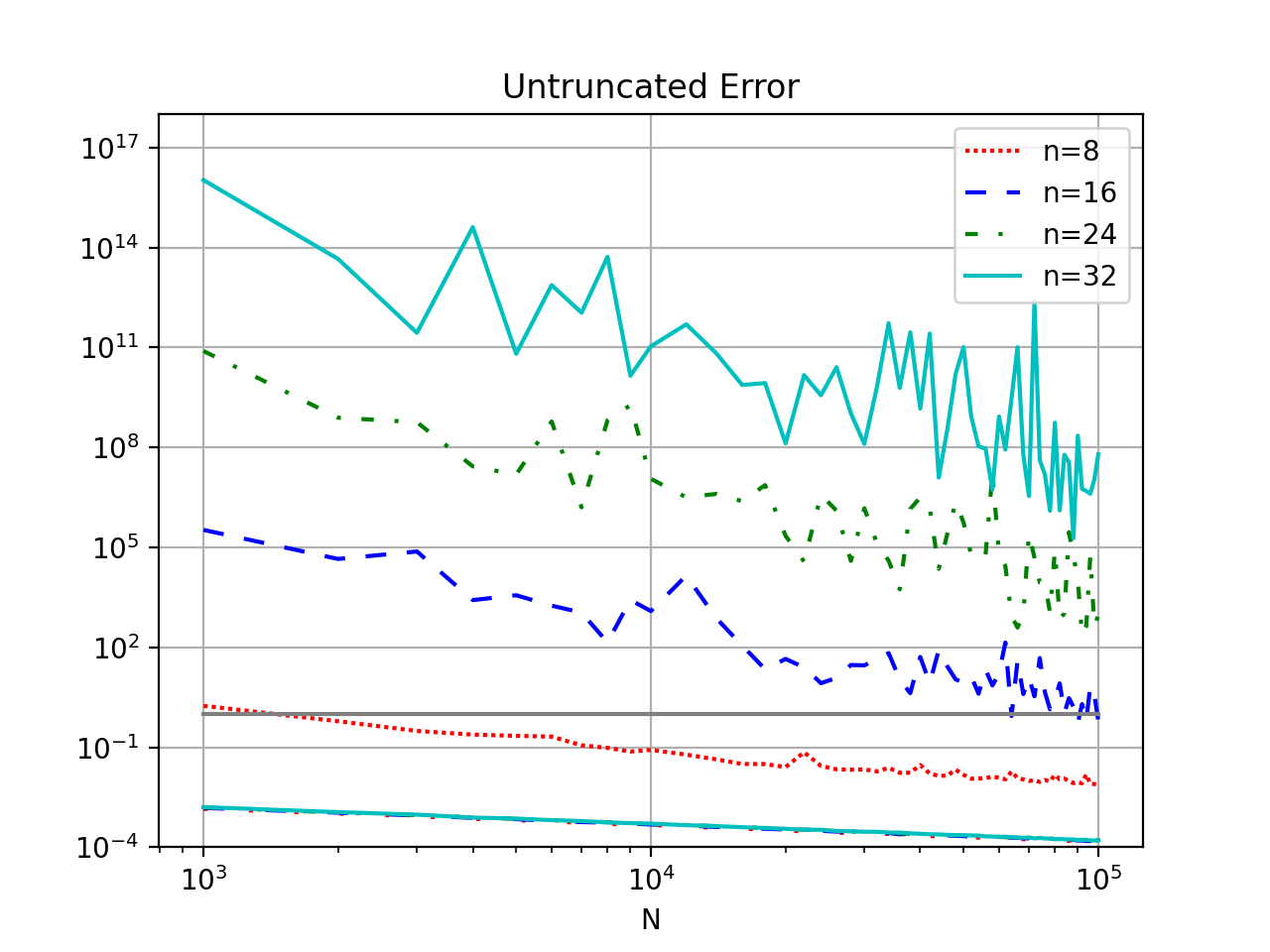} 
    \caption{no truncation} 
  \end{subfigure}
  \caption{This figure depicts on a log-log scale the error as a
    function of the number of samples $\samples$ for several choices
    of the number of agents $n$ (fixed parameters: $\dist =
    \text{Beta}(2,2)$).  The counterfactual auction B is the high
    supply auction, the incumbent auction C is the low supply auction
    (the incumbent does not mix in B).  The trivial error bound of 1
    is depicted with a solid line.  The thick line in each figure
    shows the error when the counterfactual and the incumbent are the
    same, a.k.a., the counterfactual error.}
  \label{fig:truncation}
\end{figure}

As we have discussed previously, when the B-test probability
$\epsilon$ in an A/B-test is large, truncation has limited benefit.
Indeed, \Cref{fig:eps} shows error as a function of $\epsilon$ for the
truncated estimator.  In fact, the same plots result from the
untruncated estimator.  Nonetheless, when we consider very small
$\epsilon$, the truncated estimator gives a non-trivial error, while
the untruncated estimator does not.  This comparison is depicted in
\Cref{fig:eps-truncation}.  The constant-in-$\epsilon$ bound on the truncated estimator is guaranteed by \Cref{thm:allpay-simple}.

\begin{figure}
  \begin{subfigure}[t]{0.48\textwidth}
    \centering
    \includegraphics[width=\linewidth]{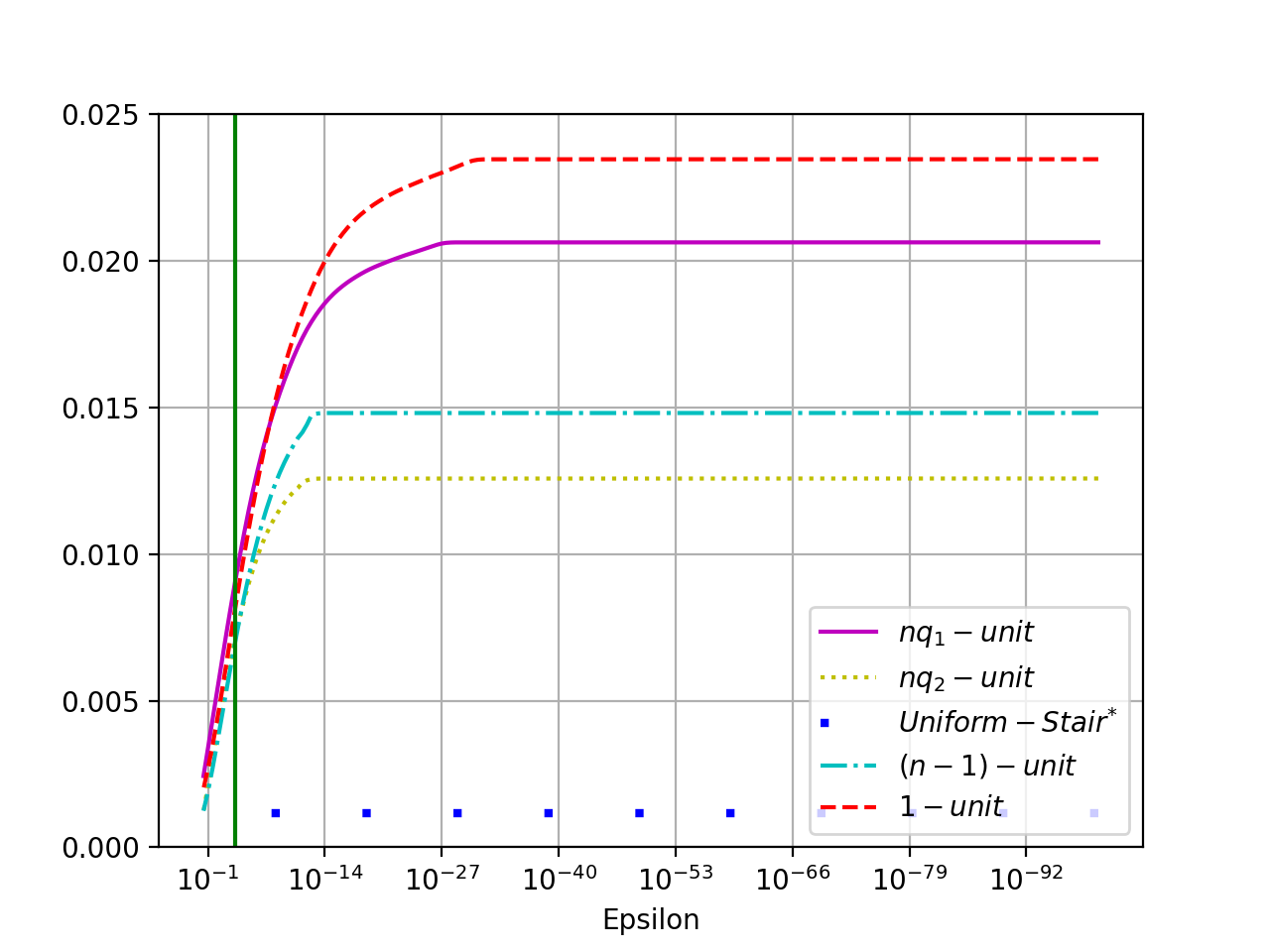} 
    \caption{truncation} 
  \end{subfigure}
  \hfill
  \begin{subfigure}[t]{0.48\textwidth}
    \centering
    \includegraphics[width=\linewidth]{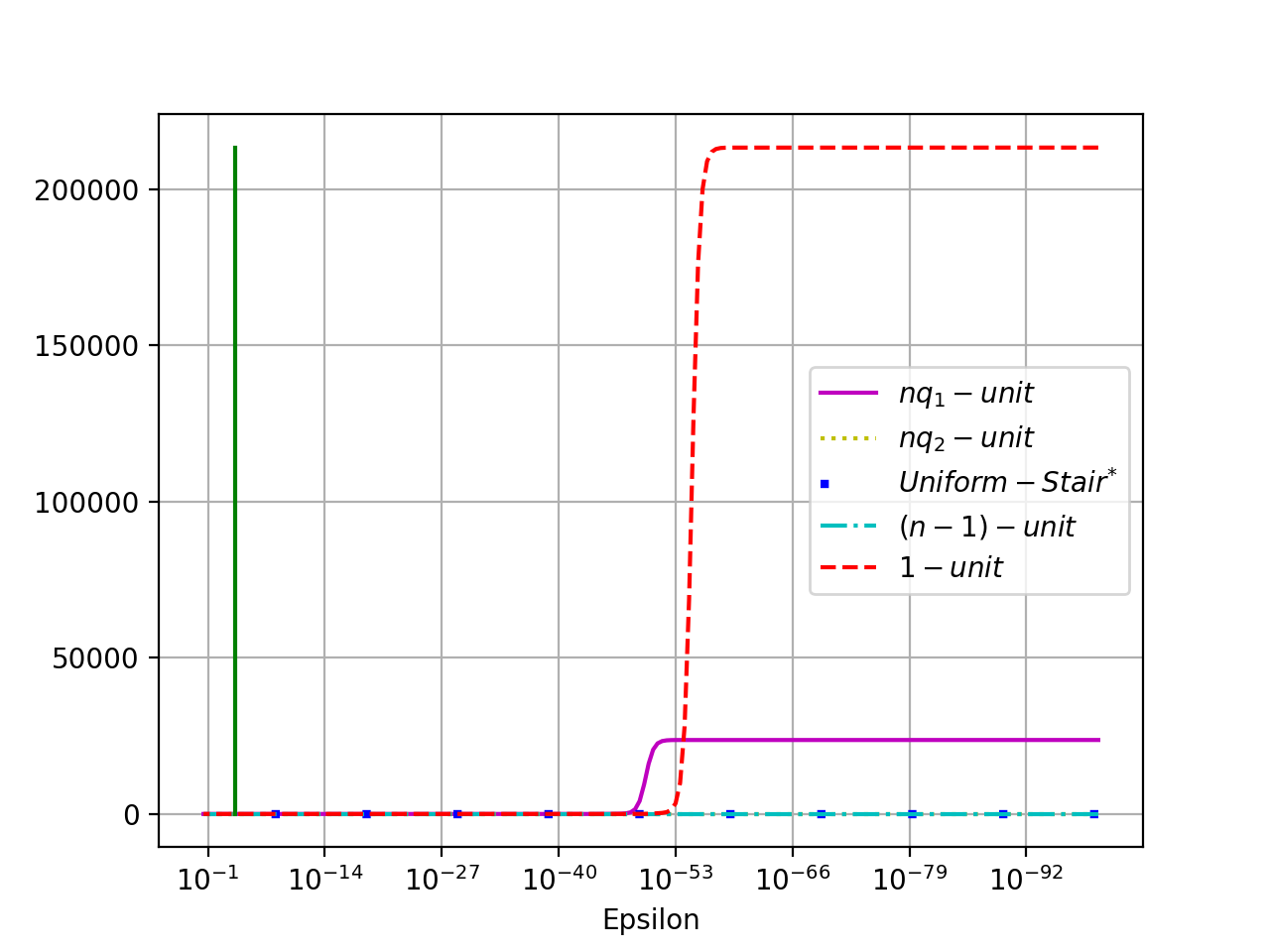} 
    \caption{no truncation} 
  \end{subfigure}
  \caption{This figure depicts the error as a function of the A/B-test mixing probability $\epsilon$ (fixed parameters $\samples=10000$, $n = 16$, $\dist = \text{Beta}(2,2)$).  The counterfactual auction B
    is fixed, and the error is depicted as a function of $\epsilon$
    when the data is drawn from the bid distribution of the A/B-test
    mechanism C that corresponds to several incumbent mechanisms A.
    The thick line in each figure
    shows the error when the counterfactual and the incumbent are the
    same, a.k.a., the counterfactual error. Note that the y-axes in the two figures are different.}
  \label{fig:eps-truncation}
\end{figure}

The truncation we use zeros out the contribution to the estimator
from extreme quantiles.  The truncation parameter does not depend on
fundamentals of the environment and instead was selected to integrate
with theoretical guarantees from statistics.  To show that this
statistically-motivated choice is good, we empirically evaluate the
extent to which other truncation parameters give better error.  We
consider a counterfactual of the $(n-1)$-unit auction, an incumbent of
the 1-unit auction, and various small numbers of agents $n$
(\Cref{fig:optimaltrunc}).  We see that for three selected auction sizes our
truncation has at most four times the error of the optimal truncation;
moreover, for a broad range
of sample size $\samples$ our truncation is at most 50\% worse than the
optimal truncation.

\begin{figure}
  \begin{subfigure}[t]{0.48\textwidth}
    \centering
    \includegraphics[width=\linewidth]{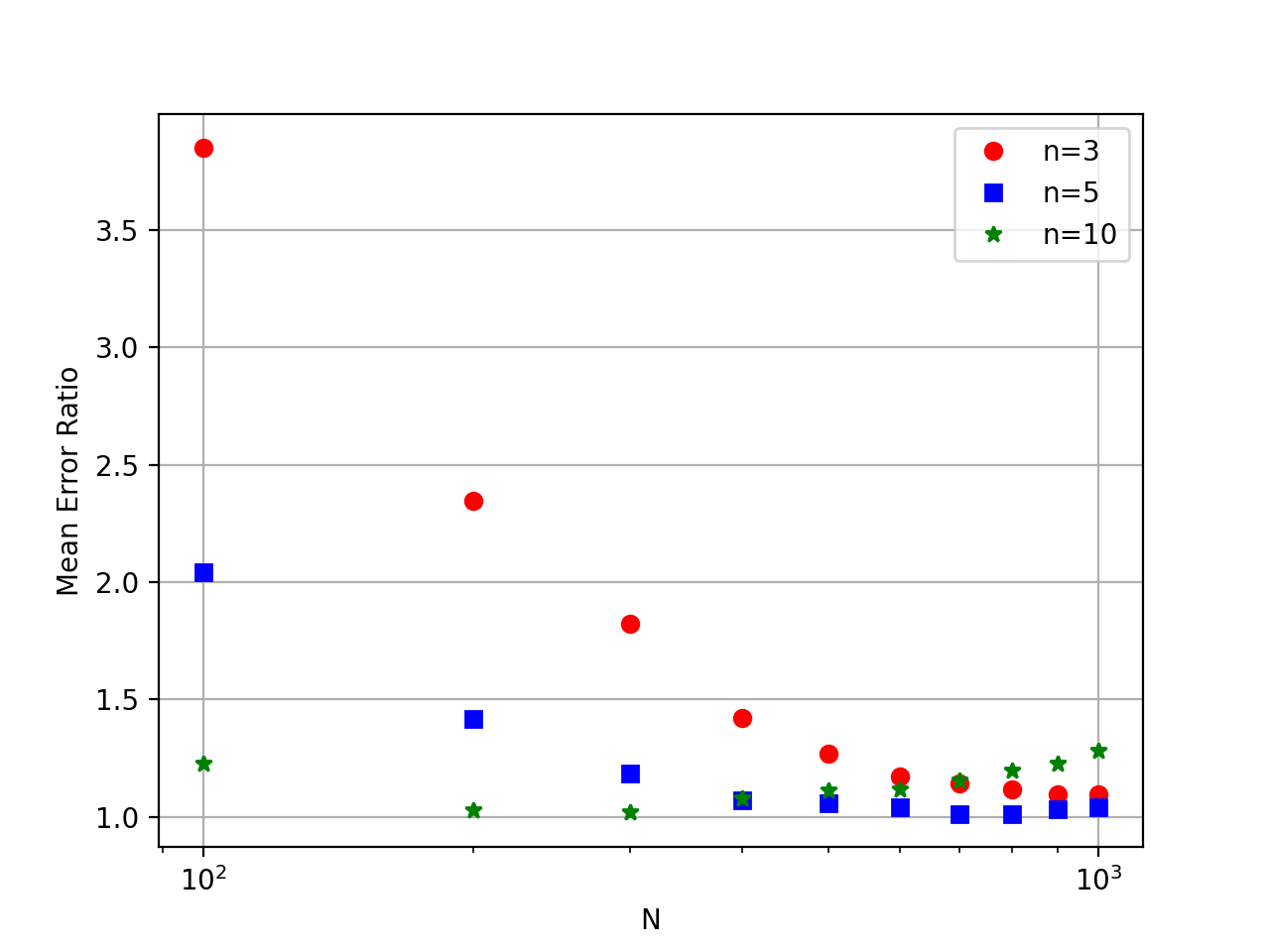} 
    \caption{truncation}
    \label{fig:optimaltrunc}
  \end{subfigure}
  \hfill
  \begin{subfigure}[t]{0.48\textwidth}
    \centering
    \includegraphics[width=\linewidth]{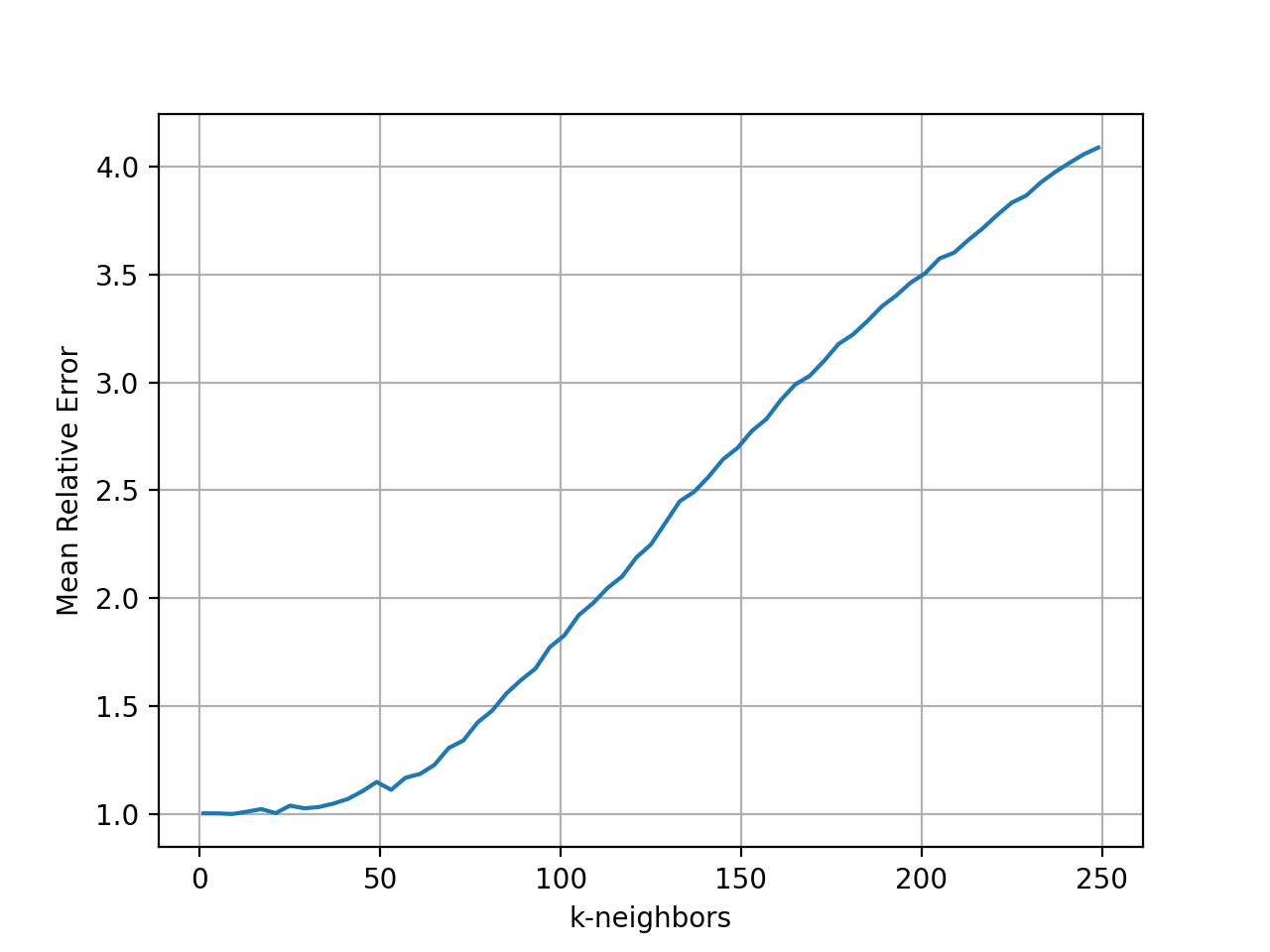} 
    \caption{smoothing}
    \label{fig:optimalsmooth}
  \end{subfigure}
  \caption{On the left, the ratio of our selected truncation parameter
    is compared to the optimal truncation parameter and the ratio of
    the errors is reported.  Ratios closer to 1 show that our
    truncation parameter is nearly optimal.  On the right, the error
    as a function of the number of adjacent bids that are smoothed in
    the classical approach to controlling value estimation error in
    auctions.  It is optimal not to smooth when using value estimates
    with the plug-in estimator for revenue.  On the right, fixed
    parameters are $N = 1000$, $n = 5$, $\dist = \text{Beta}(2,2)$.
    The counterfactual auction B is the $(n-1)$-unit auction, the
    incumbent auction A is the 1-unit auction (the incumbent does not
    mix in B).}
\end{figure}

We compare truncation to the classical approach for controlling error
in auctions which is smoothing the bid distribution.  We consider a
natural approach smoothing, namely averaging for each bid the $k$
adjacent bids in the sorted order.  The classical approach, which asks
for a uniform bound on the error in estimates of values to plug into
the revenue estimator, would tune $k$ depending on properties of the
bid distribution (which is endogenous to the environment).  Here we
show that with the plug-in estimator the optimal smoothing is no
smoothing.  Thus, estimation 
of the bid distribution via smoothing
used by the classical approach of controlling error is
unhelpful for estimating revenue.  See \Cref{fig:optimalsmooth}.
Meanwhile, as we have seen, truncation both controls error and does
not require tuning to endogenous properties of the bid distribution.

\subsection{Distribution Robustness}

We have experimented with a number of distributions over values and the
qualitative results observed above continue to hold.  Here we repeat
the study of truncation with value distributions intended to stress
the estimation procedure.  We observe that there are no significant
changes (\Cref{fig:dist-trunc}).  The distributions considered are the
equal revenue distribution on interval [0.1,1], the uniform
distribution on interval [0.3,1] and a bimodal distribution.  For
example our interest in the bimodal distribution is that the
low-supply auction and high-supply auctions have revenue driven from
different modes.  We were unable to identify any distribution that
resulted in significantly different outcomes from what we observed for
the beta distribution.

\begin{figure}
  \begin{subfigure}[t]{0.48\textwidth}
    \centering
    \includegraphics[width=\linewidth]{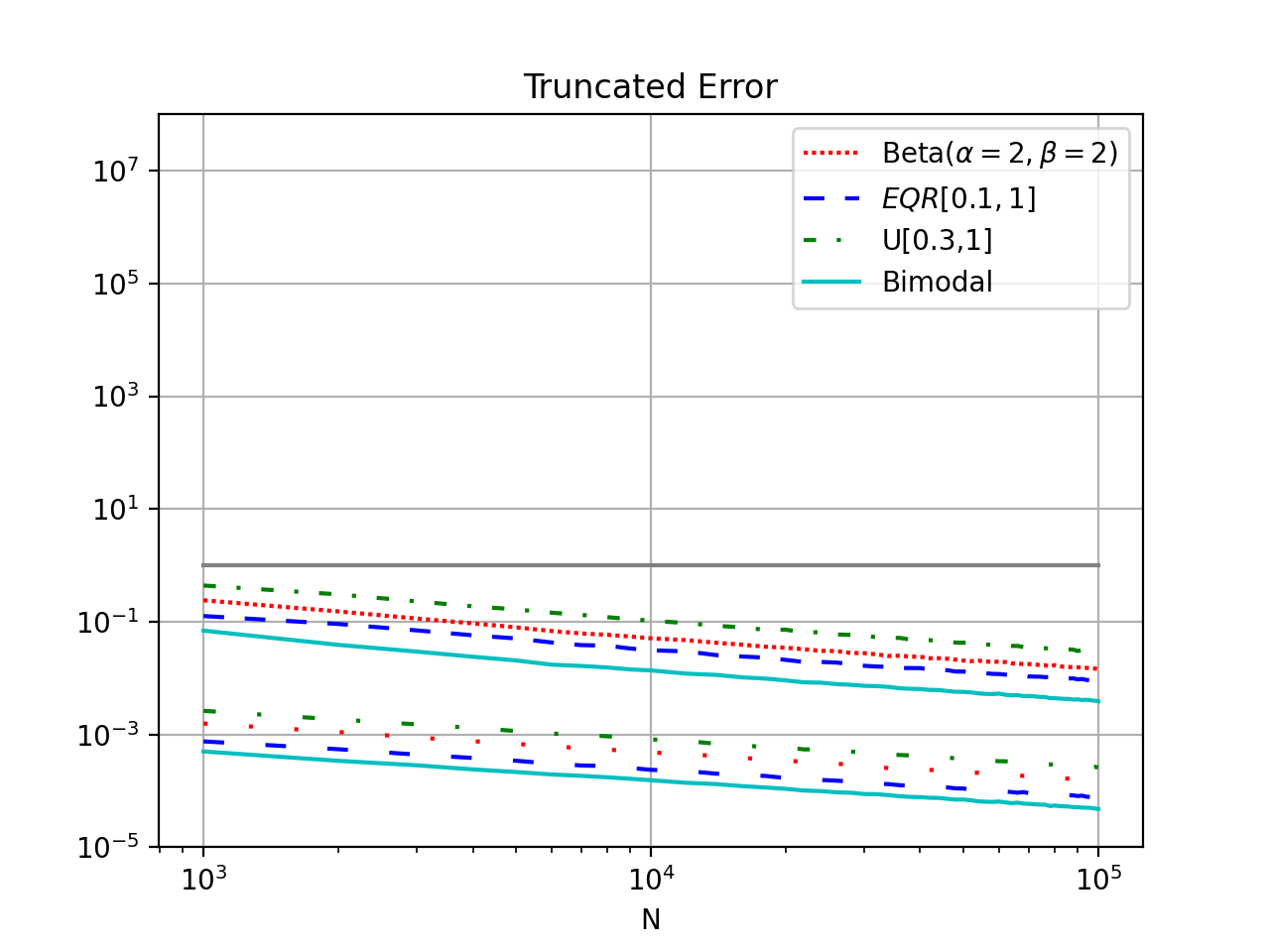} 
    \caption{truncation} 
  \end{subfigure}
  \hfill
  \begin{subfigure}[t]{0.48\textwidth}
    \centering
    \includegraphics[width=\linewidth]{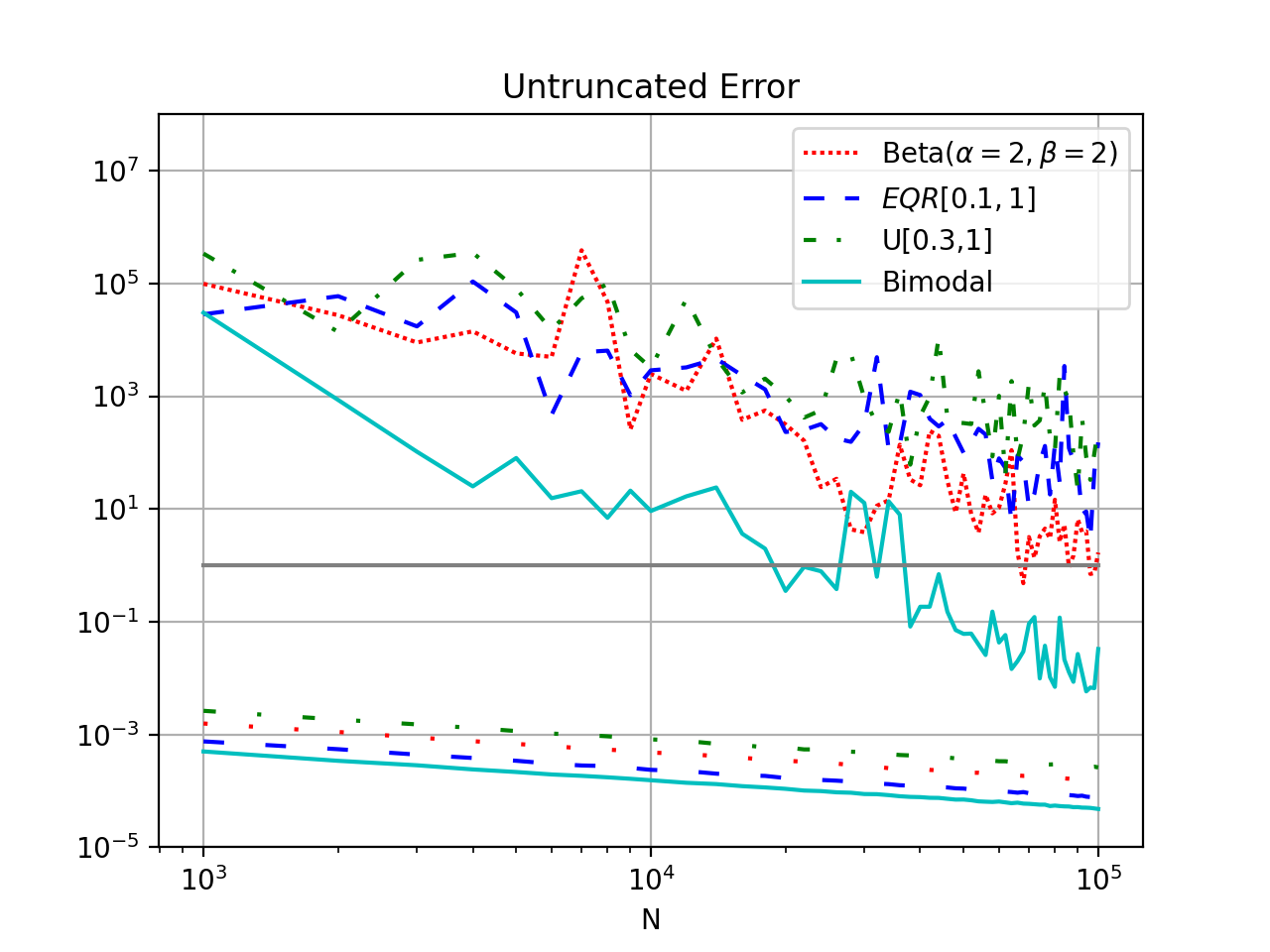} 
    \caption{no truncation} 
  \end{subfigure}
  \caption{This figure depicts on a log-log scale the error as a
    function of the number of samples $\samples$ (fixed parameters
    $n=16$, $\dist = \text{Beta}(2,2)$).  The counterfactual auction B
    is the high supply auction, the incumbent auction A is the low
    supply auction (the incumbent does not mix in B).  The trivial
    error bound of 1 is depicted with a solid line.  The thick line in
    each figure shows the error when the counterfactual and the
    incumbent are the same, a.k.a., the counterfactual error.}
  \label{fig:dist-trunc}
\end{figure}

\section{Derivation of error bounds (Proof of Theorems~\ref{thm:allpay-simple} and \ref{thm:allpay-general})}
\label{s:inference-k}

We will now derive the error bounds stated in Theorems~\ref{thm:allpay-simple} and \ref{thm:allpay-general} for the revenue estimator $\hat{P}$. Since the rate of estimation of bid density is generally worse than that of the bid distribution, we will first express the revenue estimator directly in terms of the empirical bid distribution. Written in this manner, the estimator will turn out to be a weighted order statistic of the empirical bids. We will then use standard bounds on the error in estimating bids to bound the error in estimating revenue. We begin by describing the standard statistical error bounds we use in our analysis.

\subsection{Statistical Model and Methods}

Our framework for counterfactual auction revenue analysis is based on directly
using the distribution of bids for inference. The main error in
estimation of the bid distribution is the {\em sampling error} due to
drawing only a finite number of samples from the bid distribution.
Evaluation of the auction revenue requires the knowledge of
the {\it quantile function} of bid distribution. While estimation of empirical
distributions is standard, quantile functions can be significantly more difficult to estimate
especially if the distribution density can approach zero on its support since the distribution
function is non-invertible at those points. As we show further,
estimation of the counterfactual auction revenues requires the knowledge of the 
{\it density-weighted quantile function} which can be robustly estimated despite the 
potential non-invertibility of the distribution function. In this subsection, we overview
the uniform absolute error bound of the density-weighted quantile function of the
bid distribution of a multi-unit auction based on the results in \citet{csorgo:83}. 

\begin{definition}
For function
$\bid(\cdot)$,estimator $\ebid(\cdot)$ and the weighting function $\omega(\cdot) \geq 0$, the {\em weighted uniform mean absolute
  error} 
is defined as $$\expect[\ebid]{ \sup\nolimits_\quant \omega(\quant)\big| \bid(\quant) -
\ebid(\quant) \big|}.$$
\end{definition}

The main object that will arise in our subsequent analyses will be the
weighted quantile function of the bid distribution where the weights
are determined by the allocation rules of the auctions under
consideration, e.g., $\expect[q]{\omega(q)\,b(q)}$ for some quantile
weighting function given by $\omega(\cdot)$.\footnote{Estimators of
  these functions, i.e., replacing the bid distribution with the
  empirical bid distribution, are called $L$-statistics in the
  statistics literature.} The important insight is that while the
estimation of the quantile function of the bid distribution
$\bhat(\cdot)$ maybe problematic around the points where the density
of the bid distribution is close to zero, the estimation of the
density-weighted quantile function is a lot more robust. As we will
show further, estimation of auction revenues involves such a
density-weighted form of the quantile function. Our error bounds are
based on the uniform convergence of quantile processes and weighted
quantile processes in \citet{csorgo:78}, \citet{csorgo:83}, and
\citet{cheng:97}.
For the quantile weighting function $\omega(q) = 1/b'(q)$, i.e., the
inverse derivative of the bid function, the $\sqrt{N}$-normalized mean
absolute error is bounded by a universal constant.  

\begin{lemma}
\label{error bid function}
Suppose that $b$ and $b'$ exist on $(0,\,1)$ and 
$
\sup_{q \in (0,\,1)}q(1-q)b'(q) <\infty.
$
Then the density-weighted uniform mean absolute error of the empirical quantile function
$\ebid(\cdot)$ on $\quant \in [\delta_N,\,1-\delta_N]$ with $\delta_N=\frac{25 \log\log N}{N}$ is bounded almost surely as
\begin{align*}\expect[\ebid]{
 \sup\nolimits_{\quant \in [\delta_N,\,1-\delta_N]} \big|\sqrt{N}(b^{\prime}(q))^{-1}( \bid(\quant) - \ebid(\quant)) \big|} <
1+16\frac{\log\log N}{\sqrt{N}}
\sup_q\,q(1-q)b'(q).
\end{align*}
\end{lemma}

This result is a consequence of statement (3.2.3) in Theorem 3.2.1 in
\cite{csorgo:83}.  For all-pay auctions by equation \eqref{eq:ap-inf},
the term $\sup_q\{q(1-q)b'(q)\}$ is bounded by $\frac14
\sup_q\{x'(q)\}$.  For first-price auctions by equation
\eqref{eq:fp-inf}, it is bounded by
$\sup_q\{q(1-q)x'(q)/x(q)\}$.

\subsection{Alternate formulation of the revenue estimator}
\label{s:inference-equation-alternate}

This section expresses the revenue estimator $\hat{P}$ directly as a function of the empirical bids. This alternate formulation is based on the 
same integration by parts technique that was
used to obtain the estimator for the revenue, but without
the subsequent grouping of the terms
that resulted in the weighted
sum of differences $\hat{b}_{i+1}-\hat{b}_i$ in the form of the estimator.

Recall from Section~\ref{s:inference-equation} that we have
\begin{align}
P_y &=
\expect[\quant]{y'(\quant)(1-\quant)\frac{\bid'(\quant)}{\alloc'(\quant)}}
= \expect[\quant]{Z_y(\quant)\,\bid'(\quant)}\\ \intertext{where
  $Z_y(\quant)=(1-\quant)\frac{y'(\quant)}{\alloc'(\quant)}$.
  Treating this expectation as an integral and integrating it by
  parts, when the constant terms are zero, gives:}
\label{eq:inference-integration-by-parts}
P_y &= \expect[\quant]{-Z'_y(\quant)\,\bid(\quant)}.
\end{align}
\noindent
The subsequent analysis will include consideration of the constant
terms when they are not zero.

This analysis gives two ways to write counterfactual revenue.
Equation~\eqref{eq:inference-pre-integration-by-parts} writes the
revenue as linear in the derivative of the bid function while
equation~\eqref{eq:inference-integration-by-parts} writes it as linear
in the bid function.  We will define our estimator in terms of former
for extreme quantiles and in terms of the latter for moderate
quantiles.  The reason for this definition is that the latter gives a
simple and well behaved estimator in terms of the bid function, but
might diverge at the extremes;\footnote{Both
  $Z_y(\quant)=(1-\quant)\frac{y'(\quant)}{\alloc'(\quant)}$ and
  $Z'_y(\quant)$ can be infinite at the boundary $\quant \in \{0,1\}$
  when $\alloc$ and $y$ are polynomials of different degrees.} while
the former at the extremes introduces only modest bias when
approximated by zero.

\begin{lemma}
\label{l:counterfactual-revenue}
\label{weights lemma}
The per-agent counterfactual revenue of a rank-based auction with
allocation rule $y$ can be expressed in terms of the bid function
$\bid$ of an all-pay mechanism $\alloc$ as:
\begin{align}
\label{eq:P_y}
P_y & = 
\overbrace{\expect[\quant\not\in\Lambda]{-Z'_y(\quant)\,\bid(\quant)} + 
Z_y(1-\delta_N)\,\bid(1-\delta_N)-Z_y(\delta_N)\,\bid(\delta_N)}^{\makebox[0in]{\tiny
    contribution from moderate quantiles}}\ + \ \\ \notag &\quad\,
 \underbrace{\expect[\quant\in\Lambda]{Z_y(\quant)\,\bid'(\quant)}}_{\makebox[0in]{\tiny
  contribution from extreme quantiles}}\\
\intertext{%
where $Z_y(\quant) = (1-\quant) \frac{y'(\quant)}{\alloc'(\quant)}$,
extreme quantiles are $\Lambda = [0,\delta_N] \cup [1-\delta_N,1]$,
and the truncation parameter is $\delta_N \in [0,1/2]$. For bid
functions that are constant on the extreme quantiles, the
counterfactual revenue can be written as}
\label{eq:P_y-truncated}
P_y & =
\expect[\quant\not\in\Lambda]{-Z'_y(\quant)\,\bid(\quant)} +
Z_y(1-\delta_N)\,\bid(1).
\end{align}
In this latter case or when $\delta_N = 0$, the expressed revenue is
linear in the bid function.
\end{lemma}

\begin{proof} 
The first part of the lemma follows from plugging the all-pay
inference equation~\eqref{eq:ap-inf} into the revenue
equation~\eqref{eq:bne-rev} and integrating by parts on moderate quantiles $[\delta_N,1-\delta_N]$.  The second part
of the lemma simplifies the first part using $\bid'(\quant) = 0$ for
extremal $\quant \in \Lambda$, $\bid(0) = \bid(\delta_N) = 0$, and $\bid(1-\delta_N) = \bid(1)$.
\end{proof}

This formulation allows the estimation of $P_y$ directly as a weighted
order statistic of the observed bids, with $b(\cdot)$ replaced by the
estimated bid distribution $\hat{b}(\cdot)$. Lemma~\ref{error bid
  function} tells us that, except at the extreme quantiles, the
estimated bid distribution $\hat{b}(\cdot)$ closely approximates
$b(\cdot)$.  At the extreme quantiles there is a bias-variance
tradeoff.  The variance from including the contribution to the revenue
from these quantiles in the estimator can greatly exceed the bias from
excluding them entirely.  Thus, to prevent the larger error at the
extreme quantiles from degrading the accuracy of the estimator, these
estimated bids are rounded down to zero and up to the maximum observed
bid at the low and high extremes, respectively.  Recall that the
estimated bid distribution $\hat{b}(\cdot)$ is defined, in
equation~\eqref{bid function}, as a piecewise constant function with
$N$ pieces. Thus, the estimator $\hat{P}_y =
\expect[q]{\smash{-Z'_y(q)\,\hat{\bid}(q)}}$ can be simplified as
expressed in the following definition. 

\begin{definition}
\label{d:estimator-alternate}
The estimator $\hat{P}_y$ (with truncation parameter $\delta_N$) for
the revenue of an auction with allocation rule $y$ from $N$ samples
$\hat{b}_1 \leq \cdots \leq \hat{b}_N$ from the equilibrium bid
distribution of an all-pay auction with allocation rule $x$ can be rewritten as:
\begin{align*}
\hat{P}_y& = \sum_{i=\delta_N N}^{N-\delta_N N} \left[ \left(1-\frac
    {i-1} N\right)\frac{y'(\frac {i-1}N)}{x'( \frac {i-1}N)} -
  \left(1-\frac{i}{N}\right)\frac{y'(\frac{i}N)}{x'(\frac{i}N)}
\right] \, \hat{\bid}_i + \delta_N \frac{y'(1-\delta_N)}{x'(1-\delta_N)} \hat{b}_N.
\end{align*}
\end{definition}

This alternate formulation is in fact numerically identical to the original definition of the revenue estimator, Definition~\ref{d:estimator}
but uses a different grouping
of the terms in the sum. From a computational perspective, Definition~\ref{d:estimator} turns out to be better behaved as the differences $Z'_y((i-1)/N) - Z'_y(i/N)$ can diverge for large $N$. But the alternate Definition~\ref{d:estimator-alternate}  is better suited to analyzing the statistical error, as the error in bids is better behaved than the error in bid derivatives.


\subsection{Derivation of error bounds}

We will now give the main ideas behind the proof
Theorem~\ref{thm:allpay-simple}.  This analysis is non-trivial because
the estimator, which is a weighted order statistic, is based on
weights with magnitudes that can be exponentially large.  Importantly
$Z_k(q) = (1-q) y'(q)/x'(q)$ and, while the numerator $y'(q)$ is
bounded by Fact~\ref{fact:max-slope} by $n$ (the number of agents),
the denominator can be exponentially small.  Thus, both $Z_k(q)$ and
its derivative $Z'_k(q)$ can be exponentially big, specifically
$N^{cn}$ for absolute constant $c \in (0,1)$ (see
Example~\ref{eg:extremal}, below).  The revenue estimator is a
weighted order statistic with weights proportional to $Z'_k$ and thus
straightforward analyses will not give good bounds on the error.  We
begin with one such analysis that gives an error bound that is linear
in the maximum of $Z_k(q)$ and modify it to reduce the dependence on
this term to be logarithmic.  For non-extremal quantiles $\quant$,
$Z_k(\quant)$ is bounded by $N^n$ and thus $\log Z_k(q)$ is at most
the $n \log N$ term that appears in the error bound of
Theorem~\ref{thm:allpay-simple}.

\begin{example}
\label{eg:extremal} The allocation rules and derivatives for the $k=1$ unit auction and the $k=n-1$ unit auction are:
\begin{align*}
\kalloc[n-1](q) &= 1-(1-q)^{n-1}; & \kalloc[n-1]'(q) &= (n-1)\,(1-q)^{n-2}.\\
\kalloc[1](q) &= q^{n-1};& \kalloc[1]'(q) &= (n-1) q^{n-2}. 
\end{align*}
Consider the estimator for the revenue of the $(n-1)$-unit auction
from bids in the one-unit auction, i.e., $y = \kalloc[n-1]$ and
$\alloc = \kalloc[1]$, at the lower extreme quantile $\quant = \log
\log N / N$ and with number of samples $N \gg n$.  We get $Z_k(\quant)
= (1-\quant) y'(\quant) / \alloc'(\quant) = (1-q)^{n-1}\,q^{2-n}
\approx q^{2-n}$ and $Z'_k(\quant) \approx (2-n)\,\quant^{1-n}$; thus
the magnitudes of both $Z_k(\quant)$ and $Z'_k(\quant)$ at quantile
$\quant = \log \log N / N$ are on the order of $[N/(\log \log
  N)]^{n}$ which is upper and lower bounded by $N^{cn}$ for
appropriate absolute constants $c$.  Recall that terms from extreme
quantiles below $\delta_N = O(\log \log N / N)$ are rounded down to zero
in the estimator; thus, this upper bound is tight for the subsequent
analysis.
\end{example}

The remainder of this section sketches the main ideas in deriving the
above logarithmic bound on the error.  The full proof of
Theorem~\ref{thm:allpay-simple}, as well as the more detailed analysis
that gives Theorem~\ref{thm:allpay-general}, is given in
Appendix~\ref{s:proofs-3}.
Assume that the counterfactual auction $y$ is the highest-$k$-bids-win
auction for some $k$; denote the allocation rule of this auction by
$\kalloc$, and let $Z_k = Z_{\kalloc}$. We will prove
Theorem~\ref{thm:allpay-simple} for this special case. Then, by virtue
of the fact that $P_y$ is a weighted average of the constituent
$\murevk$'s, the theorem trivially extends to all rank-based auctions
$y$. 

As in the statement of the theorem, let $\delta_N=\max(25\log\log
N,n)/N$, and let $\Lambda=[0,\delta_N]\cup[1-\delta_N, 1]$ denote the
set of extreme quantiles.  Apply equation~\eqref{eq:P_y} and
equation~\eqref{eq:P_y-truncated} from
\autoref{l:counterfactual-revenue} to the true bid function and
truncated empirical bid function to write the counterfactual revenue
and the estimated revenue, respectively, as:
\begin{align}
\notag \murevk & =
\expect[\quant\not\in\Lambda]{-Z'_k(\quant)\,\bid(\quant)} +
\expect[\quant\in\Lambda]{Z_k(\quant)\,\bid'(\quant)} \\ \notag &\quad +
Z_k(1-\delta_N)\bid(1-\delta_N)-Z_k(\delta_N)\bid(\delta_N).\\ \notag
\emurevk & =
\expect[\quant\not\in\Lambda]{\smash{-Z'_k(\quant)\,\hat{\bid}(\quant)}} +
Z_k(1-\delta_N)\hat{\bid}_N.  
\intertext{The mean absolute error is bounded by the
expected value of the absolute value of the difference in these two quantities:}
\label{eq:error-allpay}
|\emurevk-\murevk| \le & 
\left|\expect[q\not\in\Lambda]{\smash{-Z'_k(\quant)(\hat{\bid}(\quant)-\bid(\quant))}}\right| 
+ \left| \expect[\quant\in\Lambda]{Z_k(\quant)\,\bid'(\quant)}\right|\\
\notag & + \left| \smash{Z_k(1-\delta_N)\,(\bid(1-\delta_N) -\hat{\bid}_N)}\right|
+ \left| Z_k(\delta_N)\,\bid(\delta_N)\right|
\end{align}


There are now two steps to the analysis.  The first step is an analysis
of the contribution to the error from moderate quantiles, i.e., the
first term in equation~\eqref{eq:error-allpay}.  We will sketch this
step below.  The second step is analysis of the contribution to the
error from extreme quantiles, i.e., the remaining terms in
equation~\eqref{eq:error-allpay}.  In Appendix~\ref{s:proofs-3} we
show that the error from these terms is dominated by the error in the
first term.  As a summary of this deferred analysis, e.g., for the
final term $Z_k(\delta_N)\, b(\delta_N)$, the denominator of
$Z_k(\quant)$ can be very small, but it is approximately proportional
to $\bid(\quant)$ in the numerator and can be canceled.  The other
terms are similarly bounded.

The following straightforward analysis gives a bound on the error in
the estimator from moderate quantiles that is linear in $\sup_{q
  \not\in\Lambda} Z_k(q)$.  Specifically, the expected error is bounded as
\begin{align*}
\left|\expect[q\not\in\Lambda]{\smash{Z'_k(\quant) \,
  (\hat{\bid}(\quant)-\bid(\quant))}}\right| 
  &\leq \expect[q\not\in\Lambda]{|Z'_k(\quant)|} \cdot
  \sup\nolimits_{\quant \not\in \Lambda} |\hat{\bid}(\quant)-\bid(\quant)|.
\end{align*}
For the second term in this expression, Lemma~\ref{error bid function} provides a uniform
bound on the absolute error in bids
$|\hat{\bid}(\quant)-\bid(\quant)|$.  For the first term, the
following lemma shows that $Z_k$ is single-peaked and, thus,
$\expect[q\not\in\Lambda]{|Z'_k(\quant)|} \le 2 \sup_{q\not\in\Lambda}
Z_k(q)$.  The proof of this lemma, which is formally given in
Appendix~\ref{s:proofs-3}, follows from the fact that $\alloc$ is a
convex combination of multi-unit auctions and that the ratio of the
derivatives of the allocation rules of two multi-unit auctions is
single-peaked.

\begin{lemma}
\label{lem:Z-bound-1}
For any rank-based auction and $k$-highest-bids-win auction
with allocation rules $\alloc$ and $\kalloc$, respectively, the
function $Z_k(\quant)=(1-\quant)
\frac{\kalloc'(\quant)}{\alloc'(\quant)}$ achieves a single local
maximum for $\quant\in [0,1]$.
\end{lemma}

As described in Example~\ref{eg:extremal},
$\sup_{\quant\not\in\Lambda} Z_k(\quant)$ can be very large.  In order
to obtain a better bound, we observe that the error in bids is large
precisely at quantiles where $Z_k$ is small and vice versa: $Z_k$
depends inversely on the slope of the allocation rule of the incumbent
auction, $\alloc'$, whereas, the error in bids is directly
proportional to the bid density $\bid'$, which in turn is proportional
to $\alloc'$.  We utilize this observation as follows:

\begin{align}
\label{eq:error-split}
  \left|\expect[q\not\in\Lambda]{Z'_k(\quant) \,
    (\hat{\bid}(\quant)-\bid(\quant))}\right| & \le
  \expect[q\not\in\Lambda]{\left|\frac{Z'_k(\quant)}{Z_k(\quant)}\right| \,
   \left|Z_k(\quant)\,(\hat{\bid}(\quant)-\bid(\quant))\right|} \\
\notag
  & \le \expect[q\not\in\Lambda]{\left|\frac{Z'_k(\quant)}{Z_k(\quant)}\right|}
   \sup_{q}\left|Z_k(\quant)\,(\hat{\bid}(\quant)-\bid(\quant))\right|.
\end{align}

As the integral of $Z'_k(q) / Z_k(q)$ is $\log Z_k(q)$, this analysis
and the single-peaked-ness of $Z_k(q)$ gives an error bound that is
logarithmic instead of linear in $\sup_{q\not\in\Lambda} Z_k(q)$. The
following lemma, formally proved in Appendix~\ref{s:proofs-3},
summarizes the bound on the error from moderate quantiles.

\begin{lemma}
\label{first-term-bound-simple}
For $Z_k$ and $\Lambda$ defined as above, the first error term in
equation~\eqref{eq:error-allpay} of the estimator $\hat{P}_k$ is bounded by:
\begin{align*}
\expect[\hat{\bid}]{\left|\expect[q\not\in\Lambda]{Z'_k(\quant)\, (\hat{\bid}(\quant)-\bid(\quant))}\right|}
&\le \frac{8n\log\samples}{\sqrt{N}}\sup_{\quant \not\in \Lambda}\{\kalloc'(\quant)\}
\end{align*}
\end{lemma}

This lemma combines with analyses of the contribution to the error of
extremal quantiles to give Theorem~\ref{thm:allpay-simple}.  The
refined bound of Theorem~\ref{thm:allpay-simple} comes from improved
factoring of the error term over that of equation~\ref{eq:error-split}
in Lemma~\ref{first-term-bound-simple}.




\section{Inference for social welfare}
\label{s:welfare}

We now consider the problem of estimating the social welfare of a
rank-based auction using bids from another rank-based all-pay
auction. Consider a rank-based auction with induced position weights
$\wals$. By definition, the expected {\em per-agent} social welfare
obtained by this auction is as below, where $\evalk$ is the expected
value of the $k$th highest value agent, or the $k$th order statistic
of the value distribution.

$$
\sw = \frac 1n \sum_{k=1}^{n} \walk \evalk.
$$
 
We note that the value order statistics, $\evalk$, are closely
related to the expected revenues of the multi-unit auctions. The
$k$-unit second-price auction serves the top $k$ agents with
probability $1$, and charges each agent the $k+1$th highest value. Its
expected revenue is therefore $nP_k = k\evalk[k+1]$. We
therefore obtain:

$$
\sw = \walk[1]\frac{\evalk[1]}{n} + \sum_{k=1}^{n-1} \walk[k+1] \, \frac {P_{k}}{k}.
$$
 
The methodology developed in the previous sections can be used to
estimate the $P_k$'s in the above expression. The first order
statistic of the values, $\evalk[1]$, cannot be directly estimated in
this manner. Notate the expected value of an agent as
$$
\expval = \expect[\quant]{\val(\quant)} = \frac 1 n \sum_{k=1}^n \evalk.
$$
Therefore, we can calculate the social welfare of the position auction with weights $\wals$ as
\begin{eqnarray}
\sw = \walk[1] \expval - \sum_{k=2}^{n} (\walk[1]-\walk) \frac
{\evalk}{n} = \walk[1]\expval - \sum_{k=1}^{n-1} (\walk[1]-\walk[k+1]) \, \frac {P_k}{k}.
\label{eq:sw}
\end{eqnarray}

We now argue that $\expval$ can be estimated at a good rate from the
bids of another rank-based all-pay auction. Let $\alloc$ denote the allocation rule of the auction that we run,
and $\bid$ denote the bid distribution in BNE of this auction. Then we
note that
$$
\expval = \expect[\quant]{\val(\quant)} =
\expect[\quant]{\frac{\bid'(\quant)}{\alloc'(\quant)}} = \expect[\quant]{\Zbar(\quant)\bid'(\quant)}
$$
where $\Zbar(q) = 1/\alloc'(q)$. We might now try to directly apply
Theorems~\ref{thm:allpay-simple} or \ref{thm:allpay-general} to bound
the error in our estimate of $\expval$. This does not immediately
work, as Lemma~\ref{lem:Z-bound-1} fails to hold for $\Zbar$. Instead,
we observe that since $\alloc'(q)$ is a degree $n-1$ polynomial and
has fewer than $n$ local minima, therefore $\Zbar$ has fewer than $n$
local maxima. We can therefore adapt the arguments for the
aforementioned theorems to obtain the following lemma:
 
\begin{lemma}
\label{lem:inference-expval}
The mean absolute error in estimating the expected value $\expval$
using $N$ samples from the bid distribution for an all-pay rank-based
auction with allocation rule $x$ is bounded as given by the two
expressions below. Here $n$ is the number of positions in the position
auction.
\begin{align*}
\Err{\expval} & \le \frac{8n^2\log N}{\sqrt{N}}\\
\Err{\expval} & \le \frac{40n}{\sqrt{N}} \max\left\{1, \log\sup\nolimits_{\quant\not\in\Lambda}\alloc'(\quant),
\log \sup\nolimits_{\quant\not\in\Lambda}\tfrac{1}{\alloc'(\quant)} \right\} \\
\end{align*}
\end{lemma}


As an example application of Lemma~\ref{lem:inference-expval}, we
adapt Corollary~\ref{cor3} to bound the error from estimating the
social welfare of any position auction using bids from another
position auction that is mixed with the uniform-stair auction. Recall
that the uniform-stair auction is a universal B test.  Using the
universal B test of Corollary~\ref{cor:universal} instead of the
uniform-stair auction gives a slightly worse error bound, because the
slope of the allocation rule for that auction can be as small as
$N^{-O(n)}$. Other revenue estimation results can be similarly adapted
to estimate social welfare.
\begin{theorem}
\label{thm:sw}
For any rank-based auction A; uniform-stair auction B with position
weights $\walk = \frac{n-k}{n-1}$ for each $k\in [1, n]$; and all-pay
rank-based auction C with $x_C = (1-\eps)x_A + \eps x_B$; the mean
absolute error for estimating the social welfare of any rank-based
auction D from $N$ samples from the bid distribution of C is bounded
by:
\begin{align*}
O\left( \frac{n}{\sqrt{N}} +\frac{n\log n \log(n/\eps)}{\sqrt{N}}\right)
 & = O\left( \frac{n\log n \log(n/\eps)}{\sqrt{N}} \right).
\end{align*}
\end{theorem}
The theorem follows by combining Lemma~\ref{lem:inference-expval} with
equation~\eqref{eq:sw} and Corollary~\ref{cor3}. The first term
follows from Lemma~\ref{lem:inference-expval} by noting that the
uniform-stair auction satisfies $x'(q)=1$ for all $q$. The second term
follows from the error bounds on $P_k$ given by Corollary~\ref{cor3};
The extra factor of $\log n$ (relative to the statement of the
corollary) arises from the fact that the total weight of the
multipliers for the terms in equation~\eqref{eq:sw} can be as large as
$\sum_{k=1}^n 1/k \approx \log n$.



\bibliographystyle{apalike}
\bibliography{uniq,agt,references}

\begin{thebibliography}{}

\bibitem[Alaei et~al., 2012]{AFHHM-12}
Alaei, S., Fu, H., Haghpanah, N., Hartline, J., and Malekian, A. (2012).
\newblock Bayesian optimal auctions via multi- to single-agent reduction.
\newblock In {\em ACM Conference on Electronic Commerce}.

\bibitem[Athey and Ellison, 2011]{AE-11}
Athey, S. and Ellison, G. (2011).
\newblock Position auctions with consumer search.
\newblock {\em The Quarterly Journal of Economics}, 126:1213--1270.

\bibitem[Athey and Haile, 2007]{athey:2007}
Athey, S. and Haile, P.~A. (2007).
\newblock Nonparametric approaches to auctions.
\newblock {\em Handbook of econometrics}, 6:3847--3965.

\bibitem[Auer et~al., 2002]{ACFS-02}
Auer, P., Cesa-Bianchi, N., Freund, Y., and Schapire, R.~E. (2002).
\newblock The nonstochastic multiarmed bandit problem.
\newblock {\em SIAM Journal on Computing}, 32(1):48--77.

\bibitem[Baliga and Vohra, 2003]{BV-03}
Baliga, S. and Vohra, R. (2003).
\newblock Market research and market design.
\newblock {\em Advances in Theoretical Economics}, 3(1).

\bibitem[Blum and Hartline, 2005]{BH-05}
Blum, A. and Hartline, J.~D. (2005).
\newblock Near-optimal online auctions.
\newblock In {\em Proceedings of the sixteenth annual ACM-SIAM symposium on
  Discrete algorithms}, pages 1156--1163. Society for Industrial and Applied
  Mathematics.

\bibitem[Brown and Morgan, 2009]{BM-09}
Brown, J. and Morgan, J. (2009).
\newblock How much is a dollar worth? tipping versus equilibrium coexistence on
  competing online auction sites.
\newblock {\em Journal of Political Economy}, 117(4):668--700.

\bibitem[Bulow and Klemperer, 1996]{BK-96}
Bulow, J. and Klemperer, P. (1996).
\newblock Auctions vs negotiations.
\newblock {\em American Economic Review}, 86(1):180--194.

\bibitem[Cesa-Bianchi et~al., 2015]{CGM-15}
Cesa-Bianchi, N., Gentile, C., and Mansour, Y. (2015).
\newblock Regret minimization for reserve prices in second-price auctions.
\newblock {\em IEEE Transactions on Information Theory}, 61(1):549.

\bibitem[Chawla et~al., 2014]{CHN-14}
Chawla, S., Hartline, J., and Nekipelov, D. (2014).
\newblock Mechanism design for data science.
\newblock In {\em Proceedings of the fifteenth ACM conference on Economics and
  computation}, pages 711--712. ACM.

\bibitem[Chawla et~al., 2016]{CHN-16}
Chawla, S., Hartline, J., and Nekipelov, D. (2016).
\newblock A/b testing of auctions.
\newblock In {\em Proceedings of the 2016 ACM Conference on Economics and
  Computation}, pages 19--20. ACM.

\bibitem[Chawla and Hartline, 2013]{CH-13}
Chawla, S. and Hartline, J.~D. (2013).
\newblock Auctions with unique equilibria.
\newblock In {\em Proceedings of the Fourteenth ACM Conference on Electronic
  Commerce}, EC '13, pages 181--196, New York, NY, USA. ACM.

\bibitem[Cheng and Parzen, 1997]{cheng:97}
Cheng, C. and Parzen, E. (1997).
\newblock Unified estimators of smooth quantile and quantile density functions.
\newblock {\em Journal of statistical planning and inference}, 59(2):291--307.

\bibitem[Coey et~al., 2014]{coey:2014}
Coey, D., Larsen, B., and Sweeney, K. (2014).
\newblock The bidder exclusion effect.
\newblock Technical report, National Bureau of Economic Research.

\bibitem[Cole and Roughgarden, 2014]{CR-14}
Cole, R. and Roughgarden, T. (2014).
\newblock The sample complexity of revenue maximization.
\newblock In {\em Proceedings of the 46th Annual ACM Symposium on Theory of
  Computing}, pages 243--252. ACM.

\bibitem[Cs{\"o}rg{\"o}, 1983]{csorgo:83}
Cs{\"o}rg{\"o}, M. (1983).
\newblock {\em Quantile processes with statistical applications}.
\newblock SIAM.

\bibitem[Csorgo and Revesz, 1978]{csorgo:78}
Csorgo, M. and Revesz, P. (1978).
\newblock Strong approximations of the quantile process.
\newblock {\em The Annals of Statistics}, pages 882--894.

\bibitem[Devanur et~al., 2015]{DHY-14}
Devanur, N., Hartline, J., and Yan, Q. (2015).
\newblock Envy freedom and prior-free mechanism design.
\newblock {\em Journal of Economic Theory}, 156:103--143.

\bibitem[Devanur et~al., 2013]{DHH-13}
Devanur, N.~R., Ha, B.~Q., and Hartline, J.~D. (2013).
\newblock Prior-free auctions for budgeted agents.
\newblock In {\em ACM Conference on Electronic Commerce}, pages 287--304.

\bibitem[Dughmi et~al., 2012]{DRS-12}
Dughmi, S., Roughgarden, T., and Sundararajan, M. (2012).
\newblock Revenue submodularity.
\newblock {\em Theory of Computing}, 8:95--119.

\bibitem[Edelman et~al., 2007]{EOS-07}
Edelman, B., Ostrovsky, M., and Schwarz, M. (2007).
\newblock Internet advertising and the generalized second-price auction:
  Selling billions of dollars worth of keywords.
\newblock {\em The American Economic Review}, 97(1):242--259.

\bibitem[Fain and Pedersen, 2006]{FP-06}
Fain, D. and Pedersen, J. (2006).
\newblock Sponsored search: A brief history.
\newblock {\em Bulletin of the American Society for Information Science and
  Technology}, 32(2):12--13.

\bibitem[Fu et~al., 2014]{FHHK-14}
Fu, H., Haghpanah, N., Hartline, J., and Kleinberg, R. (2014).
\newblock Optimal auctions for correlated buyers with sampling.
\newblock In {\em Proceedings of the fifteenth ACM conference on Economics and
  computation}, pages 23--36. ACM.

\bibitem[Goldberg et~al., 2006]{GHKSW-06}
Goldberg, A.~V., Hartline, J.~D., Karlin, A.~R., Saks, M., and Wright, A.
  (2006).
\newblock Competitive auctions.
\newblock {\em Games and Economic Behavior}, 55(2):242--269.

\bibitem[Guerre et~al., 2000]{guerre}
Guerre, E., Perrigne, I., and Vuong, Q. (2000).
\newblock Optimal nonparametric estimation of first-price auctions.
\newblock {\em Econometrica}, 68(3):525--574.

\bibitem[Hardy et~al., 1929]{HLP-29}
Hardy, G., Littlewood, J., and P\'olya, G. (1929).
\newblock Some simple inequalities satisfied by convex functions.
\newblock {\em Messenger of Math}, 58:145--152.

\bibitem[Hartline and Taggart, 2019]{HT-19}
Hartline, J. and Taggart, S. (2019).
\newblock Sample complexity for non-truthful mechanisms.
\newblock In {\em Proceedings of the 2019 ACM Conference on Economics and
  Computation}, pages 399--416.

\bibitem[Jackson and Sonnenschein, 2007]{JS-07}
Jackson, M.~O. and Sonnenschein, H.~F. (2007).
\newblock Overcoming incentive constraints by linking decisions 1.
\newblock {\em Econometrica}, 75(1):241--257.

\bibitem[Kleinberg and Leighton, 2003]{KL-03}
Kleinberg, R. and Leighton, T. (2003).
\newblock The value of knowing a demand curve: Bounds on regret for online
  posted-price auctions.
\newblock In {\em Foundations of Computer Science, 2003. Proceedings. 44th
  Annual IEEE Symposium on}, pages 594--605. IEEE.

\bibitem[Kohavi et~al., 2009]{KLSH-09}
Kohavi, R., Longbotham, R., Sommerfield, D., and Henne, R.~M. (2009).
\newblock Controlled experiments on the web: survey and practical guide.
\newblock {\em Data mining and knowledge discovery}, 18(1):140--181.

\bibitem[Marmer and Shneyerov, 2012]{marmer}
Marmer, V. and Shneyerov, A. (2012).
\newblock Quantile-based nonparametric inference for first-price auctions.
\newblock {\em Journal of Econometrics}, 167(2):345 -- 357.

\bibitem[Myerson, 1981]{mye-81}
Myerson, R. (1981).
\newblock Optimal auction design.
\newblock {\em Mathematics of Operations Research}, 6:58--73.

\bibitem[Ostrovsky and Schwarz, 2011]{OS-11}
Ostrovsky, M. and Schwarz, M. (2011).
\newblock Reserve prices in internet advertising auctions: A field experiment.
\newblock In {\em Proceedings of the 12th ACM Conference on Electronic
  Commerce}, EC '11, pages 59--60, New York, NY, USA. ACM.

\bibitem[Paarsch and Hong, 2006]{paarsch}
Paarsch, H.~J. and Hong, H. (2006).
\newblock {\em An introduction to the structural econometrics of auction data},
  volume~1.
\newblock The MIT Press.

\bibitem[Paes~Leme et~al., 2020]{PST-20}
Paes~Leme, R., Sivan, B., and Teng, Y. (2020).
\newblock Why do competitive markets converge to first-price auctions?
\newblock In {\em Proceedings of The Web Conference 2020}, pages 596--605.

\bibitem[Reiley, 2006]{Ril-06}
Reiley, D.~H. (2006).
\newblock Field experiments on the effects of reserve prices in auctions: more
  magic on the internet.
\newblock {\em The RAND Journal of Economics}, 37(1):195--211.

\bibitem[Segal, 2003]{seg-03}
Segal, I. (2003).
\newblock Optimal pricing mechanisms with unknown demand.
\newblock {\em The American economic review}, 93(3):509--529.

\bibitem[Varian, 2007]{var-06}
Varian, H. (2007).
\newblock Position auctions.
\newblock {\em International Journal of Industrial Organization},
  25(6):1163--1178.

\end{thebibliography}

\appendix
\section{Proofs for Section~\ref{s:param-inf}}
\label{s:proofs-3}

In this section we prove the results from Section~\ref{s:param-inf}
which analyze the error of the counterfactual revenue estimator for
both multi-unit and (more generally) rank-based auctions with all-pay
payment semantics.

Recall that for all-pay auctions with allocation rule
$\alloc(\quant)$, the equilibrium bid function $\bid(\quant)$
satisfies $\bid'(\quant) = v(\quant)\,\alloc'(\quant)$.  From $N$ bids
in a mechanism with allocation rule $x$ we are estimating the
counterfactual revenue of a mechanism with allocation rule $y$.
Recall that for an implicit allocation rule $x$ and another allocation
rule $y$, we define the function $Z_y(q) = (1-q)
\frac{y'(q)}{x'(q)}$. When $y$ is the allocation rule corresponding to
a $k$-unit auction, we let $Z_k(q)$ denote $Z_{\kalloc}(q)$.  Our
analysis treats the contribution to the error from extreme quantiles
 $\quant \in \Lambda = [0,\delta_N]\cup [1-\delta_N, 1]$ for $\delta_N =
\max(25\log\log N,n)/N$ and moderate quantiles $\quant \not \in \Lambda$
separately.  In equation~\eqref{eq:error-allpay}, restated below, the
first term is the error from moderate quantiles and the latter three
terms is the error from extremal quantiles.

\begin{align}
\tag{\ref{eq:error-allpay}}
|\emurevk-\murevk| \le & 
\left|\expect[q\not\in\Lambda]{\smash{-Z'_k(\quant)(\hat{\bid}(\quant)-\bid(\quant))}}\right|
+ \left| \expect[\quant\in\Lambda]{Z_k(\quant)\,\bid'(\quant)}\right|\\
\notag & + \left| \smash{Z_k(1-\delta_N)\,(\bid(1-\delta_N) -\hat{\bid}_N)}\right|
+ \left| Z_k(\delta_N)\,\bid(\delta_N)\right|
\end{align}


The proofs in this appendix are organized as follows. The error in our
estimator for the revenue $\murevk$ of a $k$-unit auction from
moderate quantiles is analyzed in Section~\ref{sec:first-term-proof}.
Section~\ref{sec:alloc-bid-bounds} proves some basic properties of
allocation rules and bid functions for rank-based auctions that will
be employed in Section~\ref{sec:extreme-quantile-bounds} where the
error from extremal quantiles, specifically the three latter terms of
equation~\eqref{eq:error-allpay}, are analyzed.  The main results from
Section~\ref{s:inference-k-thm}, namely
Theorems~\ref{thm:allpay-simple} and~\ref{thm:allpay-general} and
Corollary~\ref{cor:allpay-y} are proven in
Section~\ref{sec:inference-theorem-proofs}.

\subsection{Bounding the error from moderate quantiles}
\label{sec:first-term-proof}

We will now restate and prove Lemmas~\ref{lem:Z-bound-1}
and~\ref{first-term-bound-simple}, bounding the contribution to the
error of the estimator from moderate quantiles,
$\expect[q\not\in\Lambda]{|Z'_k(\quant)|\,
  |\hat{\bid}(\quant)-\bid(\quant)|}$.  The first lemma proves that
$Z_k$ has a single local maximum.

\begin{numberedlemma}{\ref{lem:Z-bound-1}}
For any rank-based auction and $k$-highest-bids-win auction
with allocation rules $\alloc$ and $\kalloc$, respectively, the
function $Z_k(\quant)=(1-\quant)
\frac{\kalloc'(\quant)}{\alloc'(\quant)}$ achieves a single local
maximum for $\quant\in [0,1]$.
\end{numberedlemma}

\begin{proof}
  Consider the function $A(q) = 1/Z_k(q) = x'(q)/(1-q)x'_k(q)$.
  Recall that $x'(q)$ is a weighted sum over $x'_j(q)$ for
  $j\in\{1,\cdots,n-1\}$.  Thus, $A(q)$ is a weighted sum over terms
  $x'_j(q)/(1-q)x'_k(q)$. Let us look at these terms closely.
$$
\frac{x'_j(q)}{(1-q)x'_k(q)} = \alpha_{k,j} q^{k-j}(1-q)^{j-k-1}
$$ 
where coefficient $\alpha_{k,j}$ is a constant. The functions
$q^{k-j}(1-q)^{j-k-1}$ are convex. This implies that $A(q)$ which is a
weighted sum of convex functions is also convex. Consequently, it has
a unique minimum. Therefore, $Z_k(q) = 1/A(q)$ has a unique maximum.
\end{proof}

The following lemma gives the basic analysis of the error from
moderate quantiles.  A key aspect of this proof is that its dependence
on $\sup_{\quant \not \in \Lambda} Z_k(\quant)$ is logarithmic.
Immediately following this proof we give a more refined analysis that
enables better bounds when estimating the revenue of counterfactual
mechanism $y$ from bids in $\alloc$ when the allocation rules of
$\alloc$ and $y$ are related.

\begin{numberedlemma}{\ref{first-term-bound-simple}}
For $Z_k$ and $\Lambda$ defined as above, the first error term in
equation~\eqref{eq:error-allpay} of the estimator $\hat{P}_k$ is bounded by:
\begin{align*}
\expect[\hat{\bid}]{\left|\expect[q\not\in\Lambda]{Z'_k(\quant)\, (\hat{\bid}(\quant)-\bid(\quant))}\right|}
&\le
\frac{8n\log\samples}{\sqrt{N}}\sup_{\quant}\{\kalloc'(\quant)\}
\end{align*}
\end{numberedlemma}

\begin{proof}
Recall from Section~\ref{s:inference-k} that we can write the error on the moderate quantiles as:
\begin{align}
\tag{\ref{eq:error-split}}
  \left|\expect[q\not\in\Lambda]{Z'_k(\quant) \,
    (\hat{\bid}(\quant)-\bid(\quant))}\right| 
  & \le \expect[q\not\in\Lambda]{\left|\frac{Z'_k(\quant)}{Z_k(\quant)}\right|}
   \sup_{q}\left|Z_k(\quant)\,(\hat{\bid}(\quant)-\bid(\quant))\right|.
\end{align}
Using Lemma~\ref{lem:Z-bound-1}, the first term on the right in equation~\eqref{eq:error-split},
$\expect[q\not\in\Lambda]{\left|\frac{Z'_k(\quant)}{Z_k(\quant)}\right|}$,
is bounded by $2(\sup_{\quant\not\in\Lambda}\log
Z_k(\quant)-\inf_{\quant\not\in\Lambda}\log Z_k(\quant))$.

We note that for $\quant\not\in\Lambda$, and any rank-based
allocation rule $y$, $y'(\quant)\in (\delta_N^n, n]$. Therefore,
$Z_k(\quant)\in [\delta_N^{n}/n, n\delta_N^{-n}]\in (N^{-n},N^n)$. Therefore, we have:
\begin{align*}
  \expect[q\not\in\Lambda]{\left|\frac{Z'_k(\quant)}{Z_k(\quant)}\right|}
  & < 4\log N^n = 4n\log N.
\end{align*}
To bound the second term on the right in equation~\eqref{eq:error-split}, we write:
\begin{align*}
  \sup_{q}\left|Z_k(\quant)\,(\hat{\bid}(\quant)-\bid(\quant))\right|
  & \le \sup_{q} \kalloc'(q)
 \sup_q\left|\frac{1}{\alloc'(q)}\,(\hat{\bid}(\quant)-\bid(\quant))\right|\\
  & \le \sup_{q} \kalloc'(q)
 \sup_q\left|\frac{1}{\bid'(q)}\,(\hat{\bid}(\quant)-\bid(\quant))\right|.\\
 \intertext{Invoking Lemma~\ref{error bid function}, the expected value of this term for random samples from the bid distribution is bounded as:}
\expect[\ebid]{\sup_{q}\left|Z_k(\quant)\,(\hat{\bid}(\quant)-\bid(\quant))\right|}
 & \le \sup_{q} \kalloc'(q) \frac{1}{\sqrt{\samples}} \left(1+\frac{4n\log\log N}{\sqrt{N}}\right).
\end{align*}
Putting the two bounds together, we get,
\begin{align*}
\expect[\hat{\bid}]{\left|\expect[q\not\in\Lambda]{Z'_k(\quant)\, (\hat{\bid}(\quant)-\bid(\quant))}\right|}
&\le
\frac{4n\log\samples}{\sqrt{N}}\sup_{\quant}\{\kalloc'(\quant)\}\left(1+\frac{4n\log\log
    N}{\sqrt{N}}\right)
\end{align*}
We may assume without loss of generality that
$4n\log\samples<\sqrt{N}$, otherwise the first term, and therefore the
entire error bound, exceeds $1$ and is trivially true. Under this
assumption, the term in brackets is no more than $2$, and the lemma follows.
\end{proof}

The following lemma gives a refinement of
Lemma~\ref{first-term-bound-simple} that enables better bounds when
estimating the revenue of counterfactual mechanism $y$ from bids in
$\alloc$ when the allocation rules of $\alloc$ and $y$ are related.

Unfortunately,
$\expect[q\not\in\Lambda]{\left|\frac{Z'_k(\quant)}{Z_k(\quant)}\right|}$
can be quite large, as $Z_k(\quant)$ can take on exponentially large
values at extreme quantiles (see Example~\ref{eg:extremal} in
Section~\ref{s:inference-k}).  The main idea in the refined analysis
is a better factoring in the error from moderate quantiles in
equation~\eqref{eq:error-split}.  We instead factor this error term as
follows, for an appropriate function $h(Z_k)$ which is just slightly
sublinear in $Z_k$.
\begin{align*}
  \left|\expect[q\not\in\Lambda]{Z'_k(\quant) \,
    (\hat{\bid}(\quant)-\bid(\quant))}\right| 
 & \le \expect[q\not\in\Lambda]{\left|\frac{Z'_k(\quant)}{h(Z_k)}\right|}
   \sup_{q}\left|h(Z_k)\,(\hat{\bid}(\quant)-\bid(\quant))\right|.
\end{align*}

This factoring gives greater control in balancing the error generated from
the two terms. For an appropriate choice of the function $h(\cdot)$,
we obtain the following lemma.

\begin{lemma}
\label{first-term-bound}
For $Z_k$ and $\Lambda$ defined as above, the first error term in
equation~\eqref{eq:error-allpay} of the estimator $\hat{P}_k$ is bounded by:
\begin{align*}
&\expect[\hat{\bid}]{\left|\expect[q\not\in\Lambda]{Z'_k(\quant)\, (\hat{\bid}(\quant)-\bid(\quant))}\right|}\\
&\le \frac{40}{\sqrt{\samples}}\,\, \left(1+\frac{4n\log\log N}{\sqrt{N}}\right)
\sup_{\quant}\{\kalloc'(\quant)\}\,\,\max\left\{1,\log 
\sup_{\quant: \kalloc'(\quant)\ge 1}\frac{\alloc'(\quant)}{\kalloc'(\quant)}, \log\sup_{\quant}\frac{\kalloc'(\quant)}{\alloc'(\quant)} \right\}.
\end{align*}
\end{lemma}

\begin{proof}
For any $\alpha>0$ we can write 
$$
|\hat{P}_k-P_k|\leq \expect{\frac{\left(\log(1+Z_k(q))\right)^{\alpha}}{Z_k(q)}|Z'_k(q)|}
\sup\limits_q\left|\frac{Z_k(q)}{\left(\log(1+Z_k(q))\right)^{\alpha}}(\hat{b}(q)-b(q))\right|.
$$
We start by considering the first term.
Lemma~\ref{lem:Z-bound-1} shows that $Z'_k(\cdot)$ changes sign
only once. Consider the region where the sign of $Z'_k(\cdot)$ is constant
and make the change of variable $t=Z_k(q)$. 
Denote $Z_k^*=\sup_qZ_k(q)$, and note that $\inf_qZ_k(q) \geq 0$.
The first term evaluates as
$$
\expect{\frac{\left(\log(1+Z_k(q))\right)^{\alpha}}{Z_k(q)}|Z'_k(q)|} \leq
2 \int^{Z_k^*}_0\frac{(\log\,(1+t))^{\alpha}}{t}\,dt.
$$
Note that for any $t>0$, $\log(1+t)\le t$. Thus,
$$
 \int^{\delta}_0\frac{(\log\,(1+t))^{\alpha}}{t}\,dt< \frac{\delta^{\alpha}}{\alpha}.
$$
Now split the integral into two pieces as
$$
\int^{Z_k^*}_0\frac{(\log\,(1+t))^{\alpha}}{t}\,dt=\int^{1}_0\frac{(\log\,(1+t))^{\alpha}}{t}\,dt+\int^{Z_k^*}_1\frac{(\log\,(1+t))^{\alpha}}{t}\,dt.
$$
We just proved that the first piece is at most $1/\alpha$. Now we
upper bound the second piece and consider the integrand at $t \ge 1$. First, note that 
$$
(\log\,(1+t))^{\alpha}=\left(\log\,t+\log(1+\frac{1}{t})\right)^{\alpha}\le\left(\log\,t+\frac{1}{t}\right)^{\alpha}\le(\log\,t+1)^{\alpha}.
$$
Thus, the integral behaves as
$$
\int^{Z_k^*}_1\frac{(\log\,(1+t))^{\alpha}}{t}\,dt\le\int^{Z_k^*}_1\frac{(\log\,(t)+1)^{\alpha}}{t}\,dt=\frac{1}{1+\alpha}(\log\,Z_k^*+1)^{1+\alpha}.
$$
Thus, we just showed that 
$$
\expect{\frac{\left(\log(1+Z_k(q))\right)^{\alpha}}{Z_k(q)}|Z'_k(q)|} \leq \frac{2}{\alpha}+\frac{2}{1+\alpha}(\log\,Z_k^*+1)^{1+\alpha},
$$
which is at most $2(1+e)/\alpha$ for $\alpha<1/\log\,Z_k^*$.

Now consider the term
$$
\sup\limits_q\left|\frac{Z_k(q)}{\left(\log(1+Z_k(q))\right)^{\alpha}}(\hat{b}(q)-b(q))\right|.
$$
Note that $\log(1+t)\ge \min\{1,t\}/2$. So the first term can be bounded from above as
$$
\frac{Z_k(q)}{\left(\log(1+Z_k(q))\right)^{\alpha}} \leq 2^{\alpha} \max \left\{
Z_k(q),\,(Z_k(q))^{1-\alpha} \right\}.
$$
Thus using Lemma~\ref{error bid function},
\begin{align*}
&\expect{\sup\limits_q\left|\frac{Z_k(q)}{\left(\log(1+Z_k(q))\right)^{\alpha}}(\hat{b}(q)-b(q))\right|}\\
&\le \sup\limits_q\left|\frac{Z_k(q)}{\left(\log(1+Z_k(q))\right)^{\alpha}}
b'(q)\right|
\expect{ \sup\limits_q \left|\frac{\hat{b}(q)-b(q)}{b'(q)}\right|}\\
& \leq 2^{\alpha}\sup_q\left(\max \left\{
x'_k(q),\,(x'_k(q))^{1-\alpha}(x'(q))^{\alpha}
\right\}
\right)\frac{1}{\sqrt{N}}\left(1+ 16\frac{\log\log N}{\sqrt{N}}\sup_q q(1-q) b'(q)\right)\\
& \le
2^{\alpha}\sup_q\left(x'_k(q)\right)\left(\underbrace{\max\left(1,\sup_{q:x'_k(q)\ge 1}\frac{x'(q)}{x'_k(q)}\right)}_{=: \, A}\right)^{\alpha} \frac{1}{\sqrt{N}}\left(1+ \frac{4n\log\log N}{\sqrt{N}}\right).
\end{align*}
where the last inequality follows by noting that $b'(q)\le x'(q)\le
n$, and $q(1-q)\le 1/4$.


Now we combine the two evaluations together and pick
$\alpha=\min\{1,1/\log A, 1/\log Z_k^*\}$, with $A$ defined as above, to obtain
\begin{align*}
\expect{|\hat{P}_k-P_k|} &\leq 
\frac{2(1+e)}{\alpha}2^{\alpha}A^{\alpha}\frac{1}{\sqrt{N}} \sup_q\left(x'_k(q)\right)\\
& \leq \frac{40}{\sqrt{\samples}} \,\,
\sup_{\quant}\{\kalloc'(\quant)\}\,\,\max\left\{1,\log A,\log \sup_{\quant}
\left\{ \frac{\kalloc'(\quant)}{\alloc'(\quant)} \right\}\right\} \left(1+ \frac{4n\log\log N}{\sqrt{N}}\right).
\end{align*}
\end{proof}

\subsection{Bounds for the allocation rules and bid distributions of
  rank-based auctions}
\label{sec:alloc-bid-bounds}

In this section we prove some basic properties of allocation rules for
rank-based auctions.  These properties will be useful, in
Section~\ref{sec:extreme-quantile-bounds}, for analizing the error of
the estimator at extreme quantiles.  As desribed in
Section~\ref{s:prelim}, the allocation rule and its derivative for the
$n$-agent $k$-unit auction are

\begin{align*}
x_k(q) &= \sum_{i=0}^{k-1} \binom{n-1}{i} q^{n-i-1} (1-q)^{i},\\
x'_k(q) &= (n-1) \binom{n-2}{k-1} q^{n-k-1} (1-q)^{k-1}.\\
\end{align*}

We will be interested in the behavior of allocation rule $\kalloc$ and
its derivative $\kalloc'$ at the extremes, specifically for $q\in
[0,1/n]$ and $q\in[1-1/n,1]$.  The allocation rule is steepest at
$q=(k-1)/(n-2)$ and is convex before this point and concave after it.
Specifically, $\kalloc[1]$ is steepest at $q=1$ and is convex and
$\kalloc[n-1]$ is steepest at $q=0$ and is concave. For all other $k
\in \{2,\ldots,n-2\}$, the allocation rule derivative $\kalloc'$ is
maximized between $1/(n-2) > 1/n$ and $(n-3)/(n-2) < 1-1/n$.

The following two lemmas bound the derivative of the allocation of
multi-unit auctions at extreme quantiles.  Combining them we obtain
the subsequent theorem.

\begin{lemma} 
For $k \in \{2,n-2\}$ units and $\delta < 1/n$, the allocation rule
derivative $x'_k$ satisfies:
\begin{enumerate}
\item $\sup_{q < \delta} x'_k(q) = x'_k(\delta)$ and
\item $\sup_{q > 1-\delta} x'_k(q) = x'_k(1-\delta)$.
\end{enumerate}
\end{lemma}

\begin{proof} 
This lemma follows from convexity of the allocation rule $x_k$ on
$[0,1/n]$ and concavity on $[1-1/n,1]$.
\end{proof}

\begin{lemma} 
For $k \in \{1,n-1\}$ units and $\delta < 1/n$, the allocation rule
derivative $x'_k$ satisfies:
\begin{enumerate}
\item $\sup_{q < \delta} x'_{n-1}(q) \leq e\,x'_{n-1}(\delta)$ and
\item $\sup_{q > 1-\delta} x'_1(q) \leq e\, x'_1(1-\delta)$.
\end{enumerate}
\end{lemma}

\begin{proof} 
This lemma follows from the closed-from of the allocation rule derivatives as
$x'_1(q) = (n-1)\,q^{n-2}$ and $x'_{n-1}(q) = (n-1)\,(1-q)^{n-2}$.
Thus, $x'_1(1) = x'_{n-1}(0) = n-1$ and
\begin{align*}
x'_1(1-\delta) = x'_{n-1}(\delta) &= (n-1)\,(1-\delta)^{n-2} \\
           &\geq \frac 1e\, (n-1)\\
           &= \frac 1e\,x'_1(1) = \frac 1e\,x'_{n-1}(0).
\end{align*}  
Concavity of $x_{n-1}$ and convexity of $x_1$, then, imply the result.
\end{proof}

\begin{theorem}
\label{t:extremal-xprime-bound} 
For any $n$-agent rank-based mechanism with allocation rule $x$ and
$\delta < 1/n$, the allocation rule derivative $x'$ satisfies:
\begin{enumerate}
\item $\sup_{q < \delta} x'(q) \leq e\,x'(\delta)$ and
\item $\sup_{q > 1-\delta} x'(q) \leq e\, x'(1-\delta)$.
\end{enumerate}
\end{theorem}

The bid function $b(\cdot)$ can be bounded by the allocation rule
$x(\cdot)$ and its derivative $x'(\cdot)$ via the following lemma.
The subsequent theorem follows from the lemma via
Theorem~\ref{t:extremal-xprime-bound}.

\begin{lemma}
\label{l:bid-bound}
For any all-pay mechanism with allocation rule $x$ and $\delta \in
[0,1]$, the equilibrium bid function $b$ satisfies
\begin{enumerate}
\item \label{part:bid-first} 
   $b'(\delta) \leq x'(\delta)$,
\item \label{part:bid-second} 
   $b(\delta) \leq \delta \sup_{q < \delta} x'(\delta)$, and
\item \label{part:bid-third} 
   $b(1) - b(1-\delta) \leq \delta \sup_{q > 1-\delta} x'(\delta)$.
\end{enumerate}
\end{lemma}

\begin{proof}
The equilibrium bid function is defined by $b'(q) = v(q) x'(q)$ and
$b(0) = 0$ (where $v(q) \in [0,1]$ is the value function). Part
(\ref{part:bid-first}) follows from the upper bound $v(q) \leq 1$.  Parts
(\ref{part:bid-second}) and (\ref{part:bid-third}) follow by upper
bounding $x'(q)$ by its supremum on the interval of the integral and
integrating the bound of part (\ref{part:bid-first}).  For example for
part (\ref{part:bid-second}), $b(\delta) = \int_0^\delta v(r)\, x'(r)\, dr
\leq \int_0^\delta \sup_{q < \delta} x'(q) \,dr = \delta\,\sup_{q \leq
  \delta} x'(q)$.
\end{proof}

\begin{theorem}
\label{t:bid-bound}
For any $n$-agent all-pay rank-based mechanism with allocation rule
$x$ and $\delta < 1/n$, the equilibrium bid function $b$ satisfies
\begin{enumerate}
\item $b(\delta) \leq \delta\, e\,x'(\delta)$, and
\item $b(1) - b(1-\delta) \leq \delta\,e\,x'(1-\delta)$.
\end{enumerate}
\end{theorem}

\subsection{Bounding the error at extreme quantiles}
\label{sec:extreme-quantile-bounds}

We now bound the remaining terms in
equation~\eqref{eq:error-allpay}. Once again these bounds rely on the
observation that for any quantile $q$, $Z_k(q)b(q)$ is bounded,
because $Z_k$ depends inversely on $x'(q)$, whereas $b(q)$ is roughly
proportional to it.

\begin{lemma}
\label{second-term-bound}
For $Z_y$ and $\Lambda$ as defined above, if $\delta_N\le 1/n$,
the second error term of
 the estimator $\hat{P}_y$ is bounded as follows.
\begin{align*}
\expect[\quant\in\Lambda]{Z_y(\quant)\,\bid'(\quant)} &\le
 e\,\delta_N\, y'(\delta_N) + e\,\delta_N^2 y'(1-\delta_N).
\end{align*} 
\end{lemma}

\begin{proof}
Apply part~(\ref{part:bid-first}) of Lemma~\ref{l:bid-bound} and the
definition of $Z_y(q) = (1-q)\,y'(q)/x'(q)$ to obtain the following
upper bound:
$$
\expect[\quant\in\Lambda]{Z_y(\quant)\,\bid'(\quant)} 
  \le 
\expect[q \in \Lambda]{(1-q)\,y'(q)}.
$$ 
For $q < \delta_N$, bound this epectation by $e\,\delta_N\,
y'(\delta_N)$ from Theorem~\ref{t:extremal-xprime-bound}.  For $q >
1-\delta_N$, bound this expectation by $e\,\delta_N^2\,
y'(1-\delta_N)$.

Note, we could alternatively obtain the bound $\delta_N n$ by using
the fact that $\sup_q y'(q) \leq n$ (Fact~\ref{fact:max-slope}).
\end{proof}

\begin{lemma}
\label{third-term-bound}
For $Z_y$ and $\Lambda$ as defined above,  if $\delta_N\le 1/n$, the third error term of
 the estimator $\hat{P}_y$ is bounded as follows.
\begin{align*}
\expect[\hat{\bid}]{\left|\smash{Z_y(1-\delta_N)(\bid(1-\delta_N) -\hat{\bid}_N)}\right|}
& \le \delta_N y'(1-\delta_N)\, (e\delta_N+ \tfrac{8}{N}).
\end{align*}
\end{lemma}

\begin{proof}
Let $\equant$ be the quantile of the highest of the $N$ observed bids,
i.e., $\bid(\equant) = \ebid_N$.

Conditioned on $\equant>1-\delta_N$, bid $\ebid_N = \bid(\equant)$ is
upper bounded by $\bid(1)$.  Applying \autoref{t:bid-bound} to bound
$\bid(1) - \bid(1-\delta_N)$ gives conditional error bound of
$$\smash{Z_y(1-\delta_N)(\bid(1) - \bid(1-\delta_N))} \leq
e\,\delta_N^2 y'(1-\delta_N).$$

Now condition on $\equant < 1-\delta_N$.  For this conditioning,
Lemma~\ref{l:bid-bound} shows that $\bid(1-\delta_N) - \bid(\equant)
\le \alloc(1-\delta_N) - \alloc(\equant)$.  We will now bound
$\expect[\equant]{\alloc(1-\delta_N)-\alloc(\equant) \given \equant <
  1-\delta_N}\prob{\equant \leq 1-\delta_N}$ which is at most
$\expect[\equant]{1-\alloc(\equant)}$.

We first analize $\expect[\equant]{1-\alloc(\equant)}$ in the case
that $\alloc=\alloc_k$ is the allocation rule for the $k$-unit
auction. We have,
\begin{align*}
\expect[\equant]{1-\alloc(\equant)} & = \int_0^1 (1-x_k(q)) N q^{N-1} dq \\
& = N \int_0^1 q^{N-1}\left( \sum_{i=k}^{i=n-1} {n-1 \choose i} q^{n-1-i}
  (1-q)^i \right) dq \\
& = N \sum_{i=k}^{i=n-1} {n-1 \choose i} \int_0^1 q^{N+n-2-i} (1-q)^i
dq \\
& = N \sum_{i=k}^{i=n-1} {n-1 \choose i} \frac{(N+n-2-i)!
  i!}{(N+n-1)!} \\
& = \frac{N}{N+n-1} \sum_{i=k}^{i=n-1} \frac{{n-1 \choose i}}{{N+n-2
  \choose i}} \\
& \le \frac{N}{N+n-1}  \sum_{i=k}^{i=n-1} \left(\frac{n}{N}\right)^i \\
& \le \frac{N}{N+n-1} \left(\frac{n}{N}\right)^k \frac{1}{1-n/N} \\
& \le 2 \left(\frac{n}{N}\right)^k,
\end{align*}
where the last inequality uses $N>1.5n$.

Substituting this back, we get for $x=x_k$:
\begin{align*}
& \expect[\hat{\bid}]{\left|\smash{Z_y(1-\delta_N)(\bid(1-\delta_N)
      -\hat{\bid}_N)}\right|}\\
& \le \delta_N \frac{y'(1-\delta_N)}{x'_k(1-\delta_N)}   \left\{
  e\delta_Nx'_k(1-\delta_N) +  2 \left(\frac{n}{N}\right)^k \right\}\\
& = e\delta_N^2 y'(1-\delta_N) + 2\delta_N y'(1-\delta_N) \left\{
  \frac{1}{(n-1){n-2\choose k-1} (1-\delta_N)^{n-1-k}\delta_N^{k-1}}
  \left(\frac{n}{N}\right)^k \right\}\\
& \le e\delta_N^2 y'(1-\delta_N) + \frac{2}{N}\left(\frac{n}{n-1}\right) \delta_N y'(1-\delta_N)
\left(\frac{n}{N\delta_N}\right)^{k-1} \frac{1}{{n-2\choose k-1} (1-\delta_N)^{n-1-k}}\\
& \le e\delta_N^2 y'(1-\delta_N) + \frac{8}{N}\delta_N y'(1-\delta_N).
\end{align*}

Here the last inequality follows by noting that ${n-2\choose k-1}\ge
1$, $(1-\delta_N)^n>1/4$, and using $\delta_N\ge n/N$,
$\frac{n}{N\delta_N}\le 1$.



Finally, since $x$ is a linear combination of the $x_k$'s, we have,
\begin{align*}
& \expect[\hat{\bid}]{\left|\smash{Z_y(1-\delta_N)(\bid(1-\delta_N)
      -\hat{\bid}_N)}\right|}\\
& \le \delta_N \frac{y'(1-\delta_N)}{x'(1-\delta_N)}
\left(e\delta_Nx'(1-\delta_N) + \expect[\hat{\bid}]{|1-x(q)|} \right)\\
& \le \max_k \delta_N \frac{y'(1-\delta_N)}{x'_k(1-\delta_N)}
\left(e\delta_Nx'_k(1-\delta_N) + \expect[\hat{\bid}]{|1-x_k(q)|}
\right)\\
& \le e\delta_N^2 y'(1-\delta_N) + \frac{8}{N}\delta_N y'(1-\delta_N). \qedhere
\end{align*}
\end{proof}

\begin{lemma}
\label{fourth-term-bound}
For $Z_y$ and $\Lambda$ as defined above, if $\delta_N\le 1/n$,
the fourth error term of the estimator $\hat{P}_y$ is bounded as
follows.
\begin{align*}
Z_y(\delta_N)\bid(\delta_N)
& \le e\,\delta_N\,y'(\delta_N).
\end{align*}
\end{lemma}

\begin{proof}
The lemma follows directly from the definition of $Z_y$ with the
upper-bound on $b(\delta_N)$ of Theorem~\ref{t:bid-bound}.
\end{proof}

\subsection{Proofs of main theorems}
\label{sec:inference-theorem-proofs}

This section gives the complete proofs for the main theorems of
Section~\ref{s:inference-k-thm}.  These theorems follow fairly
directly from the previous lemmas.

\begin{numberedtheorem}{\ref{thm:allpay-simple}}
  The mean absolute error in estimating the revenue of a rank-based
  auction with allocation rule $y$ using $N$ samples from the bid
  distribution for an all-pay rank-based auction with allocation rule
  $x$ is bounded as below. Here $n$ is the number of positions in the
  two auctions, and $\hat{P}_y$ is the estimator in Definition~\ref{d:estimator}
  with $\delta_N$ set to $\max(25\log\log N,n)/N$.
\begin{align*}
\Err{P_y} & \le \frac{16n^2\log N}{\sqrt{N}}.
\end{align*}
\end{numberedtheorem}

\begin{proof}
As in the proof of Lemma~\ref{first-term-bound-simple}, we may assume
without loss of generality that $4n\log\samples<\sqrt{N}$, and indeed,
$16n^2\log N<\sqrt{N}$. This implies $\delta_N<1/n$, and then
Lemmas~\ref{first-term-bound-simple}, \ref{second-term-bound},
\ref{third-term-bound}, and \ref{fourth-term-bound} together imply
that the error in $P_k$ is bounded by:
\begin{align*}
\Err{\murevk} \le &
\,\frac{8n\log\samples}{\sqrt{N}}\sup_{\quant\not\in\Lambda}\{\kalloc'(\quant)\}
+ 2e\delta_N \kalloc'(\delta_N) + (2e+8) \delta_N^2 \kalloc'(1-\delta_N)  
\end{align*}
Further, $16n^2\log N<\sqrt{N}$ also implies that the second and third
terms together are no larger than the first. The theorem then follows
by recalling that $\sup_{\quant} \kalloc'(\quant) \le n$.
\end{proof}

We will now prove the improved error bounds of
Theorem~\ref{thm:allpay-general} and Corollary~\ref{cor:allpay-y}.
Recall the definition of $\rxy$ from equation~\ref{eq:rxy} in
Section~\ref{s:inference-k-thm}.
\begin{align}
\tag{\ref{eq:rxy}}
\rxy &:= \sup_{\quant}\{y'(\quant)\}\,\max\left\{1, \,\log 
\sup_{\quant: y'(\quant)\ge 1}\frac{\alloc'(\quant)}{y'(\quant)}, \,\log\sup_{\quant}\frac{y'(\quant)}{\alloc'(\quant)} \right\}.
\end{align}

Theorem~\ref{thm:allpay-general} follows from 
Lemma~\ref{first-term-bound} in much the same way as
Theorem~\ref{thm:allpay-simple} does from
Lemma~\ref{first-term-bound-simple}. We may assume, without loss of
generality, that $\sqrt{\samples}<80$, in which case the errors from the
extreme quantiles get absorbed into the error from the moderate quantiles.

\begin{numberedtheorem}{\ref{thm:allpay-general}}
  Let $\alloc$ and $\kalloc$ denote the allocation rules for any
  all-pay rank-based auction and the $k$-highest-bids-win auction over
  $n$ positions, respectively. Let $\hat{\murevk}$ denote the
  estimator from Definition~\ref{d:estimator} for estimating the
  revenue $P_k$ of the latter auction from $\samples$ samples of the
  bid distribution of the former, with $\delta_N$ set to
  $\max(25\log\log N,n)/N$. If $\delta_N\le 1/n$, the mean absolute
  error of the estimator $\hat{\murevk}$ is bounded as follows.
\begin{align*}
  \Err{\murevk}
\le & \,\frac{80}{\sqrt{\samples}}\,\rxy[\kalloc].
\end{align*}
\end{numberedtheorem}

We now generalize error bound to estimate the revenue $P_y$ of an
arbitrary rank-based auction with allocation rule $y$ from the bids of
another rank-based auction with allocation rule $x$.

\begin{numberedcorollary}{\ref{cor:allpay-y}}
  Let $\alloc$ and $y$ denote the allocation rules for any two all-pay
  rank-based auctions over $n$ positions. Let $\hat{P}_y$ denote the
  estimator from Definition~\ref{d:estimator} for estimating the
  revenue of the latter from $\samples$ samples of the bid
  distribution of the former, with $\delta_N$ set to $\max(25\log\log
  N,n)/N$. If $\delta_N\le 1/n$, the mean absolute error of the
  estimator $\hat{P}_y$ is bounded as follows.
\begin{align*}
  \Err{P_y}
\le & \,\frac{80}{\sqrt{\samples}}\, n \, \log \sup\nolimits_q n\, \frac{y'(q)}{x'(q)}.
\end{align*}
\end{numberedcorollary}

\begin{proof}
Write $y$ as a rank-based auction with weights $\wals$:
\begin{align*}
y & = \sum\nolimits_k \margwalk\,\kalloc, \,\,\,\text{and, } \,\,\,
P_y = \sum\nolimits_k \margwalk\,P_k.\\
\intertext{Accordingly, the error in $P_y$ is bounded by a weighted sum of the
  error in $P_k$ which are bounded by Theorem~\ref{thm:allpay-general}.  The weighted sum of these errors is simplified by observing
that $\kalloc'(q)\le y'(q)/\margwalk$ for all $k$ and $q$:}
\Err{P_y}  & \le \sum\nolimits_k \margwalk\,\Err{P_k}\\
& \le \tfrac{80}{\sqrt{N}} \sum\nolimits_k \margwalk \rxy[\kalloc]\\
& \le \tfrac{80}{\sqrt{N}} \sum\nolimits_k \margwalk \sup_q
\{\kalloc'(q)\} \max\Big\{\log n,\ \log\tfrac{1}{\margwalk}+ \log\sup_q \tfrac{y'(q)}{x'(q)}\Big\}.\\
\end{align*}
We now simplify the terms one at a time. Recall that $\sup_q
\{\kalloc'(q)\}\le n$ for all $k$. The first and third terms can therefore be
simplified using $\sum\nolimits_k \margwalk\le 1$. For the second
term, we observe $\sum\nolimits_k \margwalk \log \frac{1}{\margwalk}
\le \log n$. We therefore have:
\begin{align*}
  \Err{P_y} 
& \yestag \label{eq:weak-bound-py}
= \frac{80}{\sqrt{N}}\, n\,\log \sup_q n\,\frac{y'(q)}{x'(q)}.
\end{align*}
\end{proof}

\section{Proofs for Section~\ref{s:ab-testing}} 
\label{s:proofs-5}

We will now prove Theorem~\ref{theorem:sign-AB}, restated here for convenience. 

\begin{numberedtheorem}
{\ref{theorem:sign-AB}}
For arbitrary $n$-agent rank-based auctions $A$, $B_1$, and $B_2$ and
$N$ bids from the equilibrium bid distribution of mechanism $C=\eps
B_1+\eps B_2+(1-2\eps)A$, the estimator for the binary classifier
$\gamma={\bf 1}\{P_{B_1}-\alpha\,P_{B_2}>0\}$, that establishes
whether the revenue of mechanism $B_1$ exceeds $\alpha$ times the
revenue of mechanism $B_2$, has error rate bounded by
$$\exp\left(-O\left(\frac{Na^2}{\alpha^2n^3\log(n/\eps)}\right)\right),$$
where $a=|P_{B_1}-\alpha\,P_{B_2}|$, as long as $N\gg n/\epsilon a$.  
\end{numberedtheorem}

Whereas we bound the expected absolute error of our revenue estimator
in Section~\ref{s:param-inf}, in this section we will require a
concentration result for the error. We state this concentration result
below and prove it in Section~\ref{s:concentration-proof}. We focus on
the main term in our error bound,
$\expect[q\not\in\Lambda]{-Z'_y(\quant)(\ebid(\quant)-\bid(\quant))}$. We
can split this error into two components, one corresponding to the
bias in the estimated bid function and the other corresponding to the
deviation of the estimated bids from their mean:

\begin{align}
\left|\expect[q\not\in\Lambda]{-Z'_y(\quant)(\ebid(\quant)-\bid(\quant))}\right|
\le \left|\expect[q\not\in\Lambda]{Z'_y(\quant)(\ebid(\quant)-\btild(\quant))}\right|
+ \left|\expect[q\not\in\Lambda]{Z'_y(\quant)(\btild(\quant)-\bid(\quant))}\right|.
\label{eq:bias-deviation-split}
\end{align}
Here, $\btild$ is a step function that equals the expectation of the
empirical bid function $\ebid$: $\btild(\quant) = \expect{\ebid(\quant)}$.

The bias of the estimator, i.e., the second term above, is small:
\begin{lemma}
\label{lem:bias-bound}
With $\btild$ defined as above,
  $$\left|\expect[q\not\in\Lambda]{Z'_y(\quant)(\btild(\quant)-\bid(\quant))}\right|=
  \frac{O(1)}{N} \sup_{\quant}\{\alloc'(\quant)\}\,\,\sup_{\quant}
\left\{ \frac{y'(\quant)}{\alloc'(\quant)} \right\}.$$
\end{lemma}

The deviation from the mean, i.e., the first term in equation~\eqref{eq:bias-deviation-split}, is concentrated. 
\begin{lemma}
\label{lem:deviation-prob}
Let $\Delta=\sup_{\quant\not\in\Lambda}|(b'(q))^{-1}(\ebid(\quant)-\btild(\quant))|$. Then for any $a>0$,
  $$\prob{\left|\expect[q\not\in\Lambda]{Z'_y(\quant)(\ebid(\quant)-\btild(\quant))}\right|\ge a\, \middle| \, \Delta}\le
  \exp \left(-\frac{a^2}{n(80\Delta \rxy)^2}\right).$$
\end{lemma}

The proofs of Lemmas~\ref{lem:bias-bound} and \ref{lem:deviation-prob} are deferred to the next subsection.
We are now ready to prove Theorem~\ref{theorem:sign-AB}.

\begin{proofof}{Theorem~\ref{theorem:sign-AB}}
We need to bound the probability that the error in estimating
$\hat{P}_{B_1}-\alpha\hat{P}_{B_2}$ is greater than $|P_{B_1}-\alpha\,P_{B_2}|$. This
error can in turn be decomposed into the error in estimating $P_{B_1}$ and
that in estimating $P_{B_2}$. Denote $a=|P_{B_1}-\alpha\,P_{B_2}|>0$. Then,
\begin{align*}
& \prob{|(\hat{P}_{B_1}-\alpha\,\hat{P}_{B_2}) 
  -(P_{B_1}-\alpha\,P_{B_2})|>a} \\ 
& \leq
\prob{|\hat{P}_{B_1}-P_{B_1}|>a/2}
+ \prob{|\hat{P}_{B_2}-P_{B_2}|>a/2\alpha}.
\end{align*}
Let $\alloc$ denote the allocation rule of the mechanism $C$ that we
are running, and let $\bid$ be the corresponding bid function. Now, recall that for
\begin{align*}
\Delta & =\sup_q|(b'(q))^{-1}(\bwhat(q)-b(q))| \text{ and }  \\
\rxy[x_{B_1}] & =\sup_{\quant}\{x_{B_1}'(\quant)\}\max\left\{1, \log\sup_{\quant: x'_{B_1}(\quant)\ge 1}\frac{\alloc'(\quant)}{x_{B_1}'(\quant)},
\log\sup_{\quant}\frac{x_{B_1}'(\quant)}{\alloc'(\quant)} \right\}
\end{align*}
equations~\eqref{eq:error-allpay} and \eqref{eq:bias-deviation-split} bound the error in estimation as a sum of five terms. Of these, all but the first term in equation~\eqref{eq:bias-deviation-split} can be bounded by $O(n/\eps N)$ using Lemmas~\ref{second-term-bound}, \ref{third-term-bound}, \ref{fourth-term-bound}, and \ref{lem:bias-bound}. Then, Lemma~\ref{lem:deviation-prob} implies that, conditioned on $\Delta$,
$$
\prob{|\hat{P}_{B_1}-P_{B_1}|>a/2} \leq
2\exp\left(-\frac{1}{n(80\,\Delta\,\rxy[x_{B_1}])^2}\left(\frac a2-O\left(\frac {n}{\epsilon N}\right)\right)^2\right).
$$
Finally, $\rxy<n\log (n/\eps)$, and with high probability
$\Delta$ is at most a constant times $1/\sqrt{N}$ (Lemma~\ref{error
  bid function}). Consequently, for $N\gg n/\epsilon a$,
$$
\prob{|\hat{P}_{B_1}-P_{B_1}|>a/2} \leq
\exp\left(-O\left(\frac{Na^2}{n^3\log(n/\eps)}\right)\right).
$$ 
Likewise,
$$
\prob{|\hat{P}_{B_2}-P_{B_2}|>a/2\alpha} \leq
\exp\left(-O\left(\frac{Na^2}{\alpha^2n^3\log(n/\eps)}\right)\right).
$$
\end{proofof}

\subsection{Concentration bound for the revenue estimator}
\label{s:concentration-proof}

\begin{numberedlemma}{\ref{lem:bias-bound}}
With $\btild$ defined as above,
  $$\left|\expect[q\not\in\Lambda]{Z'_y(\quant)(\btild(\quant)-\bid(\quant))}\right|=
  \frac{O(1)}{N} \sup_{\quant}\{\alloc'(\quant)\}\,\,\sup_{\quant}
\left\{ \frac{y'(\quant)}{\alloc'(\quant)} \right\}.$$
\end{numberedlemma}

\begin{proof}
We can write the function $\tilde{b}(i/N)$ as
\begin{align*}
\tilde{b}(i/N)& =\frac{N!}{(i-1)!(N-i)!}\int G(t)^{i-1}(1-G(t))^{N-i}g(t) t\,dt\\
&=\frac{N!}{(i-1)!(N-i)!}\int t^{i-1}(1-t)^{N-i}b(t)\,dt.
\end{align*}
Note that 
$$\frac{N!}{(i-1)!(N-i)!} t^{i-1}(1-t)^{N-i}$$
is the density of the beta distribution with parameters $\alpha=i$ and $\beta=N-i+1$. Denote this density $f(t;\alpha,\beta)$. Then we can write
$$  \tilde{b}(i/N)=\int^1_0b(t)f(t;\alpha,\beta)\,dt.$$
Now let $q \in [i/N,\,(i+1)/N]$, and consider an expansion of $b(t)$ at $q$ such that 
$$ b(t)=b(q)+b'(q)(t-q)+O((t-q)^2).$$
Now we substitute this expansion into the formula for $\tilde b(\cdot)$ above to get
$$  \tilde{b}(i/N) = b(q)+b'(q)\int^1_0(t-q)f(t;\alpha,\beta)\,dt+O(\int^1_0(t-q)^2f(t;\alpha,\beta)\,dt).$$
The mean of the beta distribution is $\alpha/(\alpha+\beta)$ and the variance is $\alpha\beta/((\alpha+\beta)^2(\alpha+\beta+1))$. This means that 
$$ \tilde{b}\left(\frac{i}{N}\right)-b(q)=b'(q)\left(\frac{i}{N+1}-q\right)+O\left(\frac{1}{N^2}\right). $$
Thus
$$\sup_{q \in [i/N,(i+1)/N]}\left| \tilde{b}(i/N)-b(q)\right| \leq \sup_qb'(q)\frac{2}{N}+O\left(\frac{1}{N^2}\right).$$
Therefore, the expectation $\left|\smash{\Phat_y-\expect{\smash{\Phat_y}}}\right|$ is at most $O(1)/N
\, \sup_q\{\alloc'(q)\} \, \sup_q Z_y(q)$.
\end{proof}

We now focus on the deviation of our estimator from its mean. In order to obtain a concentration bound, we express the estimator as a sum over many independent terms.

To this end, we first identify the set of quantiles at which the
function $\ebid$ ``crosses'' the function $\btild$ from below. This
set is defined inductively. 
Define $i_0=\delta_N N$. Then, inductively, let
$i_\ell$ be the smallest integer strictly greater than $i_{\ell-1}$
such that
$$\ebid\left(\frac{i_{\ell}-1}{N}\right)\le\btild\left(\frac{i_\ell-1}{N}\right) \, \text{ and }\,
\ebid\left(\frac{i_{\ell}}{N}\right)>\btild\left(\frac{i_\ell}{N}\right).$$
Let $i_{\lastcrossing-1}$ be the last integer so defined, and let $i_\lastcrossing=(1-\delta_N)N$. Let $I$
denote the set of indices $\{i_0,\ldots, i_\lastcrossing\}$.
Let $\Tij$ denote the following integral:
$$\Tij = \int_{\quant=i/N}^{\quant=j/N}
Z'_y(\quant)(\ebid(\quant)-\btild(\quant)) \, \text{d}q$$
Then, our goal is to bound the quantity
$\expect[\hat{\bid}]{|\T_{0,N}|}$ where $\T_{0,N}$ can be written
as the sum:
$$ \T_{0,N} = \sum_{\ell=0}^{\lastcrossing-1}\T_{i_{\ell}, i_{\ell+1}}. $$
We now claim that conditioned on $I$ and the maximum weighted bid
error, this is a sum over independent random variables.

\begin{lemma}
\label{lem:independent-sums}
Conditioned on the set of indices $I$ and
$\Delta=\sup_{\quant\not\in\Lambda}|(b'(q))^{-1}(\ebid(\quant)-\btild(\quant))|$, over the
randomness in the bid sample, the random variables $\T_{i_{\ell},
  i_{\ell+1}}$ are mutually independent.
\end{lemma}

\begin{proof}
  Fix $I$ and $\ell$, and note that the function $\btild$ is fixed
  (that is, it does not depend on the empirical bid sample). Then, the
  sum $\T_{i_{\ell}, i_{\ell+1}}$ depends only on the empirical bid
  values $\ebid(\quant)$ for quantiles in the interval $[i_\ell/N,
  i_{\ell+1}/N)$. By the definition of $I$, we know that the smallest
  $i_\ell$ bids in the sample are all smaller than
  $\btild((i_\ell-1)/N)\le\btild(i_\ell/N)$, and the largest $N-i_{\ell+1}$ bids in the
  sample are all larger than $\btild(i_{\ell+1}/N)\ge \btild((i_{\ell+1}-1)/N)$. On the other
  hand, the empirical bids $\ebid(\quant)$ for $\quant\in [i_\ell/N,
  i_{\ell+1}/N)$ lie within $[\btild(i_\ell/N),
  \btild((i_{\ell+1}-1)/N)]$. Therefore, conditioned on $i_\ell$ and
  $i_{\ell+1}$, the latter set of empirical bids is independent of the
  former set of empirical bids.
\end{proof}

Since within each interval $(i_{\ell}, i_{\ell+1})$ the multiplier
$\ebid(\quant)-\btild(\quant)$ changes sign only once, we can apply
the approach of Section~\ref{s:inference-k}, to bound each
individual $\T_{i_{\ell}, i_{\ell+1}}$ by $40\Delta \rxy$. We then
apply Chernoff-Hoeffding bounds to obtain a bound on the proability
that $\expect[\hat{\bid}]{|\T_{0,N}|\, | \, I, \Delta}$ exceeds some
value $a>0$.

\begin{numberedlemma}{\ref{lem:deviation-prob}}
Let $\Delta=\sup_{\quant\not\in\Lambda}|(b'(q))^{-1}(\ebid(\quant)-\btild(\quant))|$. Then for any $a>0$,
  $$\prob{\left|\expect[q\not\in\Lambda]{Z'_y(\quant)(\ebid(\quant)-\btild(\quant))}\right|\ge a\, \middle| \, \Delta}\le
  \exp \left(-\frac{a^2}{n(80\Delta \rxy)^2}\right).$$
\end{numberedlemma}

\begin{proof}
  We will use Chernoff-Hoeffding bounds to bound the expectation of $\T_{0,N}$ over the bid sample, conditioned on $I$ and $\Delta$. We first note that $\T_{0,N}$ has mean zero because for any integer $i\in [0,N]$, $\expect[\text{samples}]{\ebid(i/N)} = \btild(i/N)$.

Next we note that the $\Tij$'s are bounded random variables. Specifically, let $Q$ be an interval of quantiles over which the difference $\ebid(\quant)-\btild(\quant)$ does not change sign. Then, following the proof of Lemma~\ref{first-term-bound}, we can bound 
\begin{align*}
|\T_Q| & = \left|\int_Q Z'_y(\quant)(\ebid(\quant)-\btild(\quant)) \, \text{d}q \right|\\
& \le 40\Delta \underbrace{\,\,\sup_{\quant}\{y'(\quant)\}\,\,\max \left\{ 1,\log\sup_{\quant: y'(\quant)\ge 1} \frac{\alloc'(\quant)}{y'(\quant)},
\log\sup_{\quant}\frac{y'(\quant)}{\alloc'(\quant)} \right\}}_{=:\, \rxy}.
\end{align*}
Likewise, over an interval $Q$ where $Z'_y$ does not change sign, we again get $|\T_Q|\le 40 \Delta \rxy$ with $\rxy$ defined as above. Moreover, for an interval $Q$ over which $Z'_y$ changes sign at most $t$ times, we have 
$$ \int_Q |Z'_y(\quant)(\ebid(\quant)-\btild(\quant))| \, \text{d}q \le t\cdot 40\Delta \rxy.$$ Finally, noting that $Z_y$ is a weighted sum over the $n$ functions $Z_k$ defined for the $k$-unit auctions, and that by Lemma~\ref{lem:Z-bound-1} each $Z_k$ has a unique maximum, we note that $Z'_y$ changes sign at most $2n$ times.

We now apply Chernoff-Hoeffding bounds to bound the probability that the sum $\sum_{\ell=0}^{\ell=\lastcrossing-1}\T_{i_{\ell}, i_{\ell+1}}$ exceeds some constant $a$. With $\tau_\ell$ denoting the upper bound on $|\T_{i_{\ell}, i_{\ell+1}}|$, this probability is at most 
$$\text{exp}\left(-\frac{a^2}{\sum_{\ell} \tau_\ell^2}\right).$$
By our observations above, for all $\ell$, $\tau_\ell\le 80\Delta \rxy$, and $\sum_\ell \tau_\ell \le \int_0^1 |Z'_y(\quant)(\ebid(\quant)-\btild(\quant))| \, \text{d}q \le 80n\Delta \rxy$. Therefore, $\sum_{\ell} \tau_\ell^2\le n(80\Delta \rxy)^2$. Since the bound does not depend on $I$, we can remove the conditioning
on $I$.
\end{proof}

\subsection{Comparing revenues}
\label{s:direct-comparison}

We have considered the case where the empirical task was to recover
the revenues for one mechanism ($y$) using the sample of bids
responding to another mechanism ($x$). In many practical situations
the empirical task is simply the verification of whether the revenue
from a given mechanism is higher than the revenue from another
mechanism. Or, equivalently, the task could be to verify whether one
mechanism provides revenue which is a certain percentage above that of
another mechanism. We now demonstrate that this is a much easier
empirical task in terms of accuracy than the task of inferring the
revenue.
\par
Suppose that we want to compare the revenues of mechanisms $B_1$ and
$B_2$ by mixing them in to an incumbent mechanism $A$, and running the
composite mechanism $C = \eps B_1 + \eps B_2 + (1-2\eps) A$.
Specifically, we would like to determine whether $\REV{B_1}>\alpha
\REV{B_2}$ for some $\alpha>0$. Consider a binary classifier
$\hat{\gamma}$ which is equal to $1$ when $\REV{B_1}>\alpha
\REV{B_2}$ and $0$ otherwise. Let $\gamma={\bf
  1}\{P_{B_1}-\alpha\,P_{B_2}>0\}$ be the corresponding ``ideal"
classifier for the case where the distribution of bids from mechanism
$C$ is known precisely. To evaluate the accuracy of the classifier, we
need to evaluate the probability $\prob{\hat{\gamma}=1|\gamma=0}$,
and likewise, $\prob{\hat{\gamma}=0|\gamma=1}$. The classifier will
give the wrong output if the sampling noise in estimating
$\hat{P}_{B_1}-\alpha\,\hat{P}_{B_2}$ is greater than
$|P_{B_1}-\alpha\,P_{B_2}|$.

Our main result of this section says that fixing the number of
positions $n$, $\alpha$, and the difference
$|P_{B_1}-\alpha\,P_{B_2}|$, with the number of samples from the bid
distribution, $N$, being large enough, the probability of incorrect
output decreases exponentially with $N$.

\begin{theorem}
\label{theorem:sign-AB}
For arbitrary $n$-agent rank-based auctions $A$, $B_1$, and $B_2$ and
$N$ bids from the equilibrium bid distribution of mechanism $C=\eps
B_1+\eps B_2+(1-2\eps)A$, the estimator for the binary classifier
$\gamma={\bf 1}\{P_{B_1}-\alpha\,P_{B_2}>0\}$, that establishes
whether the revenue of mechanism $B_1$ exceeds $\alpha$ times the
revenue of mechanism $B_2$, has error rate bounded by
$$\exp\left(-O\left(\frac{Na^2}{\alpha^2n^3\log(n/\eps)}\right)\right),$$
where $a=|P_{B_1}-\alpha\,P_{B_2}|$, as long as $N\gg n/\epsilon a$.  
\end{theorem}



We obtain a similar error bound when our goal is to estimate which of
$r$ different novel mechanisms obtains the most revenue, for any
$r>1$:

\begin{corollary}
  Suppose that our goal is to determine which of $r$ rank-based
  auctions, $B_1, B_2, \cdots, B_r$, obtains the most revenue while
  running incumbent mechanism $A$, by running each of the novel
  mechanisms with probability $\eps/r$. Then the error probability of
  the corresponding classifier constructed using $N$ bids from
  composite mechanism $C = \sum_{i=1}^{r} \eps/r B_i + (1-\eps)A$ is
  bounded from above by
$$r\exp\left(-O\left(\frac{Na^2}{n^3\log(rn/\eps)}\right)\right),$$
where $a$ is the absolute difference between the revenue obtained by
the best two of the $r$ mechanisms.
\end{corollary}

\section{Inference methodology and error bounds for first-price auctions}
\label{s:fp-inf}

In this section we define and analyze an estimator for counterfactual
revenue the bids in first-price auctions.  Our approach will be to
reduce this estimation problem to the all-pay estimation problem that
we solved previously.  Recall that the all-pay estimator is a weighted
order statistic of the empirical all-pay bid function.  Our
first-price estimator will map the empirical first-price bid function
to an empirical all-pay bid function and then apply to it the all-pay
estimator.

Recall that the Bayes-Nash equilibrium bid function of first-price
auction and all-pay auction are related by the payment identity.
Specifically an all-pay bid is deterministically equal to the expected
payment of the payment identity, while in a first-price auction an
agent only pays upon winning.  To facilitate comparison to previous
results we notate the equilibrium bid function of the all-pay auction
as $\bidap$ and the equilibrium bid function of the first-price
auction as $\bidfp$.  Given the allocation rule $\alloc$, the payment
identity requires $\bidap(\quant) = \alloc(\quant)\,\bidfp(\quant)$.
Consequently, an empirical all-pay bid function can be defined from
the empirical first-price bid function as $\ebidap(\quant) =
\alloc(\quant)\,\ebidfp(\quant)$.  Note that while in previous
sections the empirical all-pay bid function is piece-wise constant
(similarly the empirical first-price bid function is piece-wise
constant), this empirical all-pay bid function is not piece-wise
constant.

Partition the quantile range into extreme quantiles
$\Lambda=[0,\delta_N]\cup[1-\delta_N, 1]$ and the moderate quantiles
$[\delta_N,1-\delta_N]$.  Recall that truncation trades off a
(potentially diverging) variance of the estimator suggested by
\autoref{l:counterfactual-revenue} at the extreme quantiles with a
bias that can be bounded.  Specifically, truncation replaces bids at
low quantiles with zero and bids at high quantiles with the upper
bound $\ebidap(1)$ (which, in terms of the first-price bids, is
$\alloc(1)\,\ebidfp(1)$).  

As in \autoref{s:inference}, the estimator for counterfactual revenue
plugs the truncated empirical bid function into the counterfactual revenue
equation~\ref{eq:P_y-truncated} of
Lemma~\ref{l:counterfactual-revenue}.  We obtain the following
estimator in terms of the empirical first-price bids:
\begin{align*}
\hat{P}_y & =
  \expect[\quant\not\in\Lambda]{-Z'_y(\quant)\,\alloc(\quant)\,\ebidfp(\quant)}
  + Z_y(1-\delta_N)\,\alloc(1)\,\ebidfp(1).
\end{align*}
This estimator is a weighted order statistic as formalized in the
following definition.

\begin{definition}
\label{d:estimator-fp}
The estimator $\hat{P}_y$ (with truncation parameter $\delta_N$) for
the revenue of an auction with allocation rule $y$ from $N$ samples
$\ebidfp_1 \leq \cdots \leq \ebidfp_N$ from the equilibrium bid
distribution of a first price auction with allocation rule $x$ is:
\begin{align*}
\hat{P}_y 
& \def\INTERVAL{[i,i+1]/N}
= \sum\nolimits_{i=\delta_N N}^{N-\delta_N N} \expect[\quant \in \INTERVAL]{-Z'_y(\quant)\,\alloc(\quant)}\,\ebidfp_i + Z_y(q)\,\alloc(1)\,\ebidfp_N.
\end{align*}
\end{definition}

To obtain a bound on the mean absolute error of the estimator
judiciously plug the identity relating first-price and all-pay
equilibrium bids into the error bound of
equation~\eqref{eq:error-allpay} to get:
\begin{align}
\label{eq:error-first-price}
|\emurevk-\murevk| \le & 
\left|\expect[q\not\in\Lambda]{\smash{-Z'_k(\quant)\,\alloc(\quant)\,(\ebidfp(\quant)-\bidfp(\quant))}}\right|
+ \left| \expect[\quant\in\Lambda]{Z_k(\quant)\,\bidap'(\quant)}\right|\\
\notag & + \left| \smash{Z_k(1-\delta_N)\,(\bidap(1-\delta_N) -\alloc(1)\,\ebidfp_N)}\right|
+ \left| Z_k(\delta_N)\,\bidap(\delta_N)\right|
\end{align}
It is clear that terms that depend only on the equilibrium bid
functions and not the empirical bid functions need no further
analysis.  Specifically \autoref{second-term-bound} and
\autoref{fourth-term-bound} bound the contribution to the error of the
second and fourth terms of equation~\eqref{eq:error-first-price}.  It
remains to bound the contribution from the first and third terms.
These bounds come from relatively minor adjustments to the analogous
bounds for all-pay auctions.

For the third term, we can adapt the analysis of
\autoref{third-term-bound}.  Denote the quantile of bid $\ebidfp_N$ by
$\equant$, i.e., $\ebidfp_N = \bidfp(\equant)$.  There are two parts
of the analysis, the first part is for the case $\equant \geq
1-\delta_N$ and the second part is for the case $\equant \leq
1-\delta_N$.  

For the first part, the proof of \autoref{third-term-bound} upper
bounds $\ebidap_N$ by $\bidap(1)$.  We can do the same for
$\alloc(1)\,\ebidfp_N$: $\alloc(1)\,\bidfp(\equant) \leq \alloc(1)
\bidfp(1) = \bidap(1)$.  The first inequality follows from the
monotonicity of the equilibrium bid function, i.e., $\ebidfp(\equant)
\leq \ebidfp(1)$ for $\equant \leq 1$.  Thus, we can upper bound the
error in the case that $\equant \geq 1-\delta_N$ by
$\smash{Z_k(1-\delta_N)\,(\bidap(1) - \bidap(1-\delta_N))}$ which was
bounded already in the proof of \autoref{third-term-bound}.

For the second part, write
\begin{align*}
\alloc(1)\,\bidfp(\equant)
&= \alloc(1)\,\bidap(\equant) / \alloc(\equant)\\
&\geq \alloc(\equant)\,\bidap(\equant) / \alloc(\equant)\\
&= \bidap(\equant),
\end{align*}
where inequality follows from monotonicity of $\alloc$.  Thus, we can
upper bound the error in the case that $\equant \leq 1-\delta_N$ by
$\smash{Z_k(1-\delta_N)\,(\bidap(1-\delta_N) - \bidap(\equant))}$ which was
bounded already by \autoref{third-term-bound}.

To analyze the first term in the error bound of
equation~\eqref{eq:error-first-price}, we begin with the following upper
bound:
\begin{align*}
\left|\expect[q\not\in\Lambda]{\smash{-Z'_y(\quant)\alloc(\quant)(\ebidfp(\quant)-\bidfp(\quant))}}\right|
& \le \expect[q\not\in\Lambda]{\left|\frac{Z'_y(q)}{h(Z_y(q))}\right|}
\sup\limits_q\left|\alloc(\quant)h(Z_y(q))(\ebidfp(q)-\bidfp(q))\right|\\
& \le \expect[q\not\in\Lambda]{\left|\frac{Z'_y(q)}{h(Z_y(q))}\right|}
\sup\limits_q\left|h(Z_y(q))(\ebidfp(q)-\bidfp(q))\right|.
\end{align*}
We can then carry out an analysis identical to the proof of
\autoref{first-term-bound} with an appropriate choice of
$h(\cdot)$. The only difference is in the application of
\autoref{error bid function}. Whereas for all-pay auctions the lemma
bounds the weighted error in bids in terms of $\sup_q
\{q(1-q)x'(q)\}\le n/4$, in the case of first-price auctions, this
term is replaced by $\sup_q \{q(1-q)x'(q)/x(q)\}$, which is no more
than $n$ for rank-based allocation rules. We obtain the following
theorem.
\begin{theorem}\label{th: first price}
  The expected absolute error in estimating the revenue of a position
  auction with allocation rule $y$ using $N$ samples from the bid
  distribution for a first-pay position auction with allocation rule
  $x$ is bounded by both of the expressions below; Here $n$ is the number of positions
  in the two position auctions.
\begin{align*}
\Err{P_y} & \le \frac{28n^2\log N}{\sqrt{N}},\\
\Err{P_y} & \le \frac{80}{\sqrt N}\,n \log \sup\nolimits_{\quant}
n \frac{y'(\quant)}{\alloc'(\quant)}.
\end{align*}
When $y$ is the highest-$k$-bids-win allocation rule, the latter bound
improves to:
\begin{align*}
\Err{P_k}
& \le \frac{80}{\sqrt{N}} \rxy[\kalloc]
\end{align*}
with $\rxy[\kalloc]$ as defined in equation~\eqref{eq:rxy}.
\end{theorem}

Because the error bounds in Theorem~\ref{th: first price} are
identical up to constant factors to those in
Theorems~\ref{thm:allpay-simple}, \ref{thm:allpay-general} and
Corollary~\ref{cor:allpay-y}, other results in Lemma~\ref{lem:univ},
Corollaries~\ref{cor1}, \ref{cor2}, \ref{cor3}, \ref{cor:universal},
and Theorems~\ref{theorem:sign-AB} and \ref{thm:sw} continue to hold
when bids are drawn from a first-price auction.

\section{Finding the optimal iron by rank auction}
\label{s:iron-opt-app}

Recall that iron by rank auctions are weighted sums of multi-unit
auctions. Therefore, their revenue can be expressed as a weighted sum
over the revenues $P_k$ of $k$-unit auctions. We consider a position environment given by non-increasing weights $\wals
 = (\walk[1],\ldots,\walk[n]$), with $\walk[0] = 0$, $\walk[1]=1$, and $\walk[n+1] = 0$.  Define the cumulative position
weights $\cumwals = (\cumwalk[1],\ldots,\cumwalk[n])$ as $\cumwalk =
\sum_{j \leq i} \walk[j]$.

Define the {\em multi-unit revenue
  curve} as the piece-wise constant function connecting the points
$(0,\murevk[0],\ldots,(n,\murevk[n])$.  This function may or may not
be concave.  Define the {\em ironed multi-unit revenue curve} as
$\imurevs = (\imurevk[1],\ldots,\imurevk[n])$ the smallest concave
function that upper bounds the multi-unit revenue curve.  Define the
multi-unit marginal revenues as $\mumargs =
\mumargk[1],\ldots,\mumargk[n]$ and $\imumargs =
\imumargk[1],\ldots,\imumargk[n]$ as the left slope of the multi-unit
and ironed multi-unit revenue curves, respectively.  I.e., $\mumargk = \murevk - \murevk[k-1]$ and $\imumargk = \imurevk - \imurevk[k-1]$.

We now see how the revenue of any position auction can be expressed in
terms of the multi-unit revenue curves and marginal revenues.
\begin{align*}
\expect{\text{revenue}} &= \sum_{k=0}^n \murevk\,\margwalk 
                         = \sum_{k=0}^n \mumargk\,\walk\\
                        &\leq \sum_{k=0}^n \imurevk\,\margwalk 
                         = \sum_{k=0}^n \imumargk\,\walk.
\end{align*}
The first equality follows from viewing the position auction with
weights $\wals$ as a convex combination of multi-unit auctions (where
its revenue is the convex combination of the multi-unit auction
revenues).  The second and final inequality follow from rearranging
the sum (an equivalent manipulation to integration by parts).  The
inequality follows from the fact that $\imurevs$ is defined as the
smallest concave function that upper bounds $\murevs$ and, therefore,
satisfies $\imurevk \geq \murevk$ for all $k$.  Of course the
inequality is an equality if and only if $\margwalk = 0$ for every
$k$ such that $\imumargk > \mumargk$.

We now characterize the optimal ironing-by-rank position auction.
Given a position auction weights $\wals$ we would like the
ironing-by-rank which produces $\iwals$ (with cumulative weights
satisfying $\cumwals \geq \cumiwals$) with optimal revenue.  By the
above discussion, revenue is accounted for by marginal revenues, and
upper bounded by ironed marginal revenues.  If we optimize for ironed
marginal revenues and the condition for equality holds then this is
the optimal revenue.  Notice that ironed revenues are concave in $k$,
so ironed marginal revenues are monotone (weakly) decreasing in $k$.
The position weights are also monotone (weakly) decreasing.  The
assignment between ranks and positions that optimizes ironed marginal
revenue is greedy with positions corresponding to ranks with negative
ironed marginal revenue discarded.  Tentatively assign the $k$th rank
agent to slot $k$ (discarding agents that correspond to discarded
positions).  This assignment indeed maximizes ironed marginal revenue
for the given position weights but may not satisfy the condition for
equality of revenue with ironed marginal revenue.  To meet this
condition with equality we can randomly permute (a.k.a., iron by rank)
the positions that corresponds to intervals where the revenue curve is
ironed.  This does not change the surplus of ironed marginal revenue
as the ironed marginal revenues on this interval are the same, and the
resulting position weights $\iwals$ satisfy the condition for equality
of revenue and ironed marginal revenue.

\section{Constructing a position auction with a target vector of position weights}
\label{a:position-auction-construction}

In this section we show that in any position auction environment given
by position weights $\wals$, we can construct a rank based auction
with induced position weights $\iwals$ satisfying $\cumiwals \leq
\cumwals$. The allocation rule of the auction is constructed as a
(random) sequence of iron by rank and rank reserve operations.

\begin{lemma}
In any position auction environment with position weights $\wals$ and
target weights $\iwals$ satisfying $\cumiwals \leq \cumwals$, there
exists a (random) sequence of iron by rank and rank reserve operations
following which the induced position weights are exactly $\iwals$.
\end{lemma}

\begin{proof}
  Suppose that we have $\cumiwals \leq \cumwals$. We will describe how
  to assign agents to slots so as to obtain position weights
  $\iwals$. Let $\yals$ denote the position weights corresponding
  to an intermediate assignment. We begin by assigning the agent with
  the $i$th largest bid to the $i$th slot for all $i\in [n]$. The
  position weights for this assignment are given by
  $\yals=\wals$. We will then construct a series of transformations
  or reassignments of agents to positions, each time making a small
  change to the weights $\yals$, so as to bring them closer to the
  target $\iwals$. Each transformation is either an rank-based ironing
  operation or a rank reserve.

  Let $i$ denote the largest index such that $\cumiwalk=\cumyalk$
  for all $k\le i$. We will now present a (randomized) tranformation
  that will increase the value of $i$. Specifically, we will reassign
  some agents to positions in a manner such that the resulting
  position weights $\newyals$ satisfy: $\cumiwalk=\cumnewyalk$ for
  $k\le i+1$ and $\cumyalk\ge \cumnewyalk\ge \cumiwalk$ for
  $k>i+1$.

  Consider the operation of ironing by rank over the interval
  $\{i,\ldots,i'\}$ for some index $i'>i$. Recall that this operation
  averages out the position weights over this interval, setting each
  such weight equal to $(\cumyalk[i']-\cumyalk[i])/(i'-i)$,
  while leaving all other position weights intact. It also preserves
  cumulative weights at positions $k\le i$ and positions $k\ge
  i'$. Note also that the larger that $i'$ is, the smaller is the
  average weight $(\cumyalk[i']-\cumyalk[i])/(i'-i)$. In
  particular, for any $i'>i+1$, this operation strictly decreases the
  $i+1$th position weight.

  Suppose that there exists an index $i'$, with $i+1<i'< n$, such that
  $(\cumyalk[i']-\cumyalk[i])/(i'-i) \ge \iwalk[i+1]$ and
  $(\cumyalk[i'+1]-\cumyalk[i])/(i'+1-i) < \iwalk[i+1]$. Let
  $A:= (\cumyalk[i']-\cumyalk[i])/(i'-i)$ and
  $B:=(\cumyalk[i'+1]-\cumyalk[i])/(i'+1-i)$. Let $\alpha\in
  (0,1]$ be defined such that $\alpha A + (1-\alpha) B =
  \iwalk[i+1]$. Now consider the following transformation. With
  probability $\alpha$, we iron over the rank interval
  $\{i,\ldots,i'\}$ and with probability $1-\alpha$, we iron over the
  rank interval $\{i,\ldots,i'+1\}$. Let $\newyals$ and
  $\cumnewyals$ denote the new positions weights and cumulative
  position weights at the end of the (randomized) ironing
  operation. Note that both of these ironing operations preserve the
  position weights over positions $k\le i$ and $k>i'+1$. Over
  positions $k\in \{i,\ldots,i'\}$, the new position weight
  $\newyalk$ is exactly $\alpha A + (1-\alpha)B =
  \iwalk[i+1]$. Finally, both ironing operations maintain the same
  cumulative weight at position $i'+1$. Since $\cumnewyalk[i'] =
  \cumyalk[i']= \cumiwalk[i']$ and $\cumnewyalk[i'+1] =
  \cumyalk[i'+1]>\cumiwalk[i'+1]$, we get that the new position
  weight at $i'+1$ is at least $\iwalk[i'+1]$. This completes one step
  of our transformation.

  Alternately, suppose that for $i'=n$, we have
  $(\cumyalk[i']-\cumyalk[i])/(i'-i) \ge \iwalk[i+1]$. Let $A:=
  (\cumyalk[i']-\cumyalk[i])/(i'-i)$ and let $\alpha\in [0,1]$ be
  defined such that $\alpha A = \iwalk[i+1]$. Now consider the
  following transformation. With probability $\alpha$, we iron over
  the rank interval $\{i,\ldots,n\}$ and with probability $1-\alpha$,
  we set a rank reserve of $i$, that is, we reject every agent with
  rank $>i$. Note that both of these operations preserve the position
  weights over positions $k\le i$. For $k>i$, the new position weights
  are exactly $\alpha A = \iwalk[i+1]$. Therefore, once again, we
  obtain $\cumiwalk=\cumnewyalk$ for $k\le i+1$ and $\cumyalk\ge
  \cumnewyalk\ge \cumiwalk$ for $k>i+1$.

  To summarize, we described a sequence of randomized operations. Each
  step of the sequence increases the number of positions over which
  the position weights corresponding to our current assignment,
  $\yals$, match the target position weights, $\iwals$. After at
  most $n$ such operations we obtain a randomized assignment of agents
  to positions achieving the target position weights.
\end{proof}

\section{Approximation via rank-based auctions}
\label{s:rank-approx}

In this section we show that the revenue of optimal rank-based auction
approximates the optimal revenue (over all auctions) for position
environments.  Instead of making this comparison directly we will
instead identify a simple non-optimal rank-based auction that
approximates the optimal auction.  Of course the optimal rank-based
auction of Theorem~\ref{thm:rank-based-opt} has revenue at least that
of this simple rank-based auction, thus its revenue also satisfies the
same approximation bound.

Our approach is as follows.  Just as arbitrary rank-based mechanisms
can be written as convex combinations over $k$-highest-bids-win
auctions, the optimal auction can be written as a convex combination
over optimal $k$-unit auctions.  We begin by showing that the revenue
of optimal $k$-unit auctions can be approximated by multi-unit
highest-bids-win auctions when the agents' values are distributed
according to a regular distribution (Lemma~\ref{lem:approx-regular},
below). In the irregular case, on the other hand, rank-based auctions
cannot compete against arbitrary optimal auctions. For example, if the
agents' value distribution contains a very high value with probability
$o(1/n)$, then an optimal auction may exploit that high value by
setting a reserve price equal to that value; on the other hand, a
rank-based mechanism cannot distinguish very well between values
correspond to quantiles above $1-1/n$. We show that rank-based
mechanisms can approximate the revenue of any mechanism that does not
iron or reserve price within the quantile interval $[1-1/n,1]$ (but
may arbitrarily optimize over the remaining
quantiles). Theorem~\ref{thm:rank-based-approx} presents the precise
statement.




\begin{lemma}
\label{lem:approx-regular}
For regular $k$-unit $n$-agent environments, there exists a $k' \leq
$k such that the highest-bid-wins auction that restricts supply to
$k'$ units (i.e., a rank reserve) obtains at least half the revenue of
the optimal auction.
\end{lemma}
\begin{proof}
This lemma follows easily from a result of \citet{BK-96} that states
that for agents with values drawn i.i.d.\@ from a regular distribution
the revenue of the $k'$-unit $n$-agent highest-bid-wins auction is at
least the revenue of the $k'$-unit $(n-k')$-agent optimal auction. To
apply this theorem to our setting, let us use $\opt{k}{n}$ to denote
the revenue of an optimal $k$-unit $n$-agent auction, and recall that
$nP_k$ is the revenue of a $k$-unit $n$-agent highest-bids-win
auction.

When $k\le n/2$, we pick $k'=k$. Then,  
$$nP_k\ge \opt{k}{n-k} \ge \frac{(n-k)}{n} \opt{k}{n}\ge \frac
12\opt{k}{n},$$ and we obtain the lemma. Here the first inequality
follows from \citeauthor{BK-96}'s theorem and the third from the
assumption that $k\le n/2$.  The second inequality follows via by
lower bounding $\opt{k}{n-k}$ by the following auction which has
revenue exactly $\frac{(n-k)}{n} \opt{k}{n}$: simulate the optimal
$k$-unit $n$-agent on the $n-k$ real agents and $k$ fake agents with
values drawn independently from the distribution.  Winners of the
simulation that are real agents contribute to revenue and the
probability that an agent is real is $(n-k)/n$.

When $k>n/2$, we pick $k'=n/2$. As before we have: 
\begin{align*}
nP_{n/2}\ge \opt{n/2}{n/2} &= \frac 12\opt{n}{n}\ge \frac 12\opt{k}{n}.\qedhere
\end{align*}
\end{proof}

\begin{lemma}
\label{lem:approx-irregular} 
For (possibly irregular) $n$-agent environments with revenue
curve $\rev(\cdot)$ and quantile $q\le 1-1/n$, there exists an integer
$k\le (1-q)n$ such that the revenue of the $k$-highest-bids-win
auction is at least a quarter of $n \rev(q)$, the revenue from posting
a price of $\val(q)$.
\end{lemma}
\begin{proof}
First we get a lower bound on $\murevk$ for any $k$.  For any value
$z$, the total expected revenue of the $k$-highest-bids-win auction is
at least $zk$ times the probability that at least $k+1$ agents have
value at least $z$.  The median of a binomial random variable
corresponding to $n$ Bernoulli trials with success probability
$(k+1)/n$ is $k+1$.  Thus, the probability that this binomial is at
least $k+1$ is at least $1/2$.  Combining these observations by
choosing $z = \val(1-(k+1)/n)$ we have,
\begin{align*}
n \, \murevk &\geq  \val(1-(k+1)/n)\, k / 2.\\
\intertext{Choosing $k = \lfloor (1-q)n\rfloor -1$, for which $\val(1-(k+1)/n) \geq \val(q)$, the bound simplifies to,}
n \, \murevk &\geq  \val(q)\, k / 2.
\end{align*}
The ratio of $P_k$ and $\rev(q) = (1-q)\,\val(q)$ is therefore at least 
\begin{align*}
\frac{k}{2(1-q)n} &> \frac{k}{2(k+2)}.
\end{align*}
For $q\le 1-3/n$ (or, $k\ge 2$) this ratio is at least $1/4$.

For $q\in (1-3/n,1-1/n]$, we pick $k=1$. Then, $\murevk[1]$ is at least $1/n$
  times $\val(q)$ times the probability that at least two agents have a
  value greater than or equal to $\val(q)$. We can verify for $n\ge 2$ that
$$\murevk[1] \ge \frac {\val(q)}n  \left( 1-q^n-n(1-q)q^{n-1} \right)\ge \frac 14 (1-q)\,\val(q).$$
\end{proof}

\begin{theorem}
\label{thm:rank-based-approx}
For regular value distributions and position environments, the optimal
rank-based auction obtains at least half the revenue of the optimal
auction.  For any value distribution (possibly irregular) and position
environments, the optimal rank-based auction obtains at least a
quarter of the revenue of the optimal auction that does not iron or
set a reserve price for the highest $1/n$ measure of values i.e.,
$q \in [1-1/n,1]$.
\end{theorem}

\begin{proof}
In the regular setting, the theorem follows from
Lemma~\ref{lem:approx-regular} by noting that the optimal auction
(that irons by value and uses a value reserve) in a position
environment is a convex combination of optimal $k$-unit auctions:
since the revenue of each of the latter can be approximated by that of
a $k'$-unit highest-bids-win auction with $k'\le k$, the revenue of
the convex combination can be approximated by that of the same convex
combination over $k'$-unit highest-bids-win auctions; the resulting
convex combination over $k'$-unit auctions satisfies the same position
constraint as the optimal auction.

In the irregular setting, once again, any auction in a position
environment is a convex combination of optimal $k$-unit auctions. The
expected revenue of any $k$-unit auction is bounded from above by the
expected revenue of the optimal auction that sells at most $k$ items
in expectation. The per-agent revenue of such an auction is bounded by
$\irev(1-k/n)$, the revenue of the optimal allocation rule with ex
ante probability of sale $k/n$. Here $\irev(\cdot)$ is the ironed
revenue curve (that does not iron on quantiles in
$[1-1/n,1]$). $\irev(1-k/n)$ is the convex combination of at most two
points on the revenue curve $\rev(a)$ and $\rev(b)$, $a\le 1-k/n\le b
< 1-1/n$.  Now, we can use Lemma~\ref{lem:approx-irregular} to obtain
an integer $k_a < n(1-a)$ such that $P_{k_a}$ is at least a quarter of
$\rev(a)$, likewise $k_b$ for $b$. Taking the appropriate convex
combination of these multi-unit auctions gives us a $4$-approximation
to the optimal auction $k$-unit auction (that does not iron over the
quantile interval $[1-1/n,1]$).  Finally, the convex combination of
the multi-unit auctions with $k_a$ and $k_b$ corresponds to a position
auction with that is feasible for a $k$ unit auction (with respect to
serving the top $k$ positions with probability one, service
probability is only shifted to lower positions).
\end{proof}

\end{document}